\documentclass[11pt]{article}
\usepackage{amsfonts,amsmath,amssymb,amsthm,graphicx,url}
\usepackage{color}
\usepackage{algorithm}
\usepackage{algorithmic}
\usepackage{relsize}
\usepackage[margin=1in]{geometry}
\usepackage[T1]{fontenc}

\DeclareMathOperator*{\argmax}{arg\,max}

\usepackage{multirow}

\newtheorem{theorem}{Theorem}
\newtheorem{lemma}{Lemma}
\newtheorem{claim}{Claim}
\newtheorem{corollary}{Corollary}

\newcommand{\bE}{\mathbb{E}}

\newcommand{\bR}{\mathbb{R}}

\newcommand{\bU}{\mathbb{U}}

\newcommand{\cD}{\mathcal{D}}

\newcommand{\cI}{\mathcal{I}}

\newcommand{\cM}{\mathcal{M}}

\newcommand{\cP}{\mathcal{P}}

\title{Information Elicitation for Bayesian Auctions \thanks{The first author thanks Matt Weinberg for reading a draft of this paper and for helpful discussions. The authors thank Constantinos Daskalakis, J\'{a}nos Flesch, Hu Fu, Pinyan Lu, Silvio Micali, Rafael Pass, Andr\'{e}s Perea, Elias Tsakas, several anonymous reviewers, and the participants of seminars at Stony Brook University, Shanghai Jiaotong University, Shanghai University of Finance and Economics, Maastricht University, MIT, and IBM Thomas J. Watson Research Center for helpful comments.
This work is partially supported by NSF CAREER Award No. 1553385.}}

\author{Jing Chen \hspace{30pt} Bo Li \hspace{30pt} Yingkai Li  \hspace{30pt} \\
Department of Computer Science, Stony Brook University\\
Stony Brook, NY 11794, USA\\
\texttt{\{jingchen, boli2, yingkli\}@cs.stonybrook.edu}}

\date{}

\begin{document}

\maketitle

\begin{abstract}
In this paper we design information elicitation mechanisms for Bayesian auctions.
While in Bayesian mechanism design the distributions
of the players' private types are often assumed to be common knowledge,
information elicitation considers the situation
where the players know the distributions better than the decision maker.
To weaken the information assumption in Bayesian auctions,
we consider an information structure where
the knowledge about the distributions is
{\em arbitrarily scattered} among the players.
%and use information elicitation approach to design mechanisms for Bayesian auctions.
In such an unstructured information setting, we
design mechanisms for unit-demand auctions and additive auctions
that {\em aggregate} the players' knowledge,
generating revenue that are constant approximations to
the optimal Bayesian mechanisms with a common prior.
Our mechanisms are 2-step dominant-strategy truthful and
the revenue increases gracefully with the amount of knowledge the players collectively have.


\end{abstract}

\thispagestyle{empty}
\newpage

%\tableofcontents
%\newpage

\setcounter{page}{1}

\section{Introduction}\label{sec:intro}
%In this paper we study the removal of the {\em common prior assumption} from Bayesian mechanisms.

Bayesian auction design has been extremely flourishing since the seminal work of \cite{myerson1981optimal}.
% and \cite{cremer1988full}.
One of the main focuses is to generate revenue, by selling $m$ heterogenous items to $n$ players.
Each player has a private valuation function describing how much he values each subset of the items,
and the valuations are drawn from prior distributions.
An important assumption in Bayesian mechanism design is that the distributions are commonly known by the seller and the players ---the {\em common prior} assumption.
However, as pointed out by another seminal work \cite{wilson1985game},
such common knowledge is ``rarely present in experiments and never in practice'',
and ``only by repeated weakening of common knowledge
assumptions will the theory approximate reality.''

In this paper, we weaken the information assumption about the seller and the players by adopting an information elicitation approach \cite{miller2005eliciting}.
We consider a framework for auctions
where the knowledge about the players' value distributions are
{\em arbitrarily scattered} among the players and the seller.
The seller must {\em aggregate} pieces of
information from all players to gain a good understanding about the distributions,
so as to decide how to sell the items.

As in information elicitation, the players get rewards for reporting their knowledge.
However, different from classic information elicitation
where a player's utility is exactly his reward,
in our model a player's utility comes not only from his knowledge,
but also from participating in the auction (i.e., from buying items). % and competing for items.
Moreover,
information elicitation usually assumes the prior distribution is correlated:
each player observes a private signal and reports the corresponding posterior distribution.
This means every player has information about every other player.
In our model, following the convention in multi-item auctions, the players' value distributions for individual items
are assumed to be independent. A player may be totally ignorant about some players and only partially knows some other players' distributions.
%In Bayesian auctions, however, players' valuation distributions are often assumed to be independent.

%However,
%it is still a demanding requirement that some agent,
%whether the seller or a player,
%individually possesses good knowledge about {\em all} players' value distributions.
%As our main conceptual contribution, in this paper we consider a framework
%for auctions where knowledge about the players' value distributions are
%{\em arbitrarily scattered} among the players and the seller.
%The seller can no longer base his mechanism
% on a single agent's knowledge, and must really {\em aggregate} pieces of
% information from all players to gain a good understanding about the distributions.

We focus on unit-demand auctions and additive auctions ---two valuation types widely studied
in the literature \cite{chawla2007algorithmic,hart2012approximate}.
In such auctions, a player's valuation function is specified by $m$ values, one for each item.
For each player~$i$ and item $j$, the value  $v_{ij}$ is independently drawn from a distribution~ $\cD_{ij}$.
Each player privately knows his own values and some (or none) of the distributions of some other players for some items,
like long-time competitors in the market.
There is no constraint about who knows which distributions.
%It is possible that nobody knows how the whole valuation profile is distributed,
% and some value distributions are not known by anybody.
The seller may also know some of the distributions, but he does not know which player knows what.
A player may or may not know his own value distributions. However, it is hard to elicit a player's knowledge about his own distribution, and we are not aware of any such study in information elicitation. Thus we do not consider the players' self-knowledge.

We introduce directed {\em knowledge graphs} to succinctly describe
the players' knowledge. Each player knows the distributions of his neighbors,
different items' knowledge graphs may be totally different, and the structures of the graphs
are {\em not known} by anybody.
Interestingly, the intuition behind such an information structure has long been considered by philosophers.
In \cite{james1979some},
the author discussed a world where ``everything in the world might be known by somebody, yet not everything by the same knower.''
%Under such an unstructured information setting, we
%are able to design mechanisms that aggregate all players' (and the seller's)
%knowledge and generate good revenue compared with the optimal Bayesian revenue when there is a common prior.
%As our mechanisms literally ``crowdsource'' information from every player to
%fill in the jigsaw puzzle of the value distributions,
%we continue referring to them as  {\em crowdsourced Bayesian mechanisms.}
%However, compared with the original notion put forward in \cite{azar2012crowdsourced},
%this notion now has a much broader and richer meaning.
%We formalize our model in Section~\ref{sec:model}.
Below we briefly state our main results.

\subsection{Main Results}\label{subsec:results}

\paragraph{Under arbitrary knowledge graphs.}
Our goal is to design {\em 2-step dominant strategy truthful} (2-DST)
information elicitation mechanisms whose expected revenue approximates
that of the optimal Bayesian incentive compatible (BIC) mechanism, denoted by~$OPT$.%
%We may also compare the revenue of our mechanisms with that of the optimal DST Bayesian mechanism, denoted by $OPT_D$.%
\footnote{A Bayesian mechanism is BIC if it is a Bayesian Nash equilibrium for all players to report their true values.}
%, and
%DST if it is dominant for each player to report his true values.}
In order for the seller to aggregate the players' knowledge about the distributions,
it is natural for the mechanism to ask each player to report his knowledge to the seller, together with his own values.
A 2-DST mechanism  \cite{azar2012crowdsourced} is such that,
{\em (1) no matter what knowledge the players may report about each other, it is {\em dominant} for each player to report his true values; and (2) given that all players report their true values,
it is {\em dominant} for each player to report his true knowledge about others.}

When the knowledge graphs are such that
some distributions are not known by anybody,
it is easy to see that no information elicitation mechanism
can be a bounded approximation to $OPT$.
Thus it is natural to consider the following benchmark:
the optimal BIC mechanism
applied to players and items for whom the distributions are indeed known by somebody, denoted by $OPT_K$.
This is a natural benchmark when considering players with limited knowledge
and, if every distribution is known by somebody, then it
is exactly $OPT$.
We have the following, formalized in Section~\ref{sec:k=0}.

\vspace{3pt}
\noindent
{\bf Theorems \ref{thm:unit} and \ref{thm:newbvcg}.} (sketched) {\em
For any knowledge graph,
there is a 2-DST information elicitation mechanism for unit-demand auctions
%$\cM'_{CSUD}$ that is 2-DST and
with revenue
$\geq \frac{OPT_K}{96}$,
and such a mechanism for
additive auctions
% instances
%there exists a crowdsourced mechanism $\cM'_{CSA}$ that is 2-DST and achieves
with revenue
$\geq \frac{OPT_K}{70}$.
%where $\cI'$ is the benchmark instance.
}

\smallskip
To prove Theorem~\ref{thm:unit},
we actually show a general result:
any Bayesian mechanism for unit-demand auctions that is a good approximation in the COPIES setting
(formally defined in Section \ref{subsec:unit})
can be converted to information elicitation mechanisms; see Theorem \ref{col:copies}.
This applies to a large class of Bayesian mechanisms,
including the ones in \cite{chawla2007algorithmic,kleinberg2012matroid,chawla2010multi}.
%
%%\vspace{5pt}
%%\noindent
%%{\bf Theorem \ref{col:copies}.} (informally stated) {\em Any Bayesian mechanism for unit-demand auctions that is a good approximation in the COPIES setting can be crowdsourced.}
%
%\smallskip

To prove Theorem~\ref{thm:newbvcg},
we have developed a novel approach for using the {\em adjusted revenue}~\cite{yao2015n}.
Although this concept is very useful in Bayesian auctions,
it was unexpected that
we found an interesting and highly non-trivial way of using it to
analyze information elicitation mechanisms.
%(See Lemma 6 in the full version \cite{chen2017bayesian}).

%and proved a useful technical lemma:
%the optimal adjusted revenue can be crowdsourced
%(See Lemma 6 in the full version \cite{chen2017bayesian}).

\vspace{-8pt}
\paragraph{When everything is known by somebody.}
When the knowledge graphs become denser, the amount of knowledge increases and the seller may
generate more revenue.
Indeed, if every distribution is known by somebody, $OPT_K=OPT$.
We show the revenue that can be generated by information elicitation mechanisms {\em increases} gracefully
together with the amount of knowledge.
More precisely, for any integer $k\geq 1$, let $\tau_k = \frac{k}{(k+1)^{\frac{k+1}{k}}}$.
Note $\tau_1 = \frac{1}{4}$ and $\tau_k \rightarrow 1$
when $k$ gets larger.
%Also, $k$ can be much smaller than $n$ for $\tau_k$ to be close to 1.
We have the following theorems, formalized in
Section~\ref{sec:partial}.

\smallskip
\noindent
{\bf Theorems \ref{thm:unit-k} and \ref{thm:additive}.} (sketched) {\em
$\forall k\in [n-1]$, when each distribution is known by at least $k$ players,
there is a 2-DST information elicitation mechanism for unit-demand auctions
with revenue $\geq \frac{\tau_k}{24}\cdot OPT$,
and such a mechanism for additive auctions
with revenue $\geq \max\{\frac{1}{11}, \frac{\tau_k}{6+2\tau_k}\} OPT$. %, where $\tau_k = \frac{k}{(k+1)^{\frac{k+1}{k}}}$.
}
\smallskip

%Note that $\tau_1 = \frac{1}{4}$ and $\tau_k \rightarrow 1$
%when $k$ gets larger. Also, $k$ can be much smaller than $n$ for $\tau_k$ to be close to 1.
%Furthermore,
%both the common prior assumption and the crowdsourced setting considered in \cite{azar2012crowdsourced} imply $k=n-1$
%---``everything is known by everybody''.
%%would be the information setting considered in \cite.
%
Finally, by exploring the knowledge graph's combinatorial structure,
%for single-item auctions, when the knowledge graph
%shows more combinatorial structure, the revenue can be
we have the following for single-good auctions.

\smallskip
\noindent
{\bf Theorem \ref{thm:myerson}.} (sketched) {\em When the knowledge graph is 2-connected,%
\footnote{A directed graph is 2-connected if for any node $i$, the graph with $i$ and all adjacent edges removed is still strongly connected.
%For knowledge graphs, this means no player is an ``information hub'', without whom the players will split into two parts such that one part has no information about the other.
}
there is a 2-DST information elicitation mechanism for single-good auctions with revenue $\geq (1-\frac{1}{n})OPT$.
}

\vspace{-4pt}

%%%%%%%%%%%%%%%%%%%%%%%%%%%%%

\subsection{Discussions}
\vspace{-2pt}
\paragraph{The use of scoring rules.}
Since our mechanisms elicit the players' knowledge about each other's value distributions,
we will use {\em scoring rules} (see, e.g., \cite{brier1950verification}) to reward the players for their reported knowledge, as typical
in information elicitation. % \cite{miller2005eliciting,radanovic2013robust}.
However, the use of scoring rules does not solve the main problems in our auctions.
%, because they do not help solving the main difficulties in designing crowdsourced mechanisms.
Indeed,
because a player's utility comes both from the reward and from participating in the auction,
the difficulties
in designing information elicitation mechanisms are to guarantee that,
even without rewarding the players for their knowledge,
(1) it is dominant for each player to report his true values,
(2) reporting his true knowledge {\em never hurts him}, and
(3) the resulting revenue approximates the desired benchmark.
%We call mechanisms satisfying properties (1) and (2) {\em weakly 2-DST} mechanisms.
%Following the widely adopted convention that a player tells the truth
%as long as lying does not strictly benefit him, we are done when all three properties hold.

Accordingly, in Sections \ref{sec:k=0} and \ref{sec:partial}
we focus on designing information elicitation mechanisms without rewarding the players. % for their knowledge.
Scoring rules are used later solely to break the utility-ties and make it {\em strictly better} for a player to
report his true knowledge.
%, and they are used in a
%straightforward manner.
In Appendix \ref{sec:buyknowledge},
%\ref{sec:buyknowledge},
we show how to add scoring rules to our mechanisms.

\vspace{-10pt}

\paragraph{Extensions of our results.}
In our main results, the seller asks the players to report
the distributions in their entirety, without being concerned with the communication complexity for doing so.
This is common in information elicitation and allows us to focus on the main difficulties in aggregating
the players' knowledge.
In Appendix \ref{app:efficient},
we show how to modify our mechanisms
so that the players only report a small amount of information about the distributions.

Furthermore, in the main body of this paper we consider auction settings where a player
$i$'s knowledge about another player $i'$ for an item $j$ is exactly the prior distribution
$\cD_{i'j}$. This simplifies the description of the knowledge graphs.
In Appendix \ref{app:refine},
we consider settings where a player may observe private signals about
other players and can further refine the prior.

\vspace{-10pt}
\paragraph{Future directions.}
%In this paper we have strictly weakened the knowledge assumption about the seller and the players
%in Bayesian auctions,
%%and model the players' individual knowledge as knowledge graphs.
%and provided important insights about
%the relationship between the amount of knowledge in the system and
%the achievable  revenue.
As Bayesian auctions require the seller (and the players under common-prior assumption) has correct knowledge about {\em all} distributions,
%Although we strictly weaken this assumption,
in our main results
we do not consider scenarios where players have ``insider'' knowledge.
If the insider knowledge is correct (i.e., is a refinement of the prior),
then our mechanisms' revenue
increases; see Appendix \ref{app:refine}.
%~\ref{app:refine}.
%If the insider knowledge may be wrong, the problem is closely related to
%{\em robust mechanism design}~\cite{bergemann2012robust},
%which is a very important topic in game theory but not the focus of this paper.
Still, how to aggregate even the incorrect information that the players may have about each other %'s distributions
is a very interesting question for future studies.

Another important direction is to elicit players' information
for BIC mechanisms.
For example, the BIC mechanisms in \cite{cremer1988full,cai2012optimal}
are optimal in their own settings, and
it is unclear how to convert them to information elicitation mechanisms.
%As pointed out by \cite{cremer1988full}, the common prior assumption seems to be crucial for its mechanism.
%removing the common prior assumption.

\subsection{Related Work}\label{sec:related}

\paragraph{Information elicitation.}
Following \cite{miller2005eliciting},
information elicitation has become an important research area in the past decade \cite{radanovic2013robust,zhang2014elicitability,kong2016putting}.
A mechanism asks each player to report his private signal
and his private knowledge about the prior distribution.
The decision maker wants the mechanism to be BIC,
%to elicit such information truthfully at Bayesian Nash equilibrium,
and a player is rewarded based on his reported distribution and the other players' reported signals.
% and is rewarded based on everybody's report.
Different from auctions, there are no allocations or prices, and a player's utility equals his reward.
Proper scoring rules \cite{brier1950verification,cooke1991experts} are widely used.
In Appendix \ref{sec:buyknowledge}, we will formally define scoring rules and
use them to reward the players for their knowledge. % in our mechanisms.

%, a lot of effort has been devoted to.

Most studies on information elicitation require a common prior.
Mechanisms without this assumption are considered by \cite{witkowski2012peer},
and our work is information elicitation in auctions without a common prior.
Moreover, information elicitation does not consider the players to
have any cost for revealing their knowledge.
It would be interesting to include such costs in the general model as well as in ours, to see how the mechanisms will change accordingly.

\vspace{-10pt}
\paragraph {Bayesian auction design.}
In his seminal work \cite{myerson1981optimal}, Myerson introduced the first optimal Bayesian mechanism for single-good auctions,
which also applies to many single-parameter settings \cite{archer2001truthful}.
%Following
%Single good optimal auction \cite{myerson1981optimal}
%,riley1981optimal};
%Single minded \cite{lehmann2002truth};
%In \cite{cremer1988full}, the authors show that, in single-good auctions, when the players' values are highly correlated,
%a revenue equal to the highest true value can be extracted.
Since then, there has been a huge literature on designing (approximately)
optimal Bayesian mechanisms that are either
BIC or dominant-strategy truthful (DST); see \cite{hartline2007profit} for an introduction to this literature.
%In particular, \cite{ronen2001approximating} provided the 1-Lookahead mechanism which is a 2-approximation in single-good auctions even when the players' values are correlated.
Mechanisms for multi-parameter settings have been constructed recently.
%Approximate mechanism for single parameter setting \cite{, ronen2001approximating};
In \cite{cai2012optimal}, the authors characterize optimal BIC mechanisms for combinatorial auctions.
%Optimal mechanism characterization for multi-parameter auctions;
For unit-demand auctions, \cite{chawla2007algorithmic,chawla2010multi,kleinberg2012matroid,cai2016duality}
construct DST Bayesian mechanisms that are constant approximations.
% to the optimal mechanism.
For additive auctions,
\cite{hart2012approximate,li2013revenue,yao2015n,cai2016duality} provide logarithmic or constant approximations under different conditions.
Finally, 
for subadditive auctions,  
logarithmic or constant approximations 
are provided in \cite{babaioff2014simple,chawla2016mechanism,cai2017simple}.

\vspace{-10pt}
\paragraph{Removing the common prior assumption.}
Following \cite{wilson1985game},
 %and \cite{cremer1988full},
a lot of effort has been made
%in the literature trying
to remove the common prior assumption. % in game theory. %\cite{segal2003optimal,goldberg2006competitive}.
In DST Bayesian mechanisms it suffices to assume that the seller knows the prior distribution.
%\cite{kleinberg2012matroid,yao2015n}.
%Many works in this direction try to further
%weaken the assumption and
%consider settings where {\em the seller may be totally ignorant.}
%common prior assumption is very strong and
In prior-free mechanisms \cite{hartline2008optimal,devanur2009limited} the distribution is unknown and the seller learns it from the values of randomly selected players.
In \cite{cole2014sample,zhiyi2016side,morgenstern2016learning} the seller observes independent samples from the distribution before the auction begins.
In \cite{chen2013mechanism,chen2015tight} the players have arbitrary possibilistic belief hierarchies
about each other.
%In \cite{chen2017query}, the seller can only access the distributions via specific oracle queries.
In robust mechanism design \cite{bergemann2012robust}
the players have arbitrary probabilistic belief hierarchies.
 %and the studies have shown necessary and sufficient conditions for a social correspondence to be robustly implementable.
In crowdsourced Bayesian auctions \cite{azar2012crowdsourced} each player privately knows {\em all} the distributions (or their refinements), which is a special case of our model.
Indeed, all knowledge graphs will be complete graphs under their setting (that is, everybody knows everything),
while we allow arbitrary knowledge graphs.
%nformation structure here is truly distributed.
In Appendixes  \ref{app:refine} and \ref{sec:player-wise} we further discuss how to elicit the players' knowledge refinements,
and how to handle correlated distributions in a setting that is a special case of our model but is still more general than that of \cite{azar2012crowdsourced}.

\section{Preliminaries}\label{sec:model}

In this work, we focus on multi-item auctions with $n$ players (denoted by $N$) and $m$ items (denoted by $M$).
A player $i$'s value for an item~$j$, $v_{ij}$, is independently drawn from a distribution $\cD_{ij}$.
%For general combinatorial auctions, each $v_i$, player $i$'s valuation function, is independently drawn from $\cD_i$.
% where each player's valuation function is
Let $v_i = (v_{ij})_{j\in M}$, $\cD_i = \times_{j\in M} \cD_{ij}$ and $\cD =\times_{i\in N} \cD_i$.
Player $i$'s value for a subset $S$ of items is $\max_{j\in S} v_{ij}$ in {\em unit-demand} auctions, and is $\sum_{j\in S} v_{ij}$ in {\em additive} auctions.
%Our settings are downward closed \cite{hartline2009simple},
The players' utilities, denoted by $u_i$, are quasi-linear, and the players are risk-neutral.

\paragraph*{Knowledge graphs.}
It is illustrative to model the players' knowledge graphically. %
\footnote{We could have defined the players' knowledge
using the standard notion in epistemic game theory
\cite{harsanyi1967games,aumann1976agreeing,fhmvbook}:
roughly speaking, the state space consists of all possible distributions of the valuation profile, and
player~$i$ knows $\cD_{i'j}$ if he is in an information set
where all distributions have the $(i',j)$-th component equal to $\cD_{i'j}$.
However, the knowledge graph is a more succinct representation and is enough for the purpose of this work.}
We consider a vector of {\em knowledge graphs}, $G = (G_j)_{i\in M}$, one for each item. Each $G_j$ is a directed graph with
%In single-parameter settings, the {\em knowledge graph} has
$n$ nodes, one for each player.
For any $i\neq i'$, an edge $(i, i')$ is in $G_j$ if and only if player $i$ knows~$\cD_{i'j}$.
%A knowledge graph does not have self-loops.
%: a player $i$'s knowledge about his own
%value distributions is not considered.
%
%\footnote{Player $i$ may know his own distributions,
%but this knowledge is neither used by our mechanisms nor affecting $i$'s strategies.}
%there is an edge from $i$ to $j$ if and only if player $i$ (privately) knows $\cD_{j}$.%
%\footnote{We could have defined the players' knowledge
%using the standard notion
%\cite{harsanyi1967games, aumann1976agreeing, fhmvbook}:
%roughly speaking, the state space consists of all possible distributions of the valuation profile, and
%% all distributions in the same information set of player $i$ have the same $j$-th component.
%%player~$i$'s information sets form a finer partition of the partition according to the distribution of $j$'s value.
%player~$i$ knows $\cD_{i'j}$ if he is in an information set
%where all distributions have the $(i',j)$-th component equal to $\cD_{i'j}$.
%However, the knowledge graph is a more succinct representation and is enough for the purpose of this paper.}
There is no constraint about the knowledge graphs: the same player's distributions for different items may be known by
different players, different players' distributions for the same item may also be known by different players,
and some distributions may not be known by anybody.
Each player knows his own out-going edges, and neither the players nor the seller knows the whole graph.
%we distinguish two cases:
%{\em partial information},
%where for each player $i$ and item $j$, there exists a player $i'$ who knows $\cD_{ij}$; and {\em player-wise information}, where, for each player $i$, there exists a player $i'$ who knows the whole distribution $\cD_i$ of $i$'s valuation.
%Accordingly, the partial information setting corresponds to $m$ knowledge graphs, one for each item,
%and each player may only know part of $\cD_i$; while the player-wise information setting corresponds to a single knowledge graph.

We measure the amount of knowledge in the system by the number of players knowing each distribution.
For any $k\in \{0,1\dots,n-1\}$, a knowledge graph is {\em $k$-informed} if each node has in-degree at least $k$: a player's distribution is known by at least $k$ other players.
The vector $G$ is $k$-informed if all knowledge graphs are so.
%In partial information settings,
%there is no constraint about whom these $k$ players are for player $i$'s values for different items. It
%one set of $k$ players know $\cD_{i1}$, another set of $k$ players know $\cD_{i2}$, etc.
Note that every knowledge graph is $0$-informed,
%a partial information setting with
and ``everything is known by somebody'' when $k\geq 1$.
%, and $1$-informedness is a very weak assumption about the players' knowledge, and in fact the weakest when ``everything is known''.
%in the crowdsourced Bayesian model.
A common prior would imply all knowledge graphs
%(whether of partial information or player-wise information)
are complete directed graphs, or $(n-1)$-informed, which is the strongest condition in our model.
%We denote by $\cC_{n, m, k}$ the set of crowdsourced Bayesian auction instances with $n$ players and $m$ items, where the knowledge graphs are $k$-informed.
%In our setting, the information of distribution is not a common knowledge.
%We assume each player posses a belief about someone else's distribution.
%Let's use $K_{i} = \{(i',j) |$ player $i$ knows $D_{i'j}\}$ to denote the set of player-item pair whose distribution is known to player $i$.  Note that $K_{i}$ may be empty set, i.e., $i$ does not have any information at all.\footnote{When we consider correlated valuations, each player either has full information about another player or has no information. The elements in $K_i$ are simplified to players rather than player item pairs.}
%The belief $B_i: K_{i} \rightarrow \Delta(\Omega)$ of player $i$ is a mapping from each player-item in $K_i$ to its' corresponding distribution.
%We assume that each player has correct belief, i.e, $\forall (i',j) \in K_i, B_i(i',j) = D_{i'j}$. Let all players belief profile be $B=\times_{i\in N}B_{i}$.
The seller's knowledge can be naturally incorporated into the knowledge graphs by
considering him as a special ``player 0''.
%, whose out-going edges
%represent the distributions he knows.
All our mechanisms can easily utilize the seller's knowledge, and we will not further discuss this issue.
%The mechanism can ensure that no item is ever ``sold'' to the seller and his price is always 0.
 %receives any item and has a price 0.
%one can add a node for the seller, with the out-going edges representing the distributions he knows.
%This node has in-degree 0 as the seller does not have a value for the item.

\paragraph*{Information elicitation mechanisms.}
%Given the set $N = [n]$ of players, the set $M = [m]$ of items, and the distribution~$\cD$,
Let $\hat{\cI} = (N, M, \cD)$ be a Bayesian auction instance and
$\cI = (N, M, \cD, G)$ a corresponding information elicitation instance, where $G$ is a knowledge graph vector.
Different from Bayesian mechanisms, which has $\cD$ as input,
an information elicitation mechanism has neither $\cD$ nor $G$ as input.
Instead, it asks each player $i$ to report a valuation $b_i = (b_{ij})_{j\in M}$ and
a {\em knowledge} $K_i = \times_{i'\neq i, j\in M} \cD^i_{i'j}$ ---a distribution for the valuation subprofile $v_{-i}$. $K_i$ may contain ``$\bot$'' at some places,
indicating $i$ does not know the corresponding distributions.
$K_i$ is $i$'s {\em true knowledge} if $\cD^i_{i'j} = \cD_{i'j}$ whenever $(i, i')\in G_j$, and $\cD^i_{i'j} = \bot$ otherwise.
%The seller's knowledge, if any, is also given to the mechanism as an input.
%, while in the worst case this part is just $\bot$.
An information elicitation mechanism maps a strategy profile $(b_i, K_i)_{i\in N}$ to an allocation and a price profile, and may be randomized.
To distinguish whether a mechanism $\cM$ is a Bayesian or an information elicitation mechanism, we may explicitly write $\cM(\hat{\cI})$ or $\cM(\cI)$.
%Again note that the latter does not mean $\cM$ has $\cD$ or $G$ as input.
The (expected) revenue of $\cM$ is denoted by $Rev(\cM)$, and sometimes by $\bE_{\cD} Rev(\cM)$ to emphasize the distribution.

%\begin{definition}
%Given that the belief for each player is the true distribution for other players,
An information elicitation mechanism is {\em 2-step dominant strategy truthful} (2-DST) if
%the following two conditions hold:
 \begin{itemize}
\item[(1)] For any player $i$, true valuation $v_i$,
%fixing any reported knowledge $K_i$, it is {\em (weakly) dominant} for $i$ to report his true valuation $v_i$. That is, for
 valuation $b_i$, knowledge $K_i$, and strategy subprofile $s_{-i} = (b_j, K_j)_{j\neq i}$ of the other players,
$u_i((v_i, K_i), s_{-i}) \geq u_i((b_i, K_i), s_{-i})$.
% for any strategy subprofile $s_{-i}$,
%and the inequality is strict for some $s_{-i}$.

\item[(2)]
%Given that all players report their true valuations, it is {\em (weakly) dominant} for each player $i$ to report his true knowledge.
For any player $i$, true valuation $v_i$, true knowledge $K_i$,
% $i$'s true knowledge by $TK_i$, then
% = \times_{i'\neq i, j\in [m]} \cD^i_{i'j}$, where $\cD^i_{i'j} = \cD_{i'j}$ if edge $(i, i')$ is in the knowledge graph for item $j$ and $\cD^i_{i'j} = \bot$ otherwise.
%for any true valuation $v_i$,
knowledge $K'_i$, and knowledge subprofile $K'_{-i}(v_{-i}) = (K'_j(v_j))_{j\neq i}$ of the other players,
where each $K'_j(v_j)$ is a function of player $j$'s true valuation~$v_j$,
$\bE_{v_{-i}\sim \cD_{-i}} u_i((v_i, K_i), (v_{-i}, K'_{-i}(v_{-i}))) \geq \bE_{v_{-i}\sim \cD_{-i}} u_i((v_i, K'_i), (v_{-i}, K'_{-i}(v_{-i})))$.%
\footnote{In Appendix \ref{sec:buyknowledge}, we introduce scoring rules to our mechanisms
so that the inequality is strict whenever $K'_i\neq K_i$.}
% for any knowledge subprofile $K'_{-i}$,
%and
%the inequality is strict for some $K'_{-i}$.
 \end{itemize}

\section{Under Arbitrary Knowledge Graphs}
\label{sec:k=0}

\subsection{Knowledge-Based Revenue Benchmark}
When the knowledge graphs can be totally arbitrary,
some distributions may not be known by anybody.
It is not hard to see that in this case, no information elicitation mechanism can be a bounded approximation
to $OPT$.
Indeed, if all but one value distributions of the players are constantly 0, and
if the only non-zero distribution, denoted by $\cD_{ij}$, is unknown by anybody,
then a Bayesian mechanism can find the optimal reserve price based on $\cD_{ij}$,
while an information elicitation mechanism can only set the price for player $i$ based on the reported values
of the other players, which are all 0.
%The mechanisms here can be directly applied to previous settings.

%
%As we have said in Section \ref{subsec:results}, if for some player $i$ and item $j$,
%nobody knows $\cD_{ij}$,
%then no crowdsourced Bayesian mechanism can achieve a bounded approximation to $OPT$. Indeed,
%it is possible that all other values are constantly 0, and $v_{ij}$ is the only source for revenue.
%%---for example, if, then the instance is essentially a single-player single-good auction.
%Although a Bayesian mechanism can find the optimal reserve price given $\cD_{ij}$,
%a crowdsourced Bayesian mechanism can only set the same reserve price for all distributions
%%, since the seller has no knowledge about the true distribution,
%and the approximation ratio is unbounded.
% %can guarantee revenue close to the optimal reserve price as.

Thus, for arbitrary knowledge graphs,
we define a natural revenue benchmark: the optimal Bayesian revenue
{\em on players and items for which the
distributions are known in the information elicitation setting.}
More precisely,
let $\hat{\cI} = (N, M, \cD)$ be a Bayesian instance and $\cI = (N, M, \cD, G)$
 a corresponding information elicitation instance.
Let $\cD'=\times_{i\in N, j\in M} \cD'_{ij}$ be such that $\cD'_{ij} = \cD_{ij}$ if
there exists a player $i'$ with $(i', i)\in G_j$, and $\cD'_{ij}$ is
constantly 0 otherwise.
We refer to $\cD'$ as {\em $\cD$ projected on $G$}.
Letting $\cI' = (N, M, \cD')$ be the resulting Bayesian instance,
the {\em knowledge-based} revenue benchmark is $OPT_K(\cI) \triangleq OPT(\cI')$, the optimal BIC revenue 
on $\cI'$.
This is a demanding benchmark in information elicitation settings:
it takes into consideration the knowledge of {\em all} players, no matter who knows what.
When everything is known by somebody, even if $G$ is only 1-informed, we will have $\cI' = \hat{\cI}$ and $OPT_K(\cI) = OPT(\hat{\cI})$.
%Below we show how to approximate $OPT_K$ in unit-demand auctions and additive auctions, starting with the former.
%where our mechanisms are easier to describe and analyze.

%\vspace{-5pt}
\subsection{Unit-Demand Auctions}\label{subsec:unit}

%\vspace{-3pt}

%
%We start with the crowdsourced Bayesian mechanism $\cM_{CSUD}$ for
%see Mechanism~\ref{alg:generalud}.
%, which takes
%%. It uses Myerson's mechanism in a non-blackbox way and takes
%as input a positive number~$\epsilon$ that can be arbitrarily small
For unit-demand auctions, sequential post-price Bayesian mechanisms have been constructed by \cite{chawla2010multi,kleinberg2012matroid}.
For information elicitation, if the seller asks the players to report both their values and knowledge, and directly uses the reported distributions  in these mechanisms,
%to set the prices,
then a player may want to withhold his knowledge about the other players.
By doing so, a player may prevent
the seller from selling the items to the others, so that the items are still available when it is his turn to buy.
%For example, if $i$ is the only person know knowledge about the distributions

A simple idea is to partition the players into two groups: a set of {\em reporters}
who will not receive any item and is only asked to report their knowledge;
and a set of {\em potential buyers} whose knowledge is never used.
It is possible that the reported knowledge may not contain a potential buyer's value distributions on all items,
%and the seller will not sell to~$i$ the items for which his distributions are not reported.
thus the technical part is to prove that
the seller generates a good revenue even though the players' knowledge is only partially recovered.

Our mechanism $\cM_{IEUD}$ is simple and intuitive; see Mechanism~\ref{alg:generalud},
where $\cM_{UD}$ is the Bayesian mechanism of \cite{kleinberg2012matroid}.
%We have the following theorem.
%
% and \cite{cai2016duality}
%The mechanism randomly partitions the players into two groups,
%asks one group to report their knowledge about the other,
%and runs the DST Bayesian mechanism of \cite{kleinberg2012matroid},
%denoted by $\cM_{UD}$, on the latter.
%Mechanism $\cM_{UD}$ is 24-approximation to $OPT$ as shown by \cite{kleinberg2012matroid} and \cite{cai2016duality}.
%Therefore our mechanism is simple and easy to implement in reality.
It's worth pointing out that, although mechanism $\cM_{UD}$ is used as a black-box,
Mechanism \ref{alg:generalud} is not a reduction from arbitrary Bayesian mechanisms.
Instead, we will prove a {\em projection lemma}
that allows such a reduction from an important class of Bayesian mechanisms, where mechanism $\cM_{UD}$ is an important example.
%and we choose $\cM_{UD}$ for the best approximation ratio.
We have the following theorem,
whose proof is sketched below and formal proof is provided in Appendix \ref{app:unknown:unit}.
%~\ref{app:unknown:unit}.

\begin{algorithm}[htbp]
\floatname{algorithm}{Mechanism}
  \caption{\hspace{-2pt}{\bf .} $\cM_{IEUD}$}
  \label{alg:generalud}
  \begin{algorithmic}[1]
 \STATE\label{step1} Each player $i$ reports to the seller a valuation $b_i = (b_{ij})_{j\in M}$ and a knowledge $K_i = (\cD^i_{i'j})_{i'\neq i, j\in M}$.

 \STATE Randomly partition the players into two sets, $N_1$ and $N_2$,
 where each player is independently put in each set with probability $\frac{1}{2}$.
 % and $N_2$ with probability $1-q = \frac{1}{2}$.
 %The value $q$ will be decided in the analysis.\label{step2}

 \STATE Set $N_3 = \emptyset$.
 % and $p_i = 0$ for each player $i$.

 \FOR {players $i\in N_1$ lexicographically}
   \STATE\label{step5} For each player $i'\in N_2$ and item $j\in M$,
   if $\cD'_{i'j}$ has not been defined yet and $\cD^i_{i'j} \neq \bot$,
   then set $\cD'_{i'j} = \cD^i_{i'j}$
%   (2) reward player $i$ using Brier's scoring rule: that is,
%   $p_{i}=p_i -\frac{\epsilon}{2mn}BSR(\cD'_{i'j}, b_{i'j})$;
and add player $i'$ to $N_3$.
%\label{step:scoring rule:ud}
  \ENDFOR

 \STATE \label{step12} For each $i\in N_3$ and $j\in M$
  such that $\cD'_{ij}$ is not defined,
  set $\cD'_{ij}\equiv 0$ (i.e., 0 with probability 1) and $b_{ij}=0$.

 \STATE Run mechanism $\cM_{UD}$ on the unit-demand Bayesian auction $(N_3, M, (\cD'_{ij})_{i\in N_3, j\in M})$,
 with the players' values being $(b_{ij})_{i\in N_3, j\in M}$.
  Let $x' = (x'_{ij})_{i\in N_3,\, j\in M}$ be the resulting allocation where $x'_{ij}\in \{0, 1\}$, and let $p' = (p'_i)_{i\in N_3}$ be the prices.
Without loss of generality,
$x'_{ij} = 0$ if $\cD'_{ij}\equiv 0$.\label{step13_1}

 \STATE For each player $i\not\in N_3$, $i$ gets no item and his price is $p_i =0$.

 \STATE For each player $i\in N_3$, $i$ gets item $j$ if $x'_{ij}=1$, and his price is $p_i = p'_i$.

\end{algorithmic}
\end{algorithm}

\begin{theorem}\label{thm:unit}
Mechanism $\cM_{IEUD}$ for unit-demand auctions is 2-DST and,
for any instances $\hat{\cI} = (N, M, \cD)$ and
$\cI = (N, M, \cD, G)$,
% and $\cI' = (N, M, \cD')$ with $\cD'$ being $\cD$ projected on $G$,
$Rev(\cM_{IEUD}(\cI)) \geq \frac{OPT_K(\cI)}{96}$.
\end{theorem}

%From now on we may simply say that $\cM_{CSUD}$ is a $\frac{1}{96}$-approximation.

%We first consider the 2-DSTness.
\begin{lemma}\label{ud:truthful}
Mechanism ${\cal M}_{IEUD}$ is 2-DST.
\end{lemma}

\vspace{-8pt}

\begin{proof}[Proof sketch]
%{\it Proof Sketch of Lemma \ref{lem:2dst}.}
The key is that
%${\cal M}_{CSUD}$,
the use of the players' values and the use of their knowledge are disentangled:
for players in $N_1$, the mechanism only uses their knowledge but not their values; and the opposite holds for players in $N_2$.
If a player~$i$ ends up in $N_2$, then whether he is in $N_3$ or not does not depend on his own strategy.
As mechanism $\cM_{UD}$ is DST and player $i$ is assigned to $N_2$ with positive probability,
it is dominant for him to report his true values in  ${\cal M}_{IEUD}$, no matter what the reported knowledge is.
Moreover, if a player $i$ ends up in $N_1$, then he is guaranteed to get no item and pay 0, thus reporting his true knowledge never hurts him.
\end{proof}

\vspace{-5pt}

In Appendix \ref{sec:buyknowledge},
we add scoring rules to mechanism $\cM_{IEUD}$ to reward the players' knowledge,
so that a player's utility will be strictly larger when he reports his true knowledge than when he lies about it.

To analyze the revenue of $\cM_{IEUD}$, note that it runs the
Bayesian mechanism on a smaller (randomized) Bayesian instance: $\hat{\cI}$ projected to
the player-item pairs $(i, j)$ such that $i\in N_3$ and $\cD_{ij}$ has been reported.
To understand how much revenue is lost by the projection,
we consider
%
%The analysis of the revenue bound for Mechanism $\cM_{CSUD}$ relies on
the  COPIES instance~\cite{chawla2010multi}, $\hat{\cI}^{CP} = (N^{CP}, M^{CP}, \cD^{CP})$, as a bridge between
the original Bayesian instance and the information elicitation instance.
$\hat{\cI}^{CP}$ is obtained from $\hat{\cI}$ by replacing each player with $m$ copies and each item with $n$ copies,
where a player $i$'s copy $j$ only wants item $j$'s copy $i$, with the value distributed according to $\cD_{ij}$.
Thus $\hat{\cI}^{CP}$ is a single-parameter auction,
with $N^{CP} = N\times M$, $M^{CP}=M\times N$, and $\cD^{CP} = \times_{(i, j)\in N^{CP}} \cD_{ij}$.

%Our crowdsourced mechanism achieves at least $\frac{1}{6}$ fraction of the optimal revenue on the COPIES instance of the reported distributions
%by running $\cM_{UD}$ as a blackbox on reported distributions \cite{kleinberg2012matroid}.
We now lower-bound the optimal BIC revenue
in the {\em projected COPIES instance.}
%For now,
%
%connects the optimal revenue on the projected COPIES instance and the optimal revenue of the whole COPIES instance projected on the subset of players and item.
% the COPIES instance of the reported distributions with that of the original unit-demand Bayesian instance.
For any subset $NM\subseteq N\times M$,
%the instance $\hat{\cI}_{NM}$ is referred to as {\em $\hat{\cI}$ projected to $NM$}
let $\hat{\cI}^{CP}_{NM}$ be $\hat{\cI}^{CP}$ projected to $NM$.
By definition, $OPT(\hat{\cI}^{CP}_{NM})$ is the optimal BIC revenue for $\hat{\cI}^{CP}_{NM}$.
Moreover, let $OPT(\hat{\cI}^{CP})_{NM}$ be the revenue of the optimal BIC mechanism for $\hat{\cI}^{CP}$ {\em obtained from players in $NM$.}
% we have the following.

\begin{lemma} [{\em The projection lemma}]\label{lem:proj}
For any $\hat{\cI}$ and $NM\subseteq N\times M$,
$OPT(\hat{\cI}^{CP}_{NM}) \geq OPT(\hat{\cI}^{CP})_{NM}$.
\end{lemma}

We elaborate the related definitions and prove Lemma \ref{lem:proj} in Appendix \ref{app:unknown:unit}.
%\ref{app:unknown:unit}.
%\begin{proof}[Proof sketch of Theorem \ref{thm:unit}]
%Letting $\cI' = (N, M, \cD')$ be the Bayesian instance where $\cD'$ is $\cD$ projected on the knowledge graph $G$,
%we define the COPIES instance $\cI'^{CP}$ similarly.
%By \cite{kleinberg2012matroid} and \cite{cai2016duality},
%the revenue of mechanism $\cM_{UD}$ on $\cI'$ is a 6-approximation to $OPT(\cI'^{CP})$, and
%the latter is at least $\frac{OPT(\cI')}{4}$.
%Theorem \ref{thm:unit} holds by combining these facts with Lemma \ref{ud:truthful} and~\ref{lem:proj}.
Given mechanism $\cM_{IEUD}$, the subset $NM$ is the set of player-item pairs $(i,j)$ such that $i\in N_3$ and $\cD_{ij}$
is reported.
Theorem~\ref{thm:unit} holds by combining
the projection lemma,
the randomized partition in $\cM_{IEUD}$, and
 the results on COPIES setting in Bayesian auctions
% and the corresponding Bayesian instance is $\cI'$ (and also $\hat{\cI}$) projected to $NM$, denoted by $\cI'_{NM}$.
\cite{kleinberg2012matroid,cai2016duality}.
%See Appendix \ref{app:unknown:unit} for the detailed analysis.
%\end{proof}

%We immediately have the following.
%\vspace{-10pt}
%\paragraph*{Remark.}
Note that Lemma \ref{lem:proj}
is only concerned with COPIES instances. Using this lemma and similar to our proof of Theorem~\ref{thm:unit},
any Bayesian mechanism
 $\cM$ whose revenue can be properly lower-bounded by the COPIES instance
 can be converted to an information elicitation mechanism in a black-box way.
%It is interesting that the COPIES setting serves as a bridge between
%Bayesian and crowdsourced Bayesian mechanisms.
We have the following theorem, with the proof omitted.

\begin{theorem}\label{col:copies}
Let $\cM$ be any DST Bayesian mechanism such that
$Rev({\cal M}(\hat{\cal I}))
\geq \alpha OPT(\hat{\cal I}^{CP})$
for some $\alpha>0$.
%When the knowledge graphs are $k$-informed with $k\geq 1$,
There exists a 2-DST information elicitation mechanism
that uses $\cM$ as a black-box and is a $\frac{\alpha}{16}$-approximation to $OPT_K$.
\end{theorem}

By Theorem \ref{col:copies},
the mechanisms in \cite{chawla2007algorithmic} and \cite{chawla2010multi}
automatically imply information elicitation mechanisms.
For single-good auctions, replacing mechanism $\cM_{UD}$ with Myerson's mechanism,
the information elicitation mechanism is a 4-approximation to $OPT_K$.

%\vspace{-10pt}
\subsection{Additive Auctions}
\label{sec:partial:additive}

%\vspace{-5pt}

Information elicitation mechanisms for additive auctions are harder to construct and analyze than for unit-demand auctions.
First, randomly partitioning the players as before
may cause a significant revenue loss,
as the revenue of additive auctions may come from selling a subset of items as a {\em bundle}
to a player $i$.
Even when $i$'s value distribution for each item is reported with constant probability,
the probability that his distributions for all items in the bundle are reported may be very low,
thus the mechanism may rarely sell the bundle to $i$ at the optimal price.
% the optimal bundling price for $i$.
%selling the bundle to him under the desired price.
Second, the seller can no longer ``throw away'' player-item pairs whose distributions are not reported and focus on the projected instance.
% ---recall that $\cI' = (N, M, \cD')$ where $\cD'$ is $\cD$ projected on~$G$.
When the players are not partitioned into reporters and potential buyers, doing so
may cause a  player to lie and
withhold his knowledge about others, so that they are thrown away.

To simultaneously achieve truthfulness and a good revenue guarantee,
our mechanism is very stingy and never throws away any information.
If a player $i$'s value distribution for an item $j$ is reported by others,
then $j$ may be sold to $i$ via
the {\em $\beta$-Bundling} mechanism of \cite{yao2015n}, denoted by $Bund$.
If $i$'s distribution for $j$ is not reported, then $j$ may still be sold to $i$ via the second-price mechanism.
Indeed, our mechanism handles the players neither solely based on the original Bayesian instance $\hat{\cI}$
nor solely based on the projected instance $\cI'$. Rather, it works on a {\em hybrid} of the two.
% of the original Bayesian instance $\hat{\cI}$ and the projected instance $\cI'$.}

Our mechanism $\cM_{IEA}$ is still simple; see Mechanism \ref{alg:newbvcg}.
However, significant effort is needed to analyze its revenue.
Indeed, note that in Mechanism \ref{alg:newbvcg}, each $M_i$ is defined according to the original
Bayesian instance $\hat{\cI}$, while the partition of $M$ is done according to the
knowledge graphs in the information elicitation instance $\cI$.
The mechanism $Bund$ is run on a hybrid instance, where $\beta_i$
is based on $\hat{\cI}$
and $\cD_i'$ is based on $\cI$.
Finally,
part of player $i$'s winning set
is sold according to mechanism $Bund$
and part of it is sold using second-price.

%~\ref{app:unknown:add}.

\begin{algorithm}
\floatname{algorithm}{Mechanism}
  \caption{\hspace{-2pt}{\bf .} $\cM_{IEA}$}
  \label{alg:newbvcg}
  \begin{algorithmic}[1]
 \STATE  Each player $i$ reports a valuation $b_i = (b_{ij})_{j\in M}$ and a knowledge $K_i = (\cD^i_{i'j})_{i'\neq i, j\in M}$.

% \STATE  Reward the players using Brier's scoring rule, such that the total reward $R$ is at most $\epsilon$.

 \STATE  For each item $j$, set $i^*(j)= \argmax_i b_{ij}$ (ties broken lexicographically) and $p_j = \max_{i\neq i^*} b_{ij}$.

 \FOR{each player $i$}
   \STATE Let $M_i = \{j \ | \  i^*(j) = i\}$ be player $i$'s {\em winning set.}

   \STATE Partition $M$ into $M^1_{i}$ and $M^2_{i}$ as follows:
   $\forall j \in M^1_{i}$,
   some $i'$ has reported $\cD^{i'}_{ij}\neq \bot$
   (if there are more than one reporters, take the lexicographically first);
   and $\forall j \in M^2_{i}$, $\cD^{i'}_{ij}=\bot$ for all $i'$.

   \STATE $\forall j\in M^1_{i}$, set $\cD'_{ij} = \cD^{i'}_{ij}$;
   and $\forall j\in M^2_{i}$, set $\cD'_{ij}\equiv 0$.

   \STATE Compute the optimal entry fee $e_i$ and reserve prices
   $(p'_{j})_{j\in M^1_{i}}$ according to mechanism $Bund$ with respect to $(\cD'_{i}, \beta_{i})$,
   where $\beta_{ij} = \max_{i'\neq i} b_{i'j}$ $\forall j\in M$.
   By the definition of $Bund$, we always have $p'_j \geq \beta_{ij}$ for each $j$.
   If $e_i=0$ then it is possible that $p'_j >\beta_{ij}$ for some $j$;
   while if $e_i>0$ then $p'_{j} = \beta_{ij}$ for every $j$. \label{step:M'IEBVCG8}

   \STATE Sell $M_i^1\cap M_i$ to player $i$ according to $Bund$.
   That is, if $e_i>0$ then do the following:
   if $\sum_{j\in M^1_{i}\cap M_i} b_{ij}\geq e_i
   + \sum_{j\in M^1_{i}\cap M_i} p'_j$,
   player $i$ gets $M^1_{i}\cap M_i$ with price
   $e_i + \sum_{j\in M^1_{i}\cap M_i} p'_j$;
   otherwise the items in $M^1_{i}\cap M_i$ are not sold.
   If $e_i=0$ then do the following: for each item $j\in M^1_{i}\cap M_i$,
   if $b_{ij}\geq p'_j$, player $i$ gets item $j$ with price $p'_j$;
   otherwise item $j$ is not sold.\label{step:M'IEBVCG9}

   \STATE In addition, sell each item $j$ in $M^2_{i}\cap M_i$ to player $i$
   with price $p_j (=\beta_{ij})$. \label{step:M'IEBVCG10}
 \ENDFOR

\end{algorithmic}
\end{algorithm}

Although running $Bund$ on $\cI'$ achieves a constant approximation to $OPT(\cI')$,
some items sold by $Bund$ under $\cI'$ may end up being sold by $\cM_{IEA}$ using second-price,
and
the revenue of $\cM_{IEA}$ cannot be lower-bounded by that of $Bund$ on $\cI'$.
%Instead,
%and must be {\em patched up with a second-price sale to compensate the revenue loss caused by throwing away
%some player-item pairs.}
%With $\cM_{CSA}$ defined in Mechanism \ref{alg:newbvcg},
To overcome this difficulty, we develop a novel way to use the {\em adjusted revenue}~\cite{yao2015n}
in our analysis; see Lemmas \ref{lem:IEA:A} and \ref{lem:IEA:C} in Appendix~\ref{app:unknown:add},
%\ref{lem:csa:A} and \ref{lem:csa:C} in Appendix~\ref{app:unknown:add},
where we also recall the related definitions.
As we show there, the adjusted revenue in a hybrid information setting, combined with the revenue of the second-price sale,
eventually provides a desirable lower-bound to the revenue of $\cM_{IEA}$.
We have the following theorem, proved in Appendix~\ref{app:unknown:add}.

\begin{theorem}\label{thm:newbvcg}
\sloppy
Mechanism $\cM_{IEA}$ for additive auctions  is 2-DST  and,
for any instances $\hat{\cI} = (N, M, \cD)$ and
$\cI = (N, M, \cD, G)$,
% and $\cI' = (N, M, \cD')$ with $\cD'$ being $\cD$ projected on $G$,
$Rev(\cM_{IEA}(\cI)) \geq \frac{OPT_K(\cI)}{70}$.
%$\cM_{CSA}$  is a $70$-approximation to $OPT_K(\cI)$.
\end{theorem}

%\vspace{-5pt}

\section{When Everything Is Known by Somebody}
\label{sec:partial}

When the knowledge graph vector $G$ is $k$-informed with $k\geq 1$, ``everything is known by somebody'' and $OPT_K = OPT$.
Both mechanisms in Section \ref{sec:k=0} of course apply here, but we can do better when $k$ gets larger: that is, when the amount of knowledge in the system increases.

%\vspace{-5pt}

\subsection{Unit-Demand Auctions and Additive Auctions}
\label{sec:partial:ud}

For unit-demand auctions, mechanism $\cM'_{IEUD}$ here is almost the same as mechanism $\cM_{IEUD}$,
except it randomly partitions the players
differently.
%uses a  randomly partition the players.
%The probability that each player is assigned to $N_1$ is
%chosen to achieve the maximum probability for each distribution to be reported,
%which only depends on $k$.
The probability that each player is assigned to $N_1$
is now $q = 1-(k+1)^{-\frac{1}{k}}$, and the probability to $N_{2}$ is $1-q$.
When $k=1$, we have $q=\frac{1}{2}$ and mechanism $\cM'_{IEUD}$ is exactly as before.
The probability~$q$ is chosen to achieve the maximum probability for each distribution to be reported, and the latter
is exactly $\tau_k = \frac{k}{(k+1)^{\frac{k+1}{k}}}$.
We only state the theorem below,
proved in Appendix \ref{app:known:unit}.

\begin{theorem}
\label{thm:unit-k}
$\forall k\in [n-1]$, any unit-demand auction instances $\hat{\cI} = (N, M, \cD)$ and $\cI = (N, M, \cD, G)$ where $G$ is $k$-informed,
mechanism $\cM'_{IEUD}$ is 2-DST and
$Rev(\cM'_{IEUD}(\cI)) \geq \frac{\tau_k}{24} \cdot OPT(\hat{\cI}).$
\end{theorem}

%Recall that $\tau_k=\frac{k}{(k+1)^{\frac{k+1}{k}}}$.
%Note that $\tau_k$ is increasing in $k$, $\tau_1 = \frac{1}{4}$ and $\tau_k\rightarrow 1$.
As $k$ gets larger (although can still be much smaller than $n$),
the approximation ratio of $\cM'_{IEUD}$ approaches $24$,
the best known approximation to $OPT$ by DST Bayesian mechanisms~\cite{kleinberg2012matroid,cai2016duality}.
%Moreover, by Lemma 5 of  \cite{chawla2010multi}, $\cM'_{CSUD}$  is a $\frac{\tau_{k}}{6}$-approximation to the optimal deterministic DST Bayesian mechanism.

%
%
%\subsection{Additive Auctions}
%\label{sec:partial:additive}
%
%%\vspace{-5pt}
%
%Additive auctions here
%are again more difficult than unit-demand auctions,
%but still easier than the case when the knowledge graphs can be totally arbitrary.
%When $k\geq 1$, all distributions will be reported in our mechanism $\cM_{CSA}$.
%Thus no item is sold according to the second-price mechanism,
%and $\cM_{CSA}$'s outcome is the same as the $\beta$-Bundling mechanism of \cite{yao2015n} applied to the original Bayesian instance $\hat{\cI}$.

For additive auctions, to improve the approximation ratio when $k\geq 1$, following \cite{cai2016duality} we can divide
the $\beta$-Bundling mechanism into the ``bundling part'' and the ``individual sale part''.
The former is referred to as the {\em Bundle VCG} mechanism, denoted by $BVCG$;
and the latter is the {\em individual 1-lookahead} mechanism, denoted by $\cM_{1LA}$, which sells each item separately using the 1-lookahead mechanism of \cite{ronen2001approximating}.
Mechanism $\cM_{1LA}$ can also be replaced by
the {\it individual Myerson} mechanism, denoted by $IM$,
which sells each item separately using Myerson's mechanism.
By choosing the mechanism that generates a higher expected revenue between $IM$ and $BVCG$,
\cite{cai2016duality} provides a Bayesian mechanism that is an 8-approximation to $OPT$.

%For crowdsourced Bayesian auctions,
We design corresponding information elicitation mechanisms separately for mechanisms $BVCG$ and $IM$.
The resulting mechanisms are denoted by $\cM_{IEBVCG}$ and  $\cM_{IEIM}$,
which are defined in Appendix~\ref{app:known:additive}. % (Mechanism~\ref{alg:bvcg}).
Because the seller does not know the prior $\cD$,
he cannot compute the expected revenue of the two
information elicitation mechanisms and choose the better one.
Instead, we let him choose between the two mechanisms randomly, according to a probability distribution depending on $k$.
It is worth pointing out that when $k\geq 1$, 
mechanism $\cM_{IEBVCG}$ is able to recover {\em all} distributions in $\hat{\cI}$,
thus its revenue equals the corresponding Bayesian revenue, which is lower-bounded in \cite{cai2016duality}.
For $\cM_{IEIM}$, information elicitation is done by randomized partition depending on the value of $k$.

%we can easily ``crowdsource'' mechanism $BVCG$ following mechanism $\cM_{CSA}$.
%The resulting mechanism is denoted by $\cM_{CSBVCG}$ and
%defined in Appendix~\ref{app:known:additive} (Mechanism~\ref{alg:bvcg}).
%Moreover, we can easily ``crowdsource'' mechanism $IM$, similar to mechanism $\cM'_{CSUD}$.
%The resulting mechanism is denoted by $\cM_{CSIM}$ and also defined in the appendix.
%Because the seller does not know the prior $\cD$,
%he cannot compute the expected revenue of the two
%crowdsourced Bayesian mechanisms and choose the better one.
%Instead, we let him choose between the two mechanisms randomly, according to a probability distribution depending on $k$.

However, we can do even better.
Indeed, although in Bayesian auctions the mechanism $IM$ is optimal for 
individual item-sale and outperforms mechanism $\cM_{1LA}$,
in information elicitation auctions there is a tradeoff between the two.
In order for the players to report their knowledge truthfully for mechanism $IM$,
we need to randomly partition them into reporters and potential buyers,
thus each distribution is only recovered with probability $\tau_k$.
In contrast, no partition is needed for aggregating the players' knowledge in mechanism $\cM_{1LA}$,
and we can recover all distributions simultaneously with probability 1.
The resulting information elicitation mechanism, $\cM_{IE1LA}$, is also defined in the appendix.
As mechanism $\cM_{1LA}$ is a 2-approximation to mechanism $IM$,
sometimes it is actually more advantageous to use $\cM_{IE1LA}$ rather than $\cM_{IEIM}$,
depending on the value of $k$.

%the {\em Individual 1-Lookahead} mechanism, denoted by $\cM_{1LA}$,
%which sells each item separately using
%the 1-Lookahead mechanism~\cite{ronen2001approximating} and is a 2-approximation to Myerson's mechanism for each item;
%and the {\it Bundle VCG} mechanism~\cite{yao2015n, cai2016duality}, denoted by $BVCG$,
%which is the VCG mechanism with an optimal entry fee for each player.
%We crowdsource them separately---$\cM_{CSIM}$, $\cM_{CSBVCG}$, and $\cM_{CS1LA}$ respectively.
Properly combining the above gadgets together,
our mechanism $\cM'_{IEA}$ is defined as follows:
when $k\leq 7$, it runs $\cM_{IEBVCG}$ with probability $\frac{2}{11}$ and
$\cM_{IE1LA}$ with probability $\frac{9}{11}$;
when $k> 7$, it runs $\cM_{IEBVCG}$ with probability $\frac{\tau_k}{3+\tau_k}$ and
$\cM_{IEIM}$ with probability $\frac{3}{3+\tau_k}$.
The choice of the two cases is to achieve the best approximation ratio for each $k$.
%It flips a fair coin; if heads comes up then it runs $\cM_{CSBVCG}$; and if tails comes up, then it runs $\cM_{CS1LA}$ when $k\leq 3$ and $\cM_{CSIM}$ when $k>3$.
%we obtain a crowdsourced Bayesian mechanism for additive auctions, denoted by $\cM'_{CSA}$.
%The mechanisms are formally defined and analyzed in Appendix~\ref{app:known:additive}.
We have the following theorem, proved in Appendix~\ref{app:known:additive}.
%Below we state our theorem and briefly discuss the ideas.

\begin{theorem}
\label{thm:additive}
\sloppy
$\forall k \in [n-1]$, any additive auction instances $\hat{\cI} = (N, M, \cD)$ and $\cI = (N, M, \cD, G)$ where $G$ is $k$-informed, $\cM'_{IEA}$ is 2-DST and
$Rev(\cM'_{IEA}(\cI)) \geq \max\{\frac{1}{11}, \frac{\tau_k}{6+2\tau_k}\} OPT(\hat{\cI})$.
\end{theorem}

%It is worth pointing out that when $k\geq 1$, 
%the information elicitation mechanism for the bundling part 
%recovers {\em all} distributions in $\hat{\cI}$,
%thus its revenue equals the corresponding Bayesian revenue, which is lower-bounded in \cite{cai2016duality}.
%For the individual sale part, information elicitation may be done by randomized partition depending on the value of $k$.

The same paradigm can be applied to arbitrary knowledge graphs. However, when not everything is known, the bundling part does not recover all distributions in $\hat{\cI}$ and a more complex analysis is needed to lower-bound its revenue, essentially still via adjusted revenue.

%
%\vspace{-10pt}
%
%\begin{proof}[Proof sketch]
%Mechanism $\cM_{CSIM}$ is the simplest: on each item, it is similar to $\cM_{CSUD}$ and uses random sampling to disentangle the players' values from their knowledge. It is a $\tau_k$-approximation to $IM$.
%Mechanism $\cM_{CSBVCG}$  has a different structure and does not need random sampling:
%the  $BVCG$ mechanism run on the knowledge reported by {\em all} players is automatically 2-DST,
%because
%the allocation and the price for a player $i$
%%for each player $i$, $M_i$, $e_i$ and $p_i$
%depend only on the reported values and
%the knowledge about him, not the distributions of the others.
%%Thus reporting his true knowledge can never hurt a player.
%When everything is known by somebody, the revenue of $\cM_{CSBVCG}$ matches that of $BVCG$.
%Mechanism $\cM_{CS1LA}$ has a similar structure
%to $\cM_{CSBVCG}$: the 1-Lookahead mechanism can also be run on the knowledge reported by all players.
%
%
%By \cite{cai2016duality}, $\max{\{Rev(IM),Rev(BVCG)\}}\geq \frac{OPT}{8}$.
%However, we cannot compute the  crowdsourced mechanisms' expected revenue and choose the better one, because we do not know the distributions. Moreover, computing the expected revenue using the players' reported knowledge and then choosing the better one may cause the players to lie about their knowledge.
%Thus we choose from $\cM_{CSIM}$, $\cM_{CS1LA}$ and $\cM_{CSBVCG}$ randomly, with the  probabilities depending
%on~$k$.
%\end{proof}

\subsection{Single-Good Auctions with 2-Connected Knowledge Graphs}
\label{sec:warm:myerson}

As we have seen, the amount of revenue our mechanisms generate increases with $k$, the amount of knowledge in the system.
If the knowledge graph is only $k$-informed for some small $k$, but reflects certain combinatorial structures,
good revenue may also be generated by leveraging such structures.
In this subsection
%\vspace{-5pt}
%In the last part on crowdsourced Bayesian auctions when ``everything is known by somebody'',
we consider single-good auctions, so a player's value $v_i$ is a single number rather than a vector.
Following Lemma \ref{lem:add:IM} in Appendix \ref{app:known:additive},
for any $k\geq 1$, when there is a single item and the knowledge graph is $k$-informed, mechanism $\cM_{IEIM}$ is a $\tau_k$-approximation to
the optimal Bayesian mechanism of Myerson \cite{myerson1981optimal}.
Below
%by leveraging the structure of the latter,
we construct an information elicitation mechanism that is {\em nearly optimal}
under a natural structure of the knowledge graph.

More precisely, recall that a directed graph is {\em strongly connected} if there is a directed path from any node~$i$ to any other node~$i'$.
Intuitively, in a knowledge graph this means that for any two players Alice and Bob, Alice knows a guy who knows a guy ... who knows Bob.
Also recall that a directed graph is {\em 2-connected} if it remains strongly connected after removing any single node and the adjacent edges.
In a knowledge graph, this means there does not exist a crucial player as an ``information hub'', without whom the players will split into two parts, with one part having no information about the other.
%It is easy to see that strong connectedness and 2-connectedness respectively imply 1-informedness and 2-informededness, but not vice versa. In fact, a graph of $n$ nodes can be $(\lfloor\frac{n}{2}\rfloor-1)$-informed without being connected.
It is easy to see that a knowledge graph being strong connected (or 2-connected respectively) 
implies it being 1-informed (or 2-informed respectively), but not vice versa. In fact, a graph of $n$ nodes can be $(\lfloor\frac{n}{2}\rfloor-1)$-informed without being connected.

When the knowledge graph is 2-connected, we construct the {\em information elicitation Myerson} mechanism $\cM_{IEM}$ in Mechanism~\ref{alg:myerson}.
Recall that Myerson's mechanism maps
each player $i$'s reported value~$b_i$ to the {\em (ironed) virtual value}, $\phi_i(b_i; \cD_i)$.
%with $\phi_i$ monotone in~$b_i$.
It runs the second-price mechanism with reserve price~0 on virtual values and maps the resulting ``virtual price''
back to the winner's value space, as his price.

\begin{algorithm}[htbp]
\floatname{algorithm}{Mechanism}
 \caption{\hspace{-4pt} $\cM_{IEM}$}
 \label{alg:myerson}
 \begin{algorithmic}[1]
 \STATE  Each player $i$ reports a value $b_i$ and a knowledge $K_i = (\cD^i_j)_{j\in N\setminus\{i\}}$.

 \STATE  Randomly choose a player $a$, let $S = \{j \ | \  \cD^a_j\neq \bot\}$, $N' = N\setminus(\{a\}\cup S)$, and
  $\cD'_j = \cD^a_j \ \forall j\in S$. \label{csm:2}

 \STATE If $S = \emptyset$, the item is unsold, the mechanism sets price $p_i=0$ for each $ i\in N$ and stop here. \label{step3}

 \STATE Set $i^* = \argmax_{j\in S} \phi_j(b_j; \cD'_j)$, with ties broken lexicographically.

 \WHILE{$N'\neq \emptyset$}
   \STATE Set $S' = \{j \ | \ j\in N', \ \exists i'\in S\setminus\{i^*\} \mbox{ s.t. } \cD^{i'}_j\neq \bot\}$. \label{step6}

   \STATE If $S'=\emptyset$ then go to Step \ref{step11}. \label{step7}

   \STATE For each $j\in S'$, set $\cD'_j = \cD^{i'}_j$,
   where $i'$ is the first player in $S\setminus\{i^*\}$ with $\cD^{i'}_j\neq \bot$.

   \STATE Set $S = \{i^*\}\cup S'$ and $N' = N'\setminus S'$.

   \STATE Set $i^* = \argmax_{j\in S} \phi_j(b_j; \cD'_j)$,
   with ties broken lexicographically. \label{step10}

 \ENDWHILE

 \STATE Set $\phi_{second} = \max_{j\in N \setminus (\{a, i^*\}\cup N')} \phi_j(b_j; \cD'_j)$
  and the price $p_i = 0$ for each player $i$. \label{step11}

 \STATE If $\phi_{i^*}(b_{i^*}; \cD'_{i^*}) < 0$ then the item is unsold;
 otherwise, the item is sold to player $i^*$ and
 $p_{i^*} = \phi^{-1}_{i^*}(\max\{\phi_{second}, 0\}; \cD'_{i^*})$.\label{step13a}

% \STATE Finally, for each pair of players $i$ and $i'$ such that
% $\cD^i_{i'}\neq \bot$, set $p_i = p_i - \frac{\epsilon}{2n^2}
% \cdot BSR(\cD^i_{i'}, b_{i'})$. \label{step13}

\end{algorithmic}
\end{algorithm}

To help understanding our mechanism, we illustrate in Figure~\ref{fig:IEM} of Appendix \ref{app:figure:2con} the sets of players involved in the first round.
We have the following theorem, proved in Appendix \ref{app:proofwarm}.

\begin{theorem}
\label{thm:myerson}
For any single-good auction instances $\hat{\cI} = (N, M, \cD)$ and $\cI = (N, M, \cD, G)$ where $G$ is 2-connected,
$\cM_{IEM}$ is 2-DST and
$Rev(\cM_{IEM}(\cI)) \geq (1-\frac{1}{n})OPT(\hat{{\cI}})$.
\end{theorem}

\begin{proof}[Proof ideas]

The mechanism disentangles the use of the players' values and
the use of their knowledge, but in a more subtle and stingy way than randomized partition.
Indeed,
when computing a player's virtual value in Step \ref{step10},
his knowledge has not been used yet.
If he is player $i^*$ then his knowledge will not be used in the next round either.
Only when a player is removed from $S$ ---that is, when it is guaranteed that he will not get the item,
will his knowledge be used. This is why it never hurts a player to report his true knowledge.

Now consider the revenue when the players report their true values and true knowledge.
%{\it Proof Sketch of Theorem \ref{thm:myerson}.}
%To see where  2-connectedness comes into play,
Note that $|S|\geq 2$ in Step \ref{csm:2} due to 2-connectedness, so the mechanism does not stop in Step~\ref{step3}.
In the iterative steps, because player $i^*$ is excluded from the set of reporters,
we need that there is still a reporter who knows a distribution for players in~$N'$:
that is, there is an edge from $N\setminus (N'\cup \{i^*\})$ to~$N'$, and player $i^*$ is not an ``information hub'' between $N\setminus N'$ and $N'$. This is again guaranteed by  2-connectedness (note that strong connectedness alone is not enough).
%Since this needs to be true in each round, we essentially need that no player is an information hub between any two subsets of players, thus graph is 2-connected.
Accordingly, $\cM_{IEM}$ does not stop until $N'=\emptyset$ and all players' distributions have been reported (excluding, perhaps, that of player $a$).
Therefore $\cM_{IEM}$ recovers $\hat{\cI}$ and runs Myerson's mechanism on it after randomly excluding a player $a$, and the revenue guarantee follows.
\end{proof}

%\vspace{-15pt}
%\paragraph*{Remark.}
If the seller knows at least two distributions, the mechanism can use him as the starting point and the revenue will be exactly $OPT$.
%; but our result holds even when the seller knows nothing.
Since no information elicitation mechanism can be a $(\frac{1}{2}+\delta)$-approximation for any constant $\delta>0$
when $n=2$ \cite{azar2012crowdsourced}, our result is tight.
Interestingly,
after obtaining our result, we found that 2-connected graphs
%(though not as knowledge graphs)
have been explored several times in the game theory literature \cite{bach2014pairwise, renault1998repeated}, for totally different problems.

\medskip
For additive auctions, when the knowledge graphs are 2-connected, instead of using mechanism
 $\cM_{IE1LA}$ or $\cM_{IEIM}$, one can use $\cM_{IEM}$
 for each item $j$.
We thus have the following corollary, where the mechanism $\cM''_{IEA}$ runs $\cM_{IEM}$ with probability $\frac{3}{4}$ and $\cM_{IEBVCG}$ with probability~$\frac{1}{4}$.

% in additive auction if the knowledge graph of each item is 2-connected.
%\vspace{-5pt}
\begin{corollary}\label{col:additive}
For any additive auction instances $\hat{\cI} = (N, M, \cD)$ and $\cI = (N, M, \cD, G)$ where each $G_j$ is 2-connected,
%When the knowledge graphs are 2-connected,
mechanism $\cM''_{IEA}$ is 2-DST and
$\mathop\mathbb{E}_{\cD}Rev(\cM''_{IEA}({\cI}))
\geq \frac{1}{8}(1-\frac{1}{n})
OPT(\hat{{\cI}})$.
\end{corollary}

It would be very interesting to see if other combinatorial structures of knowledge graphs
can be leveraged in information elicitation mechanisms and facilitate the aggregation of the players' knowledge.

\newpage

\appendix

\noindent
\section*{Appendix}

\section{Proofs for Section \ref{sec:k=0}}
\label{app:proofpartial}

\subsection{Proof of Theorem \ref{thm:unit}}\label{app:unknown:unit}

%We first prove the following lemma.
%\vspace{-10pt}
\paragraph*{Lemma \ref{ud:truthful}.} (restated) {\em
Mechanism ${\cal M}_{IEUD}$ is 2-DST.
}

\begin{proof}
We start with the first requirement in the solution concept: it is the best for a player to report his true values, no matter what knowledge he reports and what strategies the others use.
%More precisely, we have the following.
\begin{claim}\label{clm:ud:2dst:1}
For any player $i$, true valuation $v_i$,
%fixing any reported knowledge $K_i$, it is {\em (weakly) dominant} for $i$ to report his true valuation $v_i$. That is, for
 valuation $b_i$, knowledge $K_i$, and strategy subprofile $s_{-i} = (b_j, K_j)_{j\neq i}$ of the other players,
% in mechanism $\cM_{IEUD}$ we have
$\bE_{\cM_{IEUD}}u_i((v_i, K_i), s_{-i}) \geq \bE_{\cM_{IEUD}}u_i((b_i, K_i), s_{-i})$,
where the expectation is taken over the mechanism's random coins.
%For any player $i$, knowledge $K_i$ and true valuation $v_i$, fixing the second component of $i$'s strategy to be $K_i$, it is dominant for $i$ to report $v_i$.
\end{claim}

\begin{proof}
%Arbitrarily fix a valuation $b_i=(b_{ij})_{j\in M}$ for $i$ and a strategy $s_{i'} = (b_{i'}, K_{i'})$ for each player $i'\neq i$.
%We need to compare
%$\bE_{\cM_{IEUD}} u_i(v_i, K_i)$ and $\bE_{\cM_{IEUD}} u_i(b_i, K_i)$,
% where  the expectation is taken over the mechanism's random coins.%
%\footnote{Strictly speaking, each $s_{i'}$ should depend on $v_{i'}$ and the expectation should also be taken over $\cD_{-i}$. However,
%reporting $v_i$ is dominant for player~$i$ no matter what $v_{-i}$ is, thus there is no need to consider $\cD_{-i}$.}

If player~$i$ is not in $N_3$ given the mechanism's
randomness and the reported knowledge of all players,
then his reported valuation is never used to compute
his allocation or price,
and he gets the same utility for reporting any valuation.
Thus $u_{i}((v_{i},K_{i}), s_{-i})=u_{i}((b_{i},K_{i}), s_{-i})$ conditional on $i\notin N_3$.

If player $i$ is in $N_3$, then his utility is determined by $\cM_{UD}$.
If $\cD'_{ij}$ is defined to be 0 with probability 1 in Step \ref{step12},
then $i$'s reported value for item $j$ is not given to $\cM_{UD}$ as input,
and $i$ gets the same utility for reporting any value for $j$, including $v_{ij}$.
Moreover, because $i$ does not get such an item~$j$,
his utility is the same as an imaginary player $\hat{i}$ whose valuation is the same as $i$'s, except that the true value of $\hat{i}$ for $j$ is 0.
Since $\cM_{UD}$ is DST, no matter what the distributions are and what values the other players report,
it is the best for $\hat{i}$ to report his true valuation.
Accordingly, it is the best for $i$ to report his true valuation $v_i$ as well.
That is,
$u_i((v_i, K_i), s_{-i}) \geq u_i((b_i, K_i), s_{-i})$ conditional on $i\in N_3$.
%and the inequality is strict for some strategy subprofiles~$s_{-i}$.

Combining these two cases,
%$\bE_{\cM_{IEUD}} u_i(v_i, K_i)\geq \bE_{\cM_{IEUD}} u_i(b_i, K_i)$ for all $s_{-i}$ and the inequality is strict for some $s_{-i}$. Thus, fixing $K_i$ in player $i$'s strategy,
%it is dominant for $i$ to report~$v_i$, and
Claim \ref{clm:ud:2dst:1} holds.
\end{proof}

We now prove the second requirement in the solution concept: given that all players report their true valuations, it is the best for a player to report his true knowledge.

\begin{claim}\label{clm:ud:2dst:2}
For any player $i$, true valuation $v_i$, true knowledge $K_i$,
% $i$'s true knowledge by $TK_i$, then
% = \times_{i'\neq i, j\in [m]} \cD^i_{i'j}$, where $\cD^i_{i'j} = \cD_{i'j}$ if edge $(i, i')$ is in the knowledge graph for item $j$ and $\cD^i_{i'j} = \bot$ otherwise.
%for any true valuation $v_i$,
knowledge $K'_i$, and knowledge subprofile $K'_{-i}(v_{-i}) = (K'_j(v_j))_{j\neq i}$ of the other players,
where each $K'_j(v_j)$ is a function of player $j$'s true valuation~$v_j$,
$\bE_{v_{-i}\sim \cD_{-i}} u_i((v_i, K_i), (v_{-i}, K'_{-i}(v_{-i}))) \geq \bE_{v_{-i}\sim \cD_{-i}} u_i((v_i, K'_i), (v_{-i}, K'_{-i}(v_{-i})))$.
%
%Given that all players report their true values,
%it never hurts a player $i$ to report truthfully
%$K_i = (\cD^i_{i'j})_{{i'}\neq i, j\in M}$ as defined by the knowledge graphs:
%that is,
%for each item $j$,
%$\cD^i_{i'j} = \cD_{i'j}$ for all~$i'$ such
%that $(i, i')\in G_{j}$, and $\cD^i_{i'j} = \bot$ otherwise.
\end{claim}

\begin{proof}
If player~$i$ is in $N_1$, then
he is guaranteed to get no item and pay 0, so his utility is 0 no matter which knowledge he reports.
%Thus $u_{i}(v_{i},K_{i})=u_{i}(v_{i},K'_{i})$ in this case.
If player $i$ is in $N_2$, then his knowledge is never used, and
he again gets the same utility no matter which knowledge he reports.
%Thus we also have $u_{i}(v_{i},K_{i})=u_{i}(v_{i},K'_{i})$ in this case.
%Combining the above two cases,
Thus $\bE_{v_{-i}\sim \cD_{-i}} u_i((v_i, K_i), (v_{-i}, K'_{-i}(v_{-i}))) = \bE_{v_{-i}\sim \cD_{-i}} u_i((v_i, K'_i), (v_{-i}, K'_{-i}(v_{-i})))$
and Claim \ref{clm:ud:2dst:2} holds.
\end{proof}
Lemma \ref{ud:truthful} follows directly from Claims \ref{clm:ud:2dst:1} and \ref{clm:ud:2dst:2}.
\end{proof}

%\vspace{-10pt}
\paragraph*{The COPIES setting.}
Before analyzing the revenue of mechanism ${\cal M}_{IEUD}$, we first recall
the COPIES setting for reducing multi-parameter settings to single-parameter settings \cite{chawla2010multi}.
Given a unit-demand auction instance $\hat{\cal I}=(N, M, \cD)$, the corresponding COPIES instance is constructed as follows.
We make $m$ copies for each player,
called {\em player copies}, and
denote the resulting player
set by $N^{CP}=N\times M$. We make $n$
copies for each item, called {\em item copies},
and denote the resulting item set by $M^{CP}=M\times N$.
Each player copy $(i, j)$
has value $v_{ij}\sim \cD_{ij}$ for the item copy  $(j,i)$, and
0 for all the other item copies.

The set of feasible allocations in the original unit-demand auction naturally defines the set of feasible allocations in the COPIES auction:
for any feasible allocation $A$ in the original setting, if player $i$ gets item~$j$, then in the corresponding allocation in the COPIES setting, player $i$'s copy~$j$ gets item $j$'s copy $i$, and all other copies of $i$ get nothing.
Since in the original setting each item is sold to at most one player, in the COPIES setting, for all copies of the same item, at most one of them is sold.
Moreover, since each player gets one item in the original setting, in the COPIES setting, for all copies of the same player, at most one of them gets an item copy.
We denote by $\hat{\cI}^{CP} = (N^{CP}, M^{CP}, \cD^{CP})$ the COPIES instance,
with $\cD^{CP} = \times_{(i,j)\in N^{CP}} \cD_{ij}$.

%Intuitively, $\hat{\cal I}^{CP}$ involves more
%competition among the player copies:
%copies of the same original player are different players now, and their bids may be used to decide the prices for each other.
%Thus the optimal revenue for COPIES is not too small
%compared with the optimal revenue for the original unit-demand instance,
%as shown by the following lemma.

%\begin{lemma} \label{lem:rev_4_cai:ud} {\cite{cai2016duality}}
%%For unit-demand bidders, the optimal
%%revenue is upper bounded by $4OPT(\hat{\cI}^{CP})$.
%$OPT(\hat{\cI}) \leq 4OPT(\hat{\cI}^{CP})$.
%\end{lemma}

\vspace{-10pt}
\paragraph*{The projected setting.}
Next, we consider the optimal Bayesian revenue for the COPIES setting when ``projected'' to smaller instances.
%show that the optimal Bayesian mechanism for the COPIES setting satisfies the {\em projection lemma}.
Let $NM\subseteq N\times M$ be a subset of player-item pairs
and $N'$ be $NM$ projected to $N$.
The unit-demand instance $\hat{\cI}_{NM} = (N', M, \cD'_{N'})$, referred to as {\em $\hat{\cI}$ projected to $NM$},
is such that $\cD'_{ij} = \cD_{ij}$ if $(i, j)\in NM$,
and $\cD'_{ij}\equiv 0$ otherwise.
%When $NM$ can be represented as $N'\times M'$, we may
%also denote the projected instance by $\hat{\cI}_{N',M'}$.
%If $N' = N$ or $M' = M$, we simplify the notation by using $\hat{\cI}_{M'}$ or $\hat{\cI}_{N'}$ instead.

Let
$\hat{\cI}^{CP}_{NM} = (N'^{CP}, M^{CP}, \cD'^{CP}_{N'})$
be the COPIES instance corresponding to $\hat{\cI}_{NM}$.
It can also be considered as {\em $\hat{\cI}^{CP}$ projected to $NM$}.
That is, when projecting a COPIES instance to a set of player-item pairs,
we still want the resulting instance to be a COPIES instance, thus we patch it up with
the missing player-item pairs but with values constantly 0.
The relations of these instances are illustrated in Figure \ref{fig:proj}. We are interested in
the optimal Bayesian revenue under $\hat{\cI}^{CP}_{NM}$, $OPT(\hat{\cI}^{CP}_{NM})$.

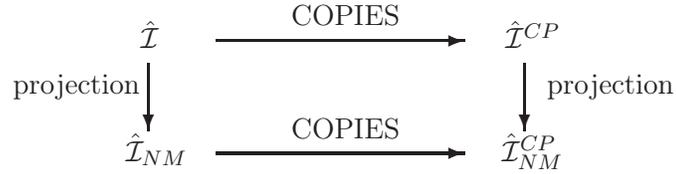
\begin{figure}[htbp]
\begin{center}
\setlength{\unitlength}{1cm}
\thicklines
\begin{picture}(8,2)

\put(-0.2,0.9){projection}
\put(1.5,1.5){$\hat{\cI}$}
\put(1.3,0){$\hat{\cI}_{NM}$}
\put(1.6,1.3){\vector(0,-1){0.9}}

\put(2.5,1.6){\vector(1,0){3.3}}

\put(2.5,0.1){\vector(1,0){3.3}}

\put(6.9,0.9){projection}
\put(6.35,1.5){$\hat{\cI}^{CP}$}
\put(6.3,0){$\hat{\cI}_{NM}^{CP}$}
\put(6.6,1.3){\vector(0,-1){0.9}}

\put(3.5,1.8){COPIES}
\put(3.5,0.3){COPIES}

\end{picture}
\caption{Relations of the COPIES and the projected instances}\label{fig:proj}
\end{center}
\end{figure}

Let $OPT(\hat{\cI}^{CP})_{NM}$ be the revenue of the optimal Bayesian mechanism for $\hat{\cI}^{CP}$
obtained from $NM$:
that is, for any pair $(i, j)\notin NM$,
%the $j$-th copy of player $i$ gets nothing and pays 0,
the price paid by the $j$-th copy of player $i$ is not counted,
even though this player copy may get the $i$-th copy of item $j$ according to the
optimal mechanism for $\hat{\cI}^{CP}$.
%$OPT(\hat{\cI}^{CP})$.
We have the following lemma, where we explicitly write out the distributions for different Bayesian instances.

\vspace{-10pt}
\paragraph*{Lemma \ref{lem:proj}.} (The projection lemma, restated) {\em
For any $\hat{\cI}$ and $NM\subseteq N\times M$,
$$\mathop\bE\limits_{\cD'^{CP}_{N'}} OPT(\hat{\cI}^{CP}_{NM})\geq \mathop\bE\limits_{\cD^{CP}} OPT(\hat{\cI}^{CP})_{NM}.$$}

%\begin{lemma} [The projection lemma (Formal)]\label{lem:proj}
%For any $NM\subseteq N\times M$,
%$$\mathop\bE\limits_{\cD'^{CP}_{N'}} OPT(\hat{\cI}^{CP}_{NM})\geq \mathop\bE\limits_{\cD^{CP}} OPT(\hat{\cI}^{CP})_{NM}.$$
%\end{lemma}

\begin{proof}
%\noindent
Given the instance $\hat{\cal I}_{NM}^{CP}$, consider a Bayesian mechanism $\cM'$ as follows:
\begin{itemize}
\item
Patch up the instance to be exactly $\hat{\cI}^{CP}$, including changing the distribution $\cD'_{ij}$ from 0 to $\cD_{ij}$ for
any $(i, j)$ with $i\in N'$ and $(i, j)\notin NM$;

\item
For any $(i, j)\in NM$, let $i$'s copy $j$ report his value, while for any $(i, j)\not\in NM$, sample the value of $i$'s copy $j$ from $\cD_{ij}$;

\item
Run the optimal DST Bayesian mechanism
 on $\hat{\cI}^{CP}$, with the reported and the sampled values;%
\footnote{Since COPIES is a single-parameter setting, Myerson's mechanism is both the optimal BIC mechanism and the optimal DST mechanism here. In particular, the optimal BIC revenue equals the optimal DST revenue.}

\item
Project the resulting outcome to $NM$.

\end{itemize}

The key is to show that $\cM'$ is a DST Bayesian mechanism for $\hat{\cal I}_{NM}^{CP}$.
Indeed, for any $(i, j)$ such that
$i\in N'$ and $(i, j)\notin NM$,
the value of $i$'s copy $j$ in $\hat{\cal I}_{NM}^{CP}$ is 0 with probability 1,
his reported value is not used by $\cM'$,
and at the end this player copy gets nothing and pays 0.
Therefore it is dominant for this player copy to report his true value 0.
For any $(i, j)\in NM$,
the utility of $i$'s copy~$j$ under $\cM'$ is the same as that under the optimal DST Bayesian mechanism
on $\hat{\cI}^{CP}$. As it is dominant for this player copy to report his true value in the latter,
so is it in $\cM'$.
Accordingly,
$$\mathop\mathbb{E}\limits_{\cD'^{CP}_{N'}} OPT(\hat{\cal I}_{NM}^{CP})
\geq \mathop\mathbb{E}\limits_{\cD'^{CP}_{N'}} Rev(\cM'(\hat{\cal I}_{NM}^{CP})),$$
by the definition of $OPT$.
By construction we have
$$\mathop\mathbb{E}\limits_{\cD'^{CP}_{N'}} Rev(\cM'(\hat{\cal I}_{NM}^{CP})) =
\mathop\mathbb{E}\limits_{\cD^{CP}}
OPT(\hat{\cI}^{CP})_{NM},$$
thus Lemma \ref{lem:proj} holds.
\end{proof}

Note that the projection lemma is not true for non-COPIES settings in general,
because the corresponding mechanism $\cM'$
is not DST.
Indeed, when the same player $i$ has $(i, j)\in NM$ and $(i, j')\notin NM$ for some
items $j$ and $j'$,
he prefers
receiving $j$ to receiving $j'$ in the patched-up auction,
even if the former leads to a smaller utility: his projected utility will be 0 under the latter.

Now we are ready to  finish the proof of Theorem \ref{thm:unit}.
% ready to analyze the revenue of mechanism $\cM_{IEUD}$.

\paragraph*{Theorem \ref{thm:unit}.} (restated) {\em
Mechanism $\cM_{IEUD}$ for unit-demand auctions is 2-DST and,
for any instances $\hat{\cI} = (N, M, \cD)$ and
$\cI = (N, M, \cD, G)$,
% and $\cI' = (N, M, \cD')$ with $\cD'$ being $\cD$ projected on $G$,
$\bE_{v\sim \cD} Rev(\cM_{IEUD}(\cI)) \geq \frac{OPT_K(\cI)}{96}$.
}
\medskip

\begin{proof}
Letting $\cI' = (N, M, \cD')$ be the Bayesian instance where $\cD'$ is $\cD$ projected on the knowledge graph $G$,
the COPIES instance $\cI'^{CP}$ is respectively defined.
%[Proof of Theorem \ref{thm:unit}]
%Recall that $q=\frac{1}{2}$.
Following Lemma \ref{ud:truthful} and by the definition of $OPT_K$, it remains to show
$$\bE_{v\sim \cD} Rev(\cM_{IEUD}(\cI)) \geq \frac{OPT(\cI')}{96}.$$
%Note that for player $i$ and item $j$ with $\cD_{ij}$ unknown by any player,
%we can view it as if $\cD_{ij} \equiv 0$ and every player knows that.
%Then %by adding the seller to the player set $N_1$,
For any player $i$ and item $j$
with $\cD_{ij}$ known by some other players, we have
\begin{eqnarray*}
\Pr(\cD_{ij} \text{ is reported in the mechanism}) &= & \Pr(i\in N_2)\Pr(\exists i'\in N_1 \mbox{ s.t. } (i', i)\in G_j \ | \  i \in N_2) \nonumber\\
&\geq &\frac{1}{2}\cdot \frac{1}{2} = \frac{1}{4},
\end{eqnarray*}
where the inequality is because there exists at least one player $i'$ with $(i', i)\in G_j$,
%we only consider player-item pair $i,j$ with $\cD_{ij}$ known by some other players and
and players $i$ and $i'$ are partitioned independently.
Below we use $NM_3$ to denote the set of player-item pairs whose distribution is reported in the mechanism: that is, the set of players $N_3$ together with their reported items.
Accordingly,
\begin{equation}\label{eq:ud:prob}
\Pr((i, j)\in NM_3) \geq \frac{1}{4}.
\end{equation}
Also, we use $\cD_{NM_3} = \times_{(i,j)\in NM_3} \cD_{ij}$ and $v_{NM_3} = (v_{ij})_{(i,j)\in NM_3}$ to denote the vector of distributions and the vector of true values for $NM_3$, respectively.
Let $\hat{\cI}_{NM_3} = (N_3, M, \cD'_{N_3})$ be the unit-demand instance given to $\cM_{UD}$ in Step \ref{step13_1}.
Note that it is exactly $\cI'$ projected to $NM_3$: that is,
$\cD'_{ij} = \cD_{ij}$ for any $(i, j)\in NM_3$
and $\cD'_{ij}\equiv 0$ for any $(i, j)\notin NM_3$.
%Moreover, let $R$ be the total reward given to the players using Brier's scoring rule in Step \ref{step:scoring rule:ud}.
%Since $BSR$ is
%bounded in $[0, 2]$, the reward each player $i$ gets in this step is no more than $\frac{\epsilon}{n}$ and
%the total reward given to the players is no more than $\epsilon$.
Accordingly,

\begin{eqnarray}
\hspace{20pt} \label{equ:4_0}
\mathop\bE\limits_{v\sim \cD} Rev({\cal M}_{IEUD}({\cal I})) & = &  \mathop\mathbb{E}\limits_{NM_{3}} \ \mathop\mathbb{E}\limits_{v_{NM_3}\sim \cD_{NM_3}}Rev({\cal M}_{UD}(\hat{\cal I}_{NM_{3}})).
%& \geq & \mathop\mathbb{E}\limits_{NM_{3}}\ \mathop\mathbb{E}\limits_{v_{NM_3}\sim \cD_{NM_3}} Rev({\cal M}_{UD}(\hat{\cal I}_{NM_{3}})) - \epsilon. \label{equ:4_0}
\end{eqnarray}

Let $\hat{\cI}^{CP}_{NM_3} = (N_3^{CP}, M^{CP}, \cD'^{CP}_{N_3})$ be the COPIES instance
corresponding to $\hat{\cI}_{NM_3}$.
Following~\cite{kleinberg2012matroid}, given any set $NM_3$,
%there exists a DST Bayesian mechanism $\cM^{CP}$ in the COPIES setting,
%which achieves less
the revenue of $\cM_{UD}$ under the unit-demand Bayesian instance $\hat{\cI}_{NM_3}$
%for the unit-demand instance,
is a $6$-approximation to the optimal revenue for the COPIES instance. That is, for any $NM_3$,
\begin{eqnarray}\label{equ:4_1}
\mathop\mathbb{E}\limits_{v_{NM_3} \sim \cD_{NM_3}}Rev({\cal M}_{UD}(\hat{\cal I}_{NM_{3}})) &  \geq & \frac{1}{6}\mathop\mathbb{E}\limits_{\cD'^{CP}_{N_3}} OPT(\hat{\cal I}_{NM_{3}}^{CP}).
\end{eqnarray}
Combining Inequalities \ref{equ:4_0} and \ref{equ:4_1}, we have
\begin{eqnarray}
\label{eq:ud:CPrevenue}
& & \mathop\bE\limits_{v\sim \cD} Rev({\cal M}_{IEUD}({\cal I})) \geq
\frac{1}{6} \mathop\mathbb{E}\limits_{NM_{3}}\ \mathop\mathbb{E}\limits_{\cD'^{CP}_{N_3}}OPT(\hat{\cal I}_{NM_{3}}^{CP}).
\end{eqnarray}
By Lemma \ref{lem:proj},
$$
%\label{eq:ud:benchmarkM}
\mathop\mathbb{E}\limits_{\cD'^{CP}_{N_3}} OPT(\hat{\cI}_{NM_{3}}^{CP}) \geq
\mathop\mathbb{E}\limits_{\cD^{CP}}
OPT(\cI'^{CP})_{NM_3},
$$
thus
\begin{equation}
\label{eq:ud:benchmarkM}
\mathop\mathbb{E}\limits_{NM_{3}}\ \mathop\mathbb{E}\limits_{\cD'^{CP}_{N_3}} OPT(\hat{\cal I}_{NM_{3}}^{CP}) \geq \mathop\mathbb{E}\limits_{NM_{3}}\
\mathop\mathbb{E}\limits_{\cD^{CP}}
OPT(\cI'^{CP})_{NM_3}.
\end{equation}

Let $P_{ij}(OPT(\cI'^{CP}))$ be the price paid by player $i$'s
copy $j$ under the optimal mechanism for~$\cI'^{CP}$.
We can rewrite the
right-hand side of Equation~\ref{eq:ud:benchmarkM} as follows.

\begin{eqnarray}
\label{eq:ud:expectation}\hspace{30pt}
& & \mathop\mathbb{E}\limits_{NM_{3}}\
\mathop\mathbb{E}\limits_{\cD^{CP}}
OPT(\cI'^{CP})_{NM_3} =
\mathop\mathbb{E}\limits_{NM_{3}}\ \mathop\mathbb{E}\limits_{\cD^{CP}}
\sum\limits_{(i,j)\in NM_{3}}
P_{ij}(OPT(\cI'^{CP})) \nonumber \\
&= &  \mathop\mathbb{E}\limits_{\cD^{CP}}\
\mathop\mathbb{E}\limits_{NM_{3}}
\sum\limits_{(i, j)\in NM_{3}}
P_{ij}(OPT(\cI'^{CP}))  \nonumber \\
& = & \mathop\mathbb{E}\limits_{\cD^{CP}}
\sum\limits_{(i, j) \in N\times M} \Pr((i, j) \in NM_{3})\cdot
P_{ij}(OPT(\cI'^{CP})),
\end{eqnarray}
where the first equality is by the definition of the projection,
the second is because
sampling from~$\cD^{CP}$ is done independently from
%the distribution of $v_{N^{CP}}$ does not depend on
$NM_3$,
and the third is
%by linearity of expectation and
because $P_{ij}(OPT(\hat{\cal I}^{CP}))$ does not depend on~$NM_3$.
We can further lower-bound the last term of Equation \ref{eq:ud:expectation}
as follows:

\begin{eqnarray}
\label{eq:ud:prob_k}\hspace{20pt}
& & \mathop\mathbb{E}\limits_{\cD^{CP}}
\sum\limits_{(i, j) \in N\times M} \Pr((i, j) \in NM_{3})\cdot
P_{ij}(OPT(\cI'^{CP})) \nonumber  \\
& = & \mathop\mathbb{E}\limits_{\cD^{CP}}
\sum\limits_{(i, j) \in N\times M: \exists i' \mbox{\scriptsize{ s.t. }} (i',i)\in G_j} \Pr((i, j) \in NM_{3})\cdot
P_{ij}(OPT(\cI'^{CP})) \nonumber  \\
&\geq& \frac{1}{4}
\mathop\mathbb{E}\limits_{\cD^{CP}}
\sum\limits_{(i, j) \in N\times M: \exists i' \mbox{\scriptsize{ s.t. }} (i',i)\in G_j} P_{ij}(OPT(\cI'^{CP})) \nonumber \\
&=& \frac{1}{4} OPT(\cI'^{CP}) \geq \frac{1}{16} OPT(\cI').
\end{eqnarray}
Here the first equality is because $P_{ij}(OPT(\cI'^{CP})) = 0$ for every $(i,j)$ such that $\cD_{ij}$ is unknown,
the first inequality is by Equation \ref{eq:ud:prob}, the second equality is by the definition of revenue, and the second inequality
is because $OPT(\cI'^{CP}) \geq \frac{OPT(\cI')}{4}$ by~\cite{cai2016duality}.
% Lemma \ref{lem:rev_4_cai:ud}.

Combining Equations \ref{eq:ud:CPrevenue}, \ref{eq:ud:benchmarkM}, \ref{eq:ud:expectation} and \ref{eq:ud:prob_k}, we have
$$\bE_{v\sim \cD} Rev(\cM_{IEUD}(\cI)) \geq \frac{1}{6}\cdot\frac{1}{16} \cdot OPT(\cI') = \frac{OPT(\cI')}{96},$$
and Theorem \ref{thm:unit} holds.
%\qed
\end{proof}

 %can be applied to any Bayesian mechanism $\cM$ for unit-demand auctions,

%\medskip
%\noindent
%{\bf Corollary \ref{col:copies}.} (restated)
%{\em Let $\cM$ be a DST Bayesian mechanism such that, for some $\alpha>0$ and for any
%Bayesian instance~$\hat{\cI}$,
%$$\mathop\mathbb{E}\limits_{\cD} Rev({\cal M}(\hat{\cal I}))
%\geq \alpha \mathop\mathbb{E}\limits_{\cD^{CP}} OPT(\hat{\cal I}^{CP}).$$
%Then there exists a 2-DST crowdsourced Bayesian mechanism
%that uses $\cM$ as a black-box and is a $\frac{\alpha}{16}$-approximation to $OPT$.
%}

\subsection{Proof of Theorem \ref{thm:newbvcg}}\label{app:unknown:add}

\begin{lemma}\label{lem:M'IEA2-dst}
Mechanism $\cM_{IEA}$ is 2-DST.
\end{lemma}
%\begin{proof}
%Arbitrarily fix a player $i$, a strategy subprofile
%of the other players, and a knowledge of $i$.
%For the set of items whose distributions for player $i$ are not reported by the others,
%$\cM_{IEA}$ runs second price on it.
%Thus it is dominant for $i$ to report his true values for those items.
%Otherwise, $\cM_{IEA}$ sells to player $i$
%in the same way as $BVCG$:
%using the other players' highest reported value as the reserve price for each item,
%either player $i$ gets the whole set of items for which his value passes the reserve price,
%or he gets nothing and those items are unsold to anybody.
%Following~\cite{yao2015n}, it is dominant for $i$ to report his true values given any entry fee that does not depend on his reported values,
%so truth-telling is still dominant for $i$ when the entry fee is computed based on $\cD'_i$ and the reserve prices.
%In sum, it is dominant for $i$ to report his true values.
%
%
%Furthermore, notice that a player $i$'s reported knowledge
%$K_i$ about others affects
%neither $M_i$ nor~$e_i$, nor the reserve prices for him.
%Thus, $K_i$ is only used to compute player $i$'s reward based on Brier's scoring rule,
%and $i$'s expected reward is maximized by reporting his true knowledge,
%given that all players report their true values.
%\end{proof}

\begin{proof}
The structure of a detailed proof for Lemma \ref{lem:M'IEA2-dst} will be following the two requirements in the solution concept and similar to that of Lemma \ref{ud:truthful}.
Thus we only highlight the key points here.
% the differences in the proof.
%When a
%distribution about a
%player $i$'s value for an
%item $j$ is not reported,
%%$\cM_{IEBVCG}$ chooses not to sell anything to player $i$ and charges him with price 0,
%while $\cM_{IEA}$ divides the items into two subsets $M^1_{i}$ and $M^2_{i}$,
%sells the first to $i$ using
%mechanism $Bund$, and sells the second to $i$ using the second-price mechanism.
%
%In mechanism $\cM_{IEA}$, g

For each player $i$, given the distributions and values reported by the other players,
the subsets $M^1_i$ and $M^2_i$ do not depend on player $i$'s strategy,
neither do the entry fee $e_i$ and reserve prices $(p'_j)_{j\in M^1_i}$.
As far as player $i$ is concerned, the other players' values are always taken to be $b_{-i}$, even if some of their value distributions are not reported.
%for $M^1_i$.
Computing $i$'s winning sets $M_i\cap M^1_i$ and $M_i\cap M^2_i$ is
just part of the two mechanisms, $Bund$ and second-price, again with the other players' values taken to be $b_{-i}$.
%Thus, mechanism $\cM_{IEA}$ is as if we are running mechanism $Bund$ for player $i$ on $M^1_i$
Thus, it is dominant for $i$ to
report his true values for $M^1_i$,
following the truthfulness of $Bund$;
and it is dominant
for him to report his true values for $M^2_i$,
following the truthfulness of the second-price mechanism.

%each player to report their true value since the allocation rule is obviously monotone and each player only need to pay the threshold payment.
Moreover, given that all players truthfully report their values,
for each player $i$, reporting his true knowledge never hurts him, no matter what knowledge the other players report.
%it is dominant for each $i$ to report his true knowledge:
Indeed, the winning set $M_i$ only depends on the players' reported values.
%\footnote{Note that $M_i$ may be different under instances $\hat{\cI}$ and $\cI'$.}
%In the rest of this section, we refer to $M_i$ as the winning set for player $i$ under instance $\hat{\cI}$.}
Player $i$'s reported knowledge may affect how $M$ is partitioned into $M^1_{i'}$ and $M^2_{i'}$ for another player $i'$,
but does not affect the sets $M^1_i$ and $M^2_i$, or whether he gets some items or not, or the prices he pays.
%It only affects his reward under Brier's scoring rule.
Thus $\cM_{IEA}$ is 2-DST as desired.
%and Lemma \ref{lem:M'IEA2-dst} holds.
\end{proof}

%their own knowledge, thus finishes the proof.
\paragraph{The adjusted revenue.}
To lower-bound the expected revenue of $\cM_{IEA}$,
we first introduce several important concepts
following \cite{yao2015n}.

For any
single-player Bayesian instance $\hat{\cI_i}=(\{i\},M,\cD_i)$
and any non-negative reserve-price vector $\beta_i=(\beta_{ij})_{j\in M}$,
a single-player DST Bayesian mechanism is {\em $\beta_i$-exclusive} if it
never sells an item $j$ to $i$
whenever his bid for $j$ is no larger than~$\beta_{ij}$.
Denote by $Rev^X(\hat{\cI}_i,\beta_i)$ the optimal $\beta_i$-exclusive revenue for
$\hat{\cI}_i$:
that is, the superior over the revenue of $\beta_i$-exclusive mechanisms.

For any single-player DST Bayesian mechanism, its {\em $\beta_i$-adjusted revenue} on $\hat{\cI}_i$
is its revenue minus
its social welfare generated from player-item pairs $(i, j)$ such that $i$'s bid for $j$ is no larger than~$\beta_{ij}$.
Denote by $Rev^A(\hat{\cI}_i,\beta_i)$ the
optimal $\beta_i$-adjusted revenue
for $\hat{\cI}_i$: that is,
the superior over the $\beta_i$-adjusted revenue of all
single-player DST Bayesian mechanisms.
Note that if a mechanism is $\beta_i$-exclusive, then its $\beta_i$-adjusted revenue is exactly its revenue.
%When player $i$'s true valuation is $v_i$, the $\beta_i$-adjusted revenue generated by the

%Where does the vector $\beta$ come from?
%Following \cite{yao2015n},
Given a Bayesian instance $\hat{\cI} = (N, M, \cD)$,
for each player $i$ and valuation subprofile $v_{-i}\sim \cD_{-i}$,
let $\beta_i(v_{-i}) = (\beta_{ij}(v_{-i}))_{j\in M}$ be such that $\beta_{ij}(v_{-i})= \max_{i'\neq i}v_{i'j}$
for each item~$j$. Note that we are slightly abusing notations here: each $\beta_{ij}$ is now a function rather than a value.
%The {\em optimal $\beta$-exclusive mechanism} is denoted by $BGR$.
%After seeing the valuation profile $v$,
%it uses the optimal $\beta_i(v_{-i})$-exclusive mechanism for $\hat{\cI}_i$
%to sell to each player $i$.
%Thus
%$$\bE_\cD Rev(BGR(\hat{\cI})) = \sum_i \bE_{v_{-i}\sim \cD_{-i}} \bE_{v_i\sim\cD_i} Rev^X(\hat{\cI_i}, \beta_i(v_{-i})).$$
The {\em optimal $\beta$-adjusted revenue} for~$\hat{\cI}$
is
$$\bE_{\cD} Rev^A(\hat{\cI}, \beta) = \bE_{v\sim \cD} \sum_i Rev^A(\hat{\cI}_i, \beta_i(v_{-i}))
= \sum_i \bE_{v_{-i}\sim \cD_{-i}} \bE_{v_i\sim \cD_i} Rev^A(\hat{\cI}_i,\beta_i(v_{-i})).$$
%obtained by applying the mechanism with the optimal $\beta_i(v_{-i})$-adjusted revenue for $\hat{\cI}_i$
%to sell to each player $i$, after seeing the valuation profile $v$.
When $v_{-i}$ is clear from the context, we may simply write $\beta_i$ and $\beta_{ij}$.

% be the
%for $\hat{\cI}$.

%
%By Theorem 3 of \cite{yao2015n}, in any Bayesian instance,
%letting $\beta_{ij}= \max_{i'\neq i}v_{i'j}$
%for each player $i$ and item $j$ given the true valuation profile $v$,
%%given a valuation profile~$v$ and
%%define $B(v_{-i})$ be the value vector of highest value on each item.
% $Bund$ is
% %$\beta$-exclusive %on each sub-instance $\hat{\cI}_i$
%an $8.5$-approximation to
%the optimal $\beta$-exclusive mechanism, denoted by $BGR$,
%uses the optimal $\beta_i(v_{-i})$-exclusive mechanism
%for each player $i$, after seeing the valuation profile $v$.
%Moreover, by Theorem 5.2 of \cite{yao2015n}, the revenue of $BGR$ is an $8$-approximation to
%%the sum of
%the optimal $\beta$-adjusted revenue.

Because we also consider the projected Bayesian instance $\cI' = (N, M, \cD')$,
let $v'$ be $v$ {\em projected on} the knowledge graph $G$: that is,
$v'_{ij} = v_{ij}$ if there exists
a player $i'$ with $(i',i)\in G_j$, and
$v'_{ij} = 0$ otherwise.
As $v$ is distributed according to $\cD$,
$v'$ is distributed according to $\cD'$.
Thus we sometimes directly sample $v'\sim \cD'$ rather than sample~$v$ first and then map it to $v'$.
Let $\beta'_{ij}(v'_{-i}) =  \max_{i'\neq i}v'_{i'j}$ for each player $i$ and item $j$.
The optimal $\beta'$-adjusted revenue for $\cI'$ is defined respectively:
$$\bE_{\cD'}Rev^A(\cI',\beta')= \bE_{v'\sim \cD'} \sum_i Rev^A(\cI'_i, \beta'_i(v'_{-i})) = \sum_i \bE_{v'_{-i}\sim \cD'_{-i}} \bE_{v'_i\sim \cD'_i} Rev^A(\cI'_i,\beta'_i(v'_{-i})).$$

It is important to emphasize that,
given $v_{-i}$ and $\beta_i(v_{-i})$,
the optimal $\beta_i$-exclusive revenue and
the optimal $\beta_i$-adjusted revenue
are well defined for any Bayesian instance for player $i$, whether it is $\hat{\cI}_i$
or $\cI'_i$.
In particular, we will consider
$Rev^X(\cI'_i, \beta_i)$ and $Rev^A(\cI'_i,\beta_i)$,
which are on the hybrid of $\hat{\cI_i}$ and~$\cI'_i$.
The optimal $\beta$-adjusted revenue on the hybrid of $\hat{\cI}$ and $\cI'$ are
similarly defined:
%$$\bE_\cD Rev(BGR(\cI', \beta)) = \sum_i \bE_{v_{-i}\sim \cD_{-i}} \bE_{v'_i\sim\cD'_i} Rev^X(\cI'_i, \beta_i(v_{-i}));$$
$$\bE_{\cD} Rev^A(\cI', \beta) = \sum_i \bE_{v_{-i}\sim \cD_{-i}} \bE_{v'_i\sim \cD'_i} Rev^A(\cI'_i,\beta_i(v_{-i})).$$
To highlight that in the inner expectation, player $i$'s value is $v'_i$ even if
it is obtained by first sampling $v_i\sim\cD_i$ and then projecting on $G$,
we may also write it as
%$\bE_{v_i\sim \cD_i} Rev^X(\cI'_i, \beta_i; v'_i)$ and
$\bE_{v_i\sim \cD_i} Rev^A(\cI'_i, \beta_i; v'_i)$ and
$$\bE_{\cD} Rev^A(\cI', \beta) = \sum_i \bE_{v_{-i}\sim \cD_{-i}} \bE_{v_i\sim \cD_i} Rev^A(\cI'_i,\beta_i(v_{-i}); v'_i)
= \mathop\bE\limits_{v \sim \cD} \sum_i Rev^A(\cI'_i, \beta_i(v_{-i}); v'_i).$$

Finally, denote by $IM$ the {\em individual Myerson} mechanism, which sells each item separately using Myerson's mechanism \cite{myerson1981optimal}.
The revenue of $IM$ under the projected Bayesian instance $\cI'$, denoted by $IM(\cI')$, is thus the optimal revenue by selling each item separately.
%Since the total reward given by the scoring rule is at most $\epsilon$,
%we will ignore it and
%treat $Rev(\cM_{IEA}(\cI))$ as the revenue without giving out the rewards.

\medskip
Having defined the notions and notations needed in our proof, we prove
Theorem \ref{thm:newbvcg} via the following two technical lemmas.
For each lemma,
note that
on the left-hand side the values are drawn from the original Bayesian instance, and on the right-hand side the values are drawn from the projected instance.

\begin{lemma}\label{lem:IEA:A}
$\mathop\bE\limits_{v\sim \cD} Rev(\cM_{IEA}(\cI))
\geq \frac{1}{68} \mathop\bE\limits_{v'\sim \cD'}Rev^A(\cI',\beta').$
\end{lemma}

\begin{lemma}\label{lem:IEA:B}
$\mathop\bE\limits_{v\sim \cD} Rev(\cM_{IEA}(\cI))
\geq \frac{1}{2} \mathop\bE\limits_{v'\sim \cD'} IM(\cI').$
\end{lemma}

\paragraph*{Theorem \ref{thm:newbvcg}.} (restated) {\em
Mechanism $\cM_{IEA}$ for additive auctions  is 2-DST  and,
for any instances $\hat{\cI} = (N, M, \cD)$ and
$\cI = (N, M, \cD, G)$,
% and $\cI' = (N, M, \cD')$ with $\cD'$ being $\cD$ projected on $G$,
$\bE_{v\sim \cD} Rev(\cM_{IEA}(\cI)) \geq \frac{OPT_K(\cI)}{70}$.
%$\cM_{IEA}$  is a $70$-approximation to $OPT_K(\cI)$.
%%\sloppy
%Mechanism $\cM_{IEA}$ is 2-DST for additive auctions and,
%for any instances $\hat{\cI} = (N, M, \cD)$,
%$\cI = (N, M, \cD, G)$ and $\cI' = (N, M, \cD')$ with $\cD'$ being $\cD$ projected on $G$,
%$\cM_{IEA}$  is a $70$-approximation to $OPT_K$.
}

\medskip

\begin{proof}
%[Proof of Theorem \ref{thm:newbvcg}]
By Theorem 8.1 of \cite{yao2015n},
$\bE_{v'\sim \cD'}Rev^{A}(\cI', \beta')+ \mathop\bE\limits_{v'\sim \cD'} IM(\cI') \geq OPT(\cI')$.
Combining this inequality with Lemmas \ref{lem:IEA:A} and \ref{lem:IEA:B},
we have
$$\mathop\bE\limits_{v\sim \cD} Rev(\cM_{IEA}(\cI)) \geq\frac{OPT(\cI')}{70}  = \frac{OPT_K(\cI)}{70},$$
and Theorem \ref{thm:newbvcg} holds.
\end{proof}

Below we prove the two lemmas.
%
%\begin{lemma}\label{lem:M'IEA}
%$\bE_{v\sim \cD} Rev(\cM_{IEA}(\cI)) \geq \frac{1}{70}OPT(\cI') - \epsilon$.
%\end{lemma}
\begin{proof}[Proof of Lemma \ref{lem:IEA:A}]

For each player $i$, the set $M_i^1$ of items for which $i$'s distributions are reported
and the set $M_i^2$ of items for which $i$'s distributions are not reported
are uniquely determined by the knowledge graph $G$ in $\cI$.
Let $\cI'_{i, M_i^1} = (\{i\}, M_i^1, \cD'_{i, M_i^1})$ be the single-player Bayesian instance obtained by
projecting~$\cI'$ on $i$ and $M_i^1$, with
$\cD'_{i, M_i^1} = \times_{j\in M_i^1} \cD'_{ij}$.
%In the discussion below, we will frequently use
%the projection of a Bayesian instance $\tilde{\cI}$ on a player set $N'$ and an item set $M'$,
%and we will write $\tilde{\cI}_{N', M'}$ without further explanation.
%By the construction of $\cM_{IEA}$, w
We have
\begin{eqnarray}\label{eq:bvcgi_0}
& & \bE_{v\sim \cD} Rev(\cM_{IEA}(\cI)) \nonumber  \\
&=& \bE_{v\sim \cD}\sum_i\left(Rev\left(Bund_i(\cD'_{i}, v_{i, M_i^1}, \beta_i(v_{-i}))\right)
+ \sum_{j\in M_i^2 \cap M_i}\max_{i'\neq i}v_{i'j}\right) \nonumber  \\
&=& \bE_{v\sim \cD}\sum_i\left(Rev\left(Bund_i(\cD'_{i, M_i^1}, v_{i, M_i^1}, \beta_{i, M_i^1})\right)
+ \sum_{j\in M_i^2 \cap M_i}\max_{i'\neq i}v_{i'j}\right),
% \nonumber  \\
%
%&=& \bE_{v\sim \cD}\sum_i\left(Rev\left(Bund_i(\cD'_{i, M_i^1}, v'_{i, M_i^1}, \beta_{i, M_i^1})\right)
%+ \sum_{j\in M_i^2 \cap M_i}\max_{i'\neq i}v_{i'j}\right),
%
%&=& \bE_{v\sim \cD}\sum_i Rev_i(\cM_{IEA}(\cI)) \nonumber \\
%
%& = & \bE_{v\sim \cD}\sum_i\left(Rev\left(Bund_i(\cD_{i, M_i^1}, v_{i, M_i^1}, \beta_i(v_{-i, M_i^1}))\right)
%+ \sum_{j\in M_i^2 \cap M_i}\max_{i'\neq i}v_{i'j}\right),
%
%& = & \bE_{v\sim \cD}\sum_i\left(Rev\left(Bund_i(\hat{\cI}_{M_i^1})\right)
%+ \sum_{j\in M_i^2 \cap M_i}\max_{i'\neq i}v_{i'j}\right),
\end{eqnarray}
where
%$Rev_i$ is the revenue generated from player $i$
$Bund_i$ is $Bund$ applied to player $i$ in Steps \ref{step:M'IEBVCG8} and \ref{step:M'IEBVCG9}.
The first equality is by the construction of $\cM_{IEA}$.
%Since $\beta_i$ depends on $v_{-i}$,
%we explicitly include the latter in the input of $Bund_i$.
%Indeed,
%the input given to $Bund_i$ is
%a hybrid of $\hat{\cI}$ and $\cI'$.
The second equality holds because in the execution of $Bund_i$
%although the set of items considered by $Bund_i$ when computing the entry fee and reserve prices for $i$
%is $M$,
the items in $M_i^2$ do not affect anything, as player $i$'s values for them are constantly 0
 in $\cD'_i$.
%Since $v'_{ij} = v_{ij}$ for any $j\in M_i^1$,
%Since $\cD'_{i, M_i^1} = \cD_{i, M_i^1}$,
% % and $\cI'_{M_i^1} = \hat{\cI}_{M_i^1}$,
%the third equality holds.
%and is solely from
%the Bayesian instance $\hat{\cI}_{M_i^1}$.
%,  thus last equality holds by definition.

By Theorem 6.1 of \cite{yao2015n}, for any single-player Bayesian instance for a player $i$
and any $\beta_i$,
%letting $\beta_{ij}= \max_{i'\neq i}v_{i'j}$
%for each player $i$ and item $j$ given the true valuation profile $v$,
%given a valuation profile~$v$ and
%define $B(v_{-i})$ be the value vector of highest value on each item.
 $Bund_i$ is
 %$\beta$-exclusive %on each sub-instance $\hat{\cI}_i$
an $8.5$-approximation to
the optimal $\beta_i$-exclusive revenue for this instance.
%, denoted by $BGR$,
%which uses the optimal $\beta_i$-exclusive mechanism
%for each player $i$.
Moreover, by Theorem~1 of \cite{yao2015n},
the optimal $\beta_i$-exclusive revenue
is an $8$-approximation to
%the sum of
the optimal $\beta_i$-adjusted revenue for the same instance.
Accordingly, we have

\begin{eqnarray}\label{eq:reva_0}
\hspace{30pt}
& & \mathop\bE\limits_{v\sim \cD}\sum_i Rev(Bund_i(\cD'_{i, M_i^1}, v_{i, M_i^1}, \beta_{i, M_i^1}))  \nonumber \\
%
%& = & \mathop\bE\limits_{v\sim \cD}\sum_i Rev(Bund_i(\cD'_{i, M_i^1}, v'_{i, M_i^1}, \beta_{i, M_i^1}))  \nonumber \\
%
& = &  \sum_i \mathop\bE\limits_{v_{-i}\sim \cD_{-i}}\
\mathop\bE\limits_{v_i\sim \cD_i} Rev(Bund_i(\cD'_{i, M_i^1}, v_{i, M_i^1}, \beta_{i, M_i^1})) \nonumber\\
& = & \sum_i  \mathop\bE\limits_{v_{-i}\sim \cD_{-i}}\
  \mathop\bE\limits_{v'_i\sim \cD'_i}
  Rev(Bund_i(\cD'_{i, M_i^1}, v'_{i, M_i^1}, \beta_{i, M_i^1})) \nonumber \\
& = & \sum_i  \mathop\bE\limits_{v_{-i}\sim \cD_{-i}}\
  \mathop\bE\limits_{v'_i\sim \cD'_i}
  Rev(Bund_i(\cI'_{i, M_i^1}, \beta_{i, M_i^1})) \nonumber \\
%
%& = & \sum_i  \mathop\bE\limits_{v_{-i}\sim \cD_{-i}}\
%  \mathop\bE\limits_{v'_i\sim \cD'_i}
%  Rev(Bund_i(\cI'_{i}, \beta_{i})) \nonumber \\
%
&\geq& \frac{1}{8.5} \sum_i \mathop\bE\limits_{v_{-i}\sim\cD_{-i}}\
\mathop\bE\limits_{v'_i \sim \cD'_i} Rev^X(\cI'_{i, M_i^1}, \beta_{i, M_i^1})\nonumber\\
&= & \frac{1}{8.5}\sum_i \mathop\bE\limits_{v_{-i}\sim\cD_{-i}}\
\mathop\bE\limits_{v'_i \sim \cD'_i} Rev^X(\cI'_i, \beta_i)\nonumber\\
&\geq& \frac{1}{8.5\times 8} \sum_i
\mathop\bE\limits_{v_{-i} \sim \cD_{-i}}\
\mathop\bE\limits_{v'_i \sim \cD'_i} Rev^A(\cI'_i,\beta_i) \nonumber \\
&=& \frac{1}{68} \sum_i
\mathop\bE\limits_{v_{-i} \sim \cD_{-i}}\
\mathop\bE\limits_{v_i \sim \cD_i} Rev^A(\cI'_i, \beta_i; v'_i) \nonumber \\
%
%&=& \frac{1}{68} \sum_i
%\mathop\bE\limits_{v_{-i} \sim \cD_{-i}}\
%\mathop\bE\limits_{\substack{{v_i \sim \cD_i,} \\ {v'_{i, M_i^1} = v_{i, M_i^1},} \\ {v'_{i, M_i^2} = 0}}} Rev^A(\cI'_i, \beta_i; v'_i) \nonumber \\
%
%& = & \frac{1}{68} \sum_i \mathop\bE\limits_{v \sim \cD} Rev^A(\cI'_i, \beta_i; v'_i)
&=&  \frac{1}{68} \mathop\bE\limits_{v \sim \cD} \sum_i Rev^A(\cI'_i, \beta_i; v'_i) = \frac{1}{68} \bE_{v\sim \cD} Rev^A(\cI', \beta).
\end{eqnarray}
%The first equality above is because both terms are the revenue of $Bund$ applied to the single-player instance
%with player $i$ and items $M_i^1$.
%, where the reserve prices for $i$ are determined by $v_{-i, M_i^1}$.
The second equality above is because
$\cD'_{i}$ and $\cD_{i}$ are the same when only $M_i^1$ is concerned, and $v'_{ij} = v_{ij}$ for all $j\in M_i^1$.
%---indeed, $M_i^1$ is the set of items for which player $i$'s distributions are known by others.
%The fourth equality is again because player $i$'s values for items in $M_i^2$ are constantly 0 according to $\cD'_i$.
The first inequality is by the relation between
$Bund_i$ and $Rev^X$.
The next equality is because $i$'s values for items in $M_i^2$ are always 0 under $\cI'_i$,
thus even when these items are available, a $\beta$-exclusive mechanism does not sell them to $i$ anyway.
The second inequality is by the relation between
$Rev^X$ and the $Rev^A$.
The next equality is because, as mentioned before,
when sampling $v_i$ from $\cD_i$ and then projecting it on $G$,
the resulting $v'_i$ is distributed according to $\cD'_i$.
%all the values for $M_i^2$ to be 0 and
%under instance $\cI'$,
%$v_i$ is first mapped to $v_i'$ and then
%give the resulting $v'_i$ to the mechanism,
%the distribution of $v'_i$ is
%exactly
%leading to the same distribution for $v_i'$ as if it were sampled from
%$\cD_i'$, as required by the instance $\cI'_i$.
%Indeed, we write $Rev^A(\cI'_i, \beta_i; v'_i)$ to emphasize that, although the values are sampled as $v_i$ from $\cD_i$, the optimal $\beta$-adjusted revenue is for instance $\cI'_i$, thus $v_i$ has been mapped to $v'_i$.

%
%\begin{eqnarray}
%&&\bE_{v\sim \cD}\sum_iRev(\cM'_{i,IEBVCG}(\cI_i))\nonumber\\
%&=&\bE_{v\sim \cD}\sum_i(Rev(\cM_{i,IEBVCG}(\cI^*_i))
%+ \sum_{j\in \cU_i,v_{ij}~is~max}\max_{i^*\neq i}v_{i^*j}) \nonumber\\
%&\geq&\bE_{v'\sim \cD'}\sum_iRev(\cM_{i,BGR}(\cI'_i,\beta_i))
%+ \bE_{v\sim \cD}\sum_i\sum_{j\in \cU_i,v_{ij}~is~max}\max_{i^*\neq i}v_{i^*j} \nonumber\\
%&\geq&\bE_{v'\sim \cD'}\sum_i\frac{1}{8.5}Rev^A(\cI'_i,\beta_i)
%+ \bE_{v\sim \cD}\sum_i\sum_{j\in \cU_i,v_{ij}~is~max}\max_{i^*\neq i}v_{i^*j}\nonumber\\
%&\geq&\frac{1}{8.5}\left(\bE_{v'\sim \cD'}\sum_iRev^A(\cI'_i,\beta_i)
%+ \bE_{v\sim \cD}\sum_i\sum_{j\in \cU_i,v_{ij}~is~max}\max_{i^*\neq i}v_{i^*j}\right)\nonumber
%\end{eqnarray}

Now we state the key lemma in our analysis, which connects the
hybrid adjusted revenue with that for $\cI'$.
Recall that each $\beta'_i$ is defined based on~$v'_{-i}\sim \cD'_{-i}$.
 %which is drawn from  $\cD'$.

\begin{lemma}\label{lem:IEA:C}
$$\bE_{v\sim \cD} Rev^A(\cI', \beta)
%\bE_{v\sim \cD}\sum_iRev^A(\cI'_i,\beta_i; v'_i)
+ \bE_{v\sim \cD}\sum_i\sum_{j\in M_i^2 \cap M_i}\max_{i'\neq i}v_{i'j}
\geq
%\bE_{v'\sim \cD'}\sum_i Rev^A(\cI'_i,\beta_i')
\bE_{v'\sim \cD'}Rev^A(\cI',\beta').$$
\end{lemma}

Before proving Lemma \ref{lem:IEA:C},
%we first finish the proof of Lemma \ref{lem:IEA:A} using it.
note that combining it with Equations~\ref{eq:bvcgi_0} and \ref{eq:reva_0} we have
\begin{equation}\label{eq:bvcgi}
\mathop\bE\limits_{v\sim \cD} Rev(\cM_{IEA}(\cI))
\geq
%\frac{1}{68} \mathop\bE\limits_{v'\sim \cD'}\sum_i Rev^A(\cI'_i,\beta_i')=
\frac{1}{68} \mathop\bE\limits_{v'\sim \cD'}Rev^A(\cI',\beta').
\end{equation}
%where the equality is by the definition of $Rev^A$.
Thus Lemma \ref{lem:IEA:A} holds.
\end{proof}

Next we  prove Lemma \ref{lem:IEA:C}.

\begin{proof}[Proof of Lemma \ref{lem:IEA:C}]
Arbitrarily fixing a player $i$ and $v'_{-i}\sim \cD'_{-i}$,
%$\beta'_i$,
let $\cM^*$ be the single-player DST Bayesian mechanism with the optimal $\beta'_i$-adjusted revenue for $\cI'_i$.%
\footnote{If the superior is not achieved by any mechanism, one can take a sequence of mechanisms whose $\beta'_i$-adjusted revenue approaches the superior in the limit.}
Accordingly, $\cM^*$ maximizes the following quantity:
$$\mathop\bE\limits_{v'_{-i}\sim \cD'_{-i}}\mathop\bE\limits_{v'_i\sim \cD'_i}[Rev(\cM^*(\cI'_{i})) - \sum_{j: v_{ij}'\leq \beta'_{ij}}q_{ij}v_{ij}'],$$
where $q_{ij}$ is the probability for $i$ to get item $j$ in $\cM^*$ under $v'_i$.
 By definition,
%\begin{eqnarray}\label{eqn:12}
%&&\bE_{v'\sim \cD'} Rev_i^A(\cI,\beta_i')
%=\bE_{v_{-i}'\sim D_{-i}'} \bE_{v_{i}'\sim \cD_{i}'}Rev_{i}^A(v_i',\beta_i')  \nonumber\\
%&=&\bE_{v_{-i}'\sim D_{-i}'} \bE_{v'\sim \cD'}(Rev(\cM^*(v_{i}'))
%- \sum_{j:v_{ij}'\leq \beta'_{ij}}q_{ij}v_{ij}')
\begin{equation}\label{equ:16_-1}
\mathop\bE\limits_{v'\sim \cD'} Rev^A(\cI'_i,\beta_i')
=\mathop\bE\limits_{v'\sim \cD'} [Rev(\cM^*(\cI_{i}'))
- \sum_{j:v_{ij}'\leq \beta'_{ij}}q_{ij}v_{ij}'].
\end{equation}
%\end{eqnarray}
Again because sampling $v$ from $\cD$ and projecting it on $G$
induces the same distribution for $v'$ as~$\cD'$, we have
%$$\bE_{v'\sim \cD'} Rev_i^A(\cI'_i,\beta_i') =   \bE_{v\sim \cD} Rev_i^A(\cI'_i,\beta_i')$$
% and
\begin{equation}\label{equ:16_0}
\mathop\bE\limits_{v'\sim \cD'} [Rev(\cM^*(\cI_{i}'))
- \sum_{j:v_{ij}'\leq \beta'_{ij}}q_{ij}v_{ij}']
 = \mathop\bE\limits_{v\sim \cD} [Rev(\cM^*(\cI_{i}'; v'_i))
- \sum_{j:v_{ij}'\leq \beta'_{ij}}q_{ij}v_{ij}'],
\end{equation}
where we explicitly include $v'_i$ in the input of $\cM^*$ to emphasize the projection.
%therefore
%\begin{equation}\label{eqn:12}
%%&&\bE_{v'\sim \cD'} Rev_i^A(\cI,\beta_i')
%%=\bE_{v_{-i}'\sim D_{-i}'} \bE_{v_{i}'\sim \cD_{i}'}Rev_{i}^A(v_i',\beta_i')  \nonumber\\
%%&=&\bE_{v_{-i}'\sim D_{-i}'} \bE_{v'\sim \cD'}(Rev(\cM^*(v_{i}'))
%%- \sum_{j:v_{ij}'\leq \beta'_{ij}}q_{ij}v_{ij}')
%\bE_{v\sim \cD} Rev_i^A(\cI'_i,\beta_i')  = \bE_{v\sim \cD} [Rev(\cM^*(\cI_{i}'))
%- \sum_{j:v_{ij}'\leq \beta'_{ij}}q_{ij}v_{ij}'].
%\end{equation}

%By the definition of $Rev^A(\cI'_i,\beta'_i)$, it is the optimal adjusted revenue for

Arbitrarily fix $v_{-i}$ and let $\bU_{i}=\{j | j\in M_{i}^{1}, \beta_{ij}' \neq \beta_{ij} \}$.
It is clear that $\beta'_{ij} < \beta_{ij}$ for each $j\in \bU_i$, as $v_{-i}' \leq v_{-i}$ (component-wise).
%First noticing that $\cI^*_i$ and $\cI'_i$ makes no difference for computing the revenue.
%Hence, $\bE_{v\sim \cD} Rev(\cM(\cI^*_i)) = \bE_{v'\sim \cD'} Rev(\cM(\cI'_i))$ and $Rev^A(\cI^*_i,\beta_i) = Rev^A(\cI'_i,\beta_i)$.
%Since $\beta'_{ij} = \beta_{ij}$ if $j\not\in \cU_i$ and $\beta'_{ij} \leq \beta_{ij}$ for $j\in \bU_i$,
%For each player $i$, we have the following,
For any $v_i$ and the corresponding $v'_i$, as $v'_{ij}=0$ for any $j\not\in M_i^1$, we have
\begin{eqnarray}\label{eqn:13}
&&\sum_{j:v'_{ij}\leq \beta'_{ij}}q_{ij}v'_{ij}
= \sum_{j\in M_{i}^{1}:v_{ij}'\leq \beta'_{ij}}q_{ij}v_{ij}' =\sum_{j\in M_{i}^{1}\backslash\bU_i: v_{ij}'\leq \beta'_{ij}} q_{ij}v_{ij}'
+ \sum_{j\in \bU_i: v_{ij}'\leq \beta'_{ij}} q_{ij}v_{ij}'\nonumber\\
&=&\sum_{j\in M_{i}^{1}\backslash\bU_i: v_{ij}'\leq \beta_{ij}} q_{ij}v_{ij}'
+ \sum_{j\in \bU_i: v_{ij}'\leq \beta_{ij}} q_{ij}v_{ij}'
- \sum_{j\in \bU_i: \beta'_{ij} < v_{ij}'\leq \beta_{ij}} q_{ij}v_{ij}'\nonumber\\
&=&\sum_{j\in M_{i}^{1}: v_{ij}'\leq \beta_{ij}} q_{ij}v'_{ij}
- \sum_{j\in \bU_i: \beta'_{ij} < v_{ij}'\leq \beta_{ij}} q_{ij}v'_{ij} \nonumber \\
&=&\sum_{j: v_{ij}'\leq \beta_{ij}} q_{ij}v'_{ij}
- \sum_{j\in \bU_i: \beta'_{ij} < v_{ij}'\leq \beta_{ij}} q_{ij}v'_{ij}.
\end{eqnarray}

Combining
%Equation \ref{eqn:12} and
Equations \ref{equ:16_-1}, \ref{equ:16_0}, \ref{eqn:13} and taking summation over all players, we have
\begin{eqnarray}\label{eqn:14}
& & \bE_{v'\sim \cD'}Rev^A(\cI',\beta')
= \sum_i \mathop\bE\limits_{v'\sim \cD'} Rev^A(\cI'_i,\beta_i') \nonumber \\
% =   \sum_i \bE_{v\sim \cD} Rev_i^A(\cI'_i,\beta_i')
& = & \sum_i \mathop\bE\limits_{v\sim \cD} [Rev(\cM^*(\cI_{i}'; v'_i))
- \sum_{j:v_{ij}'\leq \beta'_{ij}}q_{ij}v_{ij}'] \nonumber\\
&=& \sum_i \mathop\bE\limits_{v\sim \cD} \left( Rev(\cM^*(\cI'_i; v'_i)) - \sum_{j:v_{ij}'\leq \beta_{ij}} q_{ij}v_{ij}'
+ \sum_{j\in \bU_i: \beta'_{ij} < v_{ij}'\leq \beta_{ij}} q_{ij}v_{ij}'\right)\nonumber\\
%&\leq& \bE_{v_{-i}'\sim D_{-i}'} \bE_{v_{i}'\sim \cD_{i}'}Rev^A(\cI'_i,\beta_i)
%+ \bE_{v_{-i}'\sim D_{-i}'} \bE_{v_{i}'\sim \cD_{i}'} \sum_{j\in M_i^1\cap\bU_i,\beta'_{ij} < v_{ij}'\leq \beta_{ij}} q_{ij}v_{ij}'\nonumber\\
&\leq& \sum_i \mathop\bE\limits_{v\sim \cD} Rev^A(\cI'_i,\beta_i; v'_i)
+ \sum_i \mathop\bE\limits_{v\sim \cD} \sum_{j\in \bU_i: \beta'_{ij} < v_{ij}'\leq \beta_{ij}} q_{ij}v_{ij}'\nonumber\\
%
%&=& \bE_{v\sim \cD} Rev^A(\cI',\beta)
%+ \bE_{v\sim \cD} \sum_i\sum_{j\in \bU_i: \beta'_{ij} < v_{ij}'\leq \beta_{ij}} q_{ij}v_{ij}' \nonumber \\
%
&=& \sum_i \mathop\bE\limits_{v\sim \cD} Rev^A(\cI'_i,\beta_i; v'_i)
+ \mathop\bE\limits_{v\sim \cD} \sum_i \sum_{j\in \bU_i: \beta'_{ij} < v_{ij}\leq \beta_{ij}} q_{ij}v_{ij} \nonumber \\
&\leq& \sum_i \mathop\bE\limits_{v\sim \cD} Rev^A(\cI'_i,\beta_i; v'_i)
+ \mathop\bE\limits_{v\sim \cD}
\sum_i\sum\limits_{j\in \bU_i: \beta'_{ij} < v_{ij}\leq \beta_{ij}} v_{ij} \nonumber \\
&=& \bE_{v\sim \cD} Rev^A(\cI', \beta)
+ \mathop\bE\limits_{v\sim \cD}
\sum_i\sum\limits_{j\in \bU_i: \beta'_{ij} < v_{ij}\leq \beta_{ij}} v_{ij}.
\end{eqnarray}
The first inequality above is because the first two terms in the expectation
is exactly the $\beta$-adjusted revenue of mechanism $\cM^*$ on $\cI'_i$.
The following equality holds because $v'_{ij} = v_{ij}$ for any $j\in M_i^1$.
Finally, the second inequality is because $0\leq q_{ij}\leq 1$ for any $i, j$.

%
%The last equation in \ref{eqn:14} holds because
%$\bE_{v'\sim \cD'} Rev^A(\cI'_i,\beta_i) = \bE_{v\sim \cD} Rev^A(\cI'_i,\beta_i)$ for each player $i$ since the projected instance $\cI'_i$ remains the same for both distribution.
%Now for all player $i$ and all item $j\in M_i^1$, the distribution $\cD'_{ij}=\cD_{ij}$. Also by definition $q_{ij} \leq 1$, we can have
%\begin{eqnarray}\label{eqn:15}
%&&\bE_{v'\sim \cD'}\sum_i\sum\limits_{j\in M_i^1\cap\bU_i,\beta'_{ij} < v_{ij}'\leq \beta_{ij}} q_{ij}v'_{ij}
%= \bE_{v\sim \cD}\sum_i\sum\limits_{j\in M_i^1\cap\bU_i,\beta'_{ij} < v_{ij}\leq \beta_{ij}} q_{ij}v_{ij}\nonumber\\
%&\leq& \bE_{v\sim \cD}\sum_i\sum\limits_{j\in M_i^1\cap\bU_i,\beta'_{ij} < v_{ij}\leq \beta_{ij}} v_{ij}
%\end{eqnarray}

Arbitrarily fix a valuation profile $v$ and
consider the term $\sum\limits_i \sum\limits_{j\in \bU_i: \beta'_{ij} < v_{ij}\leq \beta_{ij}} v_{ij}$
 at the end of Equation \ref{eqn:14}.
We show that each item $j$ appears in the summation at most once.
Indeed, for each $v_{ij}$ that appears in the summation,
$j\in M_i^1$ and $\cD_{ij}$ is known in $\cI$.
As $v_{ij} > \beta'_{ij}$, player $i$ has strictly higher value for $j$
than any other player $i'$ whose distribution for $j$ is also known.
For any such player $i'$, $\beta'_{i'j} = v_{ij} > v_{i'j}$ and
$j$ is not in the set $\{j\in \bU_{i'}: \beta'_{i'j}< v_{i'j}\leq \beta_{i'j}\}$,
which implies that $v_{i'j}$ does not appear in the summation.
%Moreover, as $\beta'_{ij} < \beta_{ij}$,
%%as $v_{ij}\leq \beta_{ij}$,
%%player $i$'s value for~$j$ is at most the second highest among all players.
%%Since $\beta'_{ij}< v_{ij}$
%%\beta_{ij}$,
%any player $i''\neq i$
%who has the highest value for $j$ must have $\cD_{i''j}$ unknown.
Moreover, for any player $i''\neq i$ with $\cD_{i''j}$ unknown, $j\notin M^1_{i''j}$ and
%Accordingly, $j$ is not in $M_{i''}^1$ and
$v_{i''j}$ does not appear in the summation either.
Therefore item $j$ only appears once in the summation.
Accordingly,
\begin{equation}\label{equ:20_0}
\sum\limits_i \sum\limits_{j\in \bU_i: \beta'_{ij} < v_{ij}\leq \beta_{ij}} v_{ij}
= \sum_{j: \ \exists!\  i \mbox{\scriptsize{ s.t. }} j\in \bU_i, \ \beta'_{ij}< v_{ij} \leq \beta_{ij}} v_{ij}.
\end{equation}

%Arbitrarily fix $(i, j)$ such that $v_{ij}$ appears in the summation.
We now show that
\begin{equation}\label{equ:21_0}
\sum_{j: \ \exists!\  i \mbox{\scriptsize{ s.t. }} j\in \bU_i, \ \beta'_{ij}< v_{ij} \leq \beta_{ij}} v_{ij} \leq
\sum_i\sum_{j\in M_i^2 \cap M_i}\max_{i'\neq i}v_{i'j},
\end{equation}
where the right-hand side is exactly
the revenue generated by mechanism $\cM_{IEA}$ in Step \ref{step:M'IEBVCG10}
and also the desired term in the statement of Lemma \ref{lem:IEA:C}.
Indeed, for each $v_{ij}$ that appears in the left-hand side,
because $v_{ij}\leq \beta_{ij}$,
player $i$'s value for~$j$ is at most the second highest among all players.
Letting
$i^* = \argmax_{i'\neq i} v_{i'j}$
 with ties broken lexicographically, we have $\beta_{ij} = v_{i^*j}$
 and $v_{i^*j}$ is the highest value for $j$ among all players.
Since $\beta'_{ij}< \beta_{ij}$, $\cD_{i^*j}$ is unknown and $j\in M_{i^*}^2$.
Below we show $j\in M_{i^*}$, which then implies $j\in M_{i^*}^2 \cap M_{i^*}$ and
 the price paid by $i^*$ in Step \ref{step:M'IEBVCG10} for item $j$ is $\max_{i'\neq i^*}v_{i'j}\geq v_{ij}$.

When the distributions are generic, there are no ties in the players' values and $v_{i^*j}$ is the unique maximum value for $j$, thus $j\in M_{i^*}$.
For arbitrary distributions, problems occur when $v_{ij} = \beta_{ij}$ and $i<i^*$, which implies $j\notin M_{i^*}$.
To deal with this special case,
 consider the following tie-breaking method for the players:
%for any player $i$, item $j$ and $v_{-i}$,
%again let $i^* = \argmax_{i'\neq i}v_{i'j}$ with ties broken lexicographically.
while the value $\beta_{ij}$ is still defined to be $v_{i^*j}$,
we denote it by $\beta_{ij}^+$ if $i^*<i$ and $\beta_{ij}^-$ if $i^*>i$.
When $v_{ij} = \beta_{ij}$, we treat $v_{ij}$
as strictly smaller if facing $\beta_{ij}^+$ and strictly larger if facing $\beta_{ij}^-$.
All results proved in \cite{yao2015n} and above continue
to hold with respect to this tie-breaking method.
Now for any $v_{ij}$ that
appears in
%$\sum_i\sum\limits_{j\in \bU_i:\beta'_{ij} < v_{ij}\leq \beta_{ij}} v_{ij}$,
the summation,
either we have $v_{ij}< \beta_{ij}$ or we have $v_{ij} = \beta_{ij}$ and $i^*< i$, thus it is always the case that $j\in M_{i^*}$.

%Accordingly, for any $v_{ij}$ appearing in the summation,
%$v_{ij}\leq \max_{i'\neq i^*}v_{i'j}$ and $j\in M_{i^*}^2\cap M_{i^*}$.
%
%$v_{ij}< v_{i^*j}$ and $j\in M_{i^*}$.
%
Accordingly,
$$\sum_{j: \ \exists!\  i \mbox{\scriptsize{ s.t. }} j\in \bU_i, \ \beta'_{ij}< v_{ij} \leq \beta_{ij}} v_{ij}
\leq
\sum_{j: \ \exists!\  i^* \mbox{\scriptsize{ s.t. }} j\in M_{i^*}^2\cap M_{i^*}} \max_{i'\neq i^*}v_{i'j}
=
\sum_i\sum_{j\in M_i^2 \cap M_i}\max_{i'\neq i}v_{i'j},$$
and Equation \ref{equ:21_0} holds.
%\begin{eqnarray}\label{eqn:16}
%\sum_i\sum\limits_{j\in \bU_i:\beta'_{ij} < v_{ij}\leq \beta_{ij}} v_{ij}
%\leq \sum_i\sum_{j\in M_i^2 \cap M_i}\max_{i'\neq i}v_{i'j}.
%\end{eqnarray}
Combining Equations \ref{eqn:14}, \ref{equ:20_0} and \ref{equ:21_0}, we have
$$
\bE_{v'\sim \cD'}Rev^A(\cI',\beta')
%\sum_i \mathop\bE\limits_{v'\sim \cD'} Rev^A(\cI'_i,\beta_i')
\leq \bE_{v\sim \cD} Rev^A(\cI', \beta)
%\sum_i \mathop\bE\limits_{v\sim \cD} Rev^A(\cI'_i,\beta_i; v'_i)
+ \mathop\bE\limits_{v\sim \cD} \sum_i\sum_{j\in M_i^2 \cap M_i}\max_{i'\neq i}v_{i'j},
$$
and Lemma \ref{lem:IEA:C} holds.
\end{proof}

We finish the whole analysis by proving Lemma \ref{lem:IEA:B}.

\begin{proof}[Proof of Lemma \ref{lem:IEA:B}]
Denote by $\cM_{1LA}$ the {\em individual 1-lookahead} mechanism, which sells each item separately using the 1-lookahead mechanism \cite{ronen2001approximating} and
is a 2-approximation to Myerson's mechanism for each item.
Similar to the proof of Lemma \ref{lem:IEA:A}, the expected revenue generated by
$Bund_i$ in Steps~\ref{step:M'IEBVCG8} and \ref{step:M'IEBVCG9} is

\begin{eqnarray}\label{eq:rev1la}
& & \mathop\bE\limits_{v\sim \cD}\sum_i Rev(Bund_i(\cD'_i, v_{i, M_i^1}, \beta_i(v_{-i})))  \nonumber \\
& = & \sum_i  \mathop\bE\limits_{v_{-i}\sim \cD_{-i}}\
  \mathop\bE\limits_{v_i\sim \cD_i}
  Rev(Bund_i(\cD_{i, M_i^1}, v_{i, M_i^1}, \beta_{i, M_i^1})) \nonumber \\
&\geq& \sum_i \mathop\bE\limits_{v_{-i}\sim\cD_{-i}}\
\mathop\bE\limits_{v_i \sim \cD_i} Rev(\cM_{1LA}(\cD_{i, M_i^1}, v_{i, M_i^1}, \beta_{i, M_i^1})) \nonumber \\
&=& \sum_i \mathop\bE\limits_{v\sim \cD} Rev(\cM_{1LA}(\hat{\cI}_{i, M_i^1}, \beta_{i, M_i^1})).
\end{eqnarray}
The first equality above is again because items in $M_i^2$ does not affect
$Bund_i$ given that player $i$'s values
for them are constantly 0 according to $\cD'_i$,
and $\cD'_i$ and $\cD_i$ coincide
when only items in $M_i^1$ are concerned.
The inequality holds because given $\cD_{i, M_i^1}$ and $\beta_{i, M_i^1}$,
mechanism $Bund_i$
chooses between optimally selling to $i$ the items as a bundle (i.e., $e_i>0$) and optimally selling to $i$ each item separately
(i.e., $e_i=0$),
whichever generates higher expected revenue over $v_{i, M_i^1}\sim \cD_{i, M_i^1}$;
while $\cM_{1LA}$ is a particular mechanism
that sells each item $j\in M_i^1$ to $i$ separately based on $\cD_{ij}$ and $\beta_{ij}$.
%and the revenue generated from selling items separately for player $i$
%on item set $M_i^1$ with optimal reserve if $i$ is the highest bidder
%is indeed the revenue generated from player $i$
%by running $\cM_{1LA}$ on item set $M_i^1$.
Accordingly, letting $\cD_j=\times (\cD_{ij})_{i\in N}$ and $v_j\sim \cD_j$ for each item $j$,
we have
\begin{eqnarray}\label{eq:1la:bd_1}
%&& Rev(\cM_{IEA}(\cI))
& & \bE_{v\sim \cD} Rev(\cM_{IEA}(\cI)) \nonumber  \\
&=& \bE_{v\sim \cD}\sum_i\left(Rev\left(Bund_i(\cD'_{i}, v_{i, M_i^1}, \beta_i(v_{-i}))\right)
+ \sum_{j\in M_i^2 \cap M_i}\max_{i'\neq i}v_{i'j}\right) \nonumber  \\
&\geq& \sum_i \mathop\bE\limits_{v\sim \cD}
\left( Rev(\cM_{1LA}(\hat{\cI}_{i, M_i^1}, \beta_{i, M_i^1}))
%Rev(\cM_{1LA}(\hat{\cI}_{i, M_i^1}, v_{i, M_i^1}, v_{-i, M_i^1}))
+ \sum_{j\in M_i^2 \cap M_i}\max_{i'\neq i}v_{i'j} \right) \nonumber\\
&=& \sum_{j\in M} \sum_i \mathop\bE\limits_{v_j\sim \cD_j}
\left( {\bf I}_{j\in M_i^1} \cdot Rev(\cM_{1LA}(\hat{\cI}_{i, j}, \beta_{ij}))
+ {\bf I}_{j\in M_i^2 \cap M_i} \cdot \max_{i'\neq i}v_{i'j} \right).
%&\geq& \sum_j \sum_i \mathop\bE\limits_{v_j\sim \cD_j}
%{\bf I}_{j\in M_i^1} \cdot Rev(\cM_{1LA}(\cI'_{i, j}, v'_{ij}, v'_{-i, j})).
\end{eqnarray}

Arbitrarily fixing an item $j$,
% since the $\cM_{1LA}$ runs independently on each item,
we show that
%for each coupled valuation profile $(v_i, v_{-i})$ and $(v'_i, v'_{-i})$,
%\begin{eqnarray*}
%&& \sum_i \mathop\bE\limits_{v\sim \cD}
%\left( Rev(\cM_{1LA}(\hat{\cI}_{i, M_i^1}, v_{i, M_i^1}, v_{-i, M_i^1}))
%+ \sum_{j\in M_i^2 \cap M_i}\max_{i^*\neq i}v_{i^*j} \right)\\
%&\geq& \sum_i \mathop\bE\limits_{v\sim \cD}
%Rev(\cM_{1LA}(\cI'_{i, M_i^1}, v'_{i, M_i^1}, v'_{-i, M_i^1})).
%\end{eqnarray*}
%
%Note that
%Therefore, it is sufficient to show that for all item $j$,
\begin{eqnarray}\label{eq:1la:bd}
%&& Rev(\cM_{IEA}(\cI))
%& & \bE_{v\sim \cD} Rev(\cM_{IEA}(\cI)) \nonumber  \\
%%
%&=& \bE_{v\sim \cD}\sum_i\left(Rev\left(Bund_i(\cD'_{i}, v_{i, M_i^1}, \beta_i(v_{-i}))\right)
%+ \sum_{j\in M_i^2 \cap M_i}\max_{i'\neq i}v_{i'j}\right) \nonumber  \\
%&\geq& \sum_i \mathop\bE\limits_{v\sim \cD}
%\left( Rev(\cM_{1LA, i}(\hat{\cI}_{i, M_i^1}, \beta_{i, M_i^1}))
%%Rev(\cM_{1LA}(\hat{\cI}_{i, M_i^1}, v_{i, M_i^1}, v_{-i, M_i^1}))
%+ \sum_{j\in M_i^2 \cap M_i}\max_{i^*\neq i}v_{i^*j} \right) \nonumber\\
%&=& \sum_j \sum_i \mathop\bE\limits_{v_j\sim \cD_j}
%\left( {\bf I}_{j\in M_i^1} \cdot Rev(\cM_{1LA}(\hat{\cI}_{i, j}, v_{ij}, v_{-i, j}))
%+ {\bf I}_{j\in M_i^2 \cap M_i} \cdot \max_{i^*\neq i}v_{i^*j} \right) \nonumber\\
& & \sum_i \mathop\bE\limits_{v_j\sim \cD_j}
\left( {\bf I}_{j\in M_i^1} \cdot Rev(\cM_{1LA}(\hat{\cI}_{i, j}, \beta_{ij}))
+ {\bf I}_{j\in M_i^2 \cap M_i} \cdot \max_{i'\neq i}v_{i'j} \right) \nonumber \\
&\geq& \sum_i \mathop\bE\limits_{v'_j\sim \cD'_j}
{\bf I}_{j\in M_i^1} \cdot Rev(\cM_{1LA}(\cI'_{i, j}, \beta'_{ij})).
\end{eqnarray}
%which together with Equation \ref{eq:1la:bd_1} implies
%\begin{eqnarray*}
%%&& Rev(\cM_{IEA}(\cI))
%& & \bE_{v\sim \cD} Rev(\cM_{IEA}(\cI)) \geq
%\sum_{j\in M} \sum_i \mathop\bE\limits_{v'_j\sim \cD'_j}
%{\bf I}_{j\in M_i^1} \cdot Rev(\cM_{1LA}(\cI'_{i, j}, \beta'_{ij})) = ,
%
%Lemma \ref{lem:IEA:B}.

To do so, by definition,
if $j\in M_i^1$ for a player $i$ then in the Bayesian instance $\cI'_{i,j}$ and given the reserve price $\beta'_{ij}$,
the
1-lookahead mechanism tries to sell $j$ to $i$
at price%
\footnote{Strictly speaking, under the Bayesian instance $\cI'$,
the individual 1-lookahead mechanism for item $j$ works as follows.
Given $v'_j\sim \cD_j$,
it finds the highest bidder $i$ for $j$ with ties broken lexicographically, as well as the second highest bid which is exactly $\beta'_{ij}$.
It then tries to sell $j$ to $i$ at price $r'_{ij}$.
For generic distributions,
there are no ties in $v'_j$ and
the mechanism can be run on each instance $\cI'_{ij}$ and $\beta'_{ij}$ separately,
without ever selling $j$ to more than one players.
For arbitrary distributions, this can be done by introducing proper tie-breaking rules.
}

$$r'_{ij} = \argmax_{x} x \Pr_{v'_{ij}\sim\cD'_{ij}} (v'_{ij} \geq x | v'_{ij} \geq \beta'_{ij}).$$
%First we have $v_{ij} = v'_{ij}, \cD_{ij} = \cD'_{ij}$ for $j\in M_i^1$,
%and given $v_{-i,j}$ and $\beta_{ij} = \max_{i'\neq i} v_{i'j}$, the 1-Lookahead mechanism for player $i$ on item $j$ is to sell the item to him if he is the highest bidder with price
%$r_{ij}(v_{-i,j}) = \argmax_{x} x \Pr_{v_{ij}\sim\cD_{ij}} (v_{ij} \geq x | v_{ij} \geq \beta_{ij})$.
%\footnote{Similarly we define $\beta'_{ij} = \max_{i'\neq i} v'_{i'j}$ and
%$r_{ij}(v'_{-i,j}) = \argmax_{x} x \Pr_{v_{ij}\sim\cD_{ij}} (v_{ij} \geq x | v_{ij} \geq \beta'_{ij})$.}
Thus
\begin{eqnarray}\label{equ:23_b}
& & \sum_i \mathop\bE\limits_{v'_j\sim \cD'_j}
{\bf I}_{j\in M_i^1} \cdot Rev(\cM_{1LA}(\cI'_{i, j}, \beta'_{ij})) \nonumber \\
& = & \sum_i \mathop\bE\limits_{v'_{-i,j} \sim \cD'_{-i,j}} \mathop\bE\limits_{v'_{ij} \sim \cD'_{ij}}
{\bf I}_{j\in M_i^1} \cdot  {\bf I}_{v'_{ij} \geq r'_{ij}} \cdot r'_{ij} \nonumber  \\
& = & \sum_i \mathop\bE\limits_{v_{-i,j} \sim \cD_{-i,j}} \mathop\bE\limits_{v_{ij} \sim \cD_{ij}}
{\bf I}_{j\in M_i^1} \cdot  {\bf I}_{v'_{ij} \geq r'_{ij}(v'_{-i, j})} \cdot r'_{ij}(v'_{-i, j}),
\end{eqnarray}
where the second equality is again by sampling $v_j$ from $\cD_j$ and projecting on $G$ to get $v'_j$.
We write $r'_{ij}$ as $r'_{ij}(v'_{-i,j})$ to highlight this fact.
Note that given $v'_{-i,j}$, $r'_{ij}(v'_{-i,j})$ depends on the distribution $\cD'_{ij}$ but not
on any concrete $v'_{ij}$.
Also note that $r'_{ij}(v'_{-i,j})\geq \beta'_{ij}$.
Next, we divide the last term in Equation~\ref{equ:23_b} into two parts,
%above term into two cases.
depending on whether $r'_{ij}(v'_{-i,j}) < \beta_{ij}$ or not.
That is,
\begin{eqnarray}\label{eq:1la:lowr_1}
&& \sum_i \mathop\bE\limits_{v_{-i,j} \sim \cD_{-i,j}} \mathop\bE\limits_{v_{ij} \sim \cD_{ij}}
{\bf I}_{j\in M_i^1} \cdot  {\bf I}_{v'_{ij} \geq r'_{ij}(v'_{-i,j})} \cdot r'_{ij}(v'_{-i,j}) \nonumber\\
&=& \sum_i \mathop\bE\limits_{v_{-i,j} \sim \cD_{-i,j}} \mathop\bE\limits_{v_{ij} \sim \cD_{ij}}
{\bf I}_{j\in M_i^1} \cdot {\bf I}_{v'_{ij} \geq r'_{ij}(v'_{-i,j})} \cdot
\left(
{\bf I}_{r'_{ij}(v'_{-i,j})<\beta_{ij}}
+  {\bf I}_{r'_{ij}(v'_{-i,j})\geq \beta_{ij}}
\right)\cdot r'_{ij}(v'_{-i,j}).\nonumber\\
\end{eqnarray}
For the $r'_{ij}(v'_{-i,j})<\beta_{ij}$ part in Equation \ref{eq:1la:lowr_1}, we have

\begin{eqnarray}\label{eq:1la:lowr}
& & \sum_i \mathop\bE\limits_{v_{-i,j} \sim \cD_{-i,j}} \mathop\bE\limits_{v_{ij} \sim \cD_{ij}}
{\bf I}_{j\in M_i^1} \cdot {\bf I}_{v'_{ij} \geq r'_{ij}(v'_{-i,j})} \cdot
{\bf I}_{r'_{ij}(v'_{-i,j})<\beta_{ij}} \cdot r'_{ij}(v'_{-i,j}) \nonumber \\
%
%&=& \sum_i \mathop\bE\limits_{v_{-i,j} \sim \cD_{-i,j}} \mathop\bE\limits_{v_{ij} \sim \cD_{ij}}
%{\bf I}_{j\in M_i^1} \cdot  {\bf I}_{v_{ij} \geq r_{ij}(v'_{-i,j})} \cdot r_{ij}(v'_{-i,j}) \nonumber\\
%
&=& \sum_i \mathop\bE\limits_{v_{-i,j} \sim \cD_{-i,j}}
\mathop\bE\limits_{v_{ij} \sim \cD_{ij}}
{\bf I}_{j\in M_i^1} \cdot
\left(
{\bf I}_{v'_{ij} \geq \beta_{ij}> r'_{ij}(v'_{-i,j})} +
{\bf I}_{\beta_{ij} > v'_{ij} \geq r'_{ij}(v'_{-i,j})}
\right) \cdot r'_{ij}(v'_{-i,j}) \nonumber\\
&=& \sum_i \mathop\bE\limits_{v_{-i,j} \sim \cD_{-i,j}}
\mathop\bE\limits_{v_{ij} \sim \cD_{ij}}
{\bf I}_{j\in M_i^1} \cdot
\left(
{\bf I}_{v_{ij} \geq \beta_{ij}> r'_{ij}(v'_{-i,j})}
+
{\bf I}_{\beta_{ij} > v_{ij} \geq r'_{ij}(v'_{-i,j})}
\right) \cdot r'_{ij}(v'_{-i,j}) \nonumber\\
&\leq& \sum_i \mathop\bE\limits_{v_{-i,j} \sim \cD_{-i,j}}
\mathop\bE\limits_{v_{ij} \sim \cD_{ij}}
{\bf I}_{j\in M_i^1} \cdot
\left(
{\bf I}_{v_{ij} \geq \beta_{ij}>r'_{ij}(v'_{-i,j})}
\cdot \beta_{ij}
+
{\bf I}_{\beta_{ij} > v_{ij} \geq r'_{ij}(v'_{-i,j})} \cdot v_{ij}
\right) \nonumber\\
&=& \sum_i \mathop\bE\limits_{v_{-i,j} \sim \cD_{-i,j}}
\mathop\bE\limits_{v_{ij} \sim \cD_{ij}}
{\bf I}_{j\in M_i^1} \cdot
{\bf I}_{\beta_{ij}>r'_{ij}(v'_{-i,j})} \cdot {\bf I}_{v_{ij} \geq \beta_{ij}}
\cdot \beta_{ij} \nonumber \\
& & \quad
+ \sum_i \mathop\bE\limits_{v_{-i,j} \sim \cD_{-i,j}} \mathop\bE\limits_{v_{ij} \sim \cD_{ij}}
{\bf I}_{j\in M_i^1} \cdot {\bf I}_{\beta_{ij} > v_{ij} \geq r'_{ij}(v'_{-i,j})} \cdot v_{ij}.
\end{eqnarray}
The first equality above is by distinguishing whether $\beta_{ij}\leq v'_{ij}$ or not.
The second equality is because $v_{ij} = v'_{ij}$ whenever $j\in M_i^1$.
The inequality is because $r'_{ij}(v'_{-i,j})< \beta_{ij}$ following the
indicator in the first term
and $r'_{ij}(v'_{-i,j})\leq v_{ij}$ following
the indicator in the second term.
Finally, the last equality is because ${\bf I}_{v_{ij} \geq \beta_{ij}>r'_{ij}(v'_{-i,j})}  = {\bf I}_{\beta_{ij}>r'_{ij}(v'_{-i,j})} \cdot {\bf I}_{v_{ij} \geq \beta_{ij}}$
in the first term.

For the first term in Equation \ref{eq:1la:lowr}, because the
indicators ${\bf I}_{j\in M_i^1}$ and ${\bf I}_{\beta_{ij}>r'_{ij}(v'_{-i,j})}$
does not depend on $v_{ij}$,
we have
\begin{eqnarray*}
\hspace{30pt}
& & \sum_i \mathop\bE\limits_{v_{-i,j} \sim \cD_{-i,j}}
\mathop\bE\limits_{v_{ij} \sim \cD_{ij}}
{\bf I}_{j\in M_i^1} \cdot
{\bf I}_{\beta_{ij}>r'_{ij}(v'_{-i,j})} \cdot {\bf I}_{v_{ij} \geq \beta_{ij}}
\cdot \beta_{ij} \\
&=& \sum_i \mathop\bE\limits_{v_{-i,j} \sim \cD_{-i,j}}
{\bf I}_{j\in M_i^1} \cdot
{\bf I}_{\beta_{ij}>r'_{ij}(v'_{-i,j})}
\left(
\mathop\bE\limits_{v_{ij} \sim \cD_{ij}}
{\bf I}_{v_{ij} \geq \beta_{ij}}
\cdot \beta_{ij}\right) \\
&\leq& \sum_i \mathop\bE\limits_{v_{-i,j} \sim \cD_{-i,j}}
{\bf I}_{j\in M_i^1} \cdot {\bf I}_{\beta_{ij}>r'_{ij}(v'_{-i,j})}
\left(
\mathop\bE\limits_{v_{ij} \sim \cD_{ij}}
Rev(\cM_{1LA}(\hat{\cI}_{i, j}, \beta_{ij}))
\right),
\end{eqnarray*}
where the inequality is because on the Bayesian instance $\hat{\cI}_{i,j}$ and given
the reserve price $\beta_{ij}$, mechanism $\cM_{1LA}$
chooses the optimal price $r_{ij}$ to maximize the expected revenue
and
is no worse than simply setting $r_{ij} = \beta_{ij}$.

For the second term in Equation \ref{eq:1la:lowr},
for any valuation profile $v$, if there exists a player $i$ is such that the two indicators
both equal to 1, then
we have $\beta_{ij} > v_{ij} \geq \beta'_{ij}$,
because $\beta_{ij} > v_{ij} \geq r'_{ij}(v'_{-i,j})$
and $r'_{ij}(v'_{-i,j}) \geq \beta'_{ij}$.
Thus the (lexicographically first) highest bidder $i^*$ for $j$ in $v_j$ has his distribution unknown,
and $j\in M^2_{i^*}\cap M_{i^*}$;
and player $i$ is the (lexicographically first) highest bidder for $j$ in $v'_j$, which is unique.
Accordingly,
\begin{eqnarray*}
\hspace{30pt}
& & \sum_i \mathop\bE\limits_{v_{-i,j} \sim \cD_{-i,j}} \mathop\bE\limits_{v_{ij} \sim \cD_{ij}}
{\bf I}_{j\in M_i^1} \cdot {\bf I}_{\beta_{ij} > v_{ij} \geq r'_{ij}(v'_{-i,j})} \cdot v_{ij}\\
&=& \mathop\bE\limits_{v_{j} \sim \cD_{j}}
\sum_i
{\bf I}_{j\in M_i^1} \cdot {\bf I}_{\beta_{ij} > v_{ij} \geq r'_{ij}(v'_{-i,j})} \cdot v_{ij} \\
&\leq & \mathop\bE\limits_{v_{j} \sim \cD_{j}}
{\bf I}_{j\in M^2_{i^*}\cap M_{i^*}} \cdot
\max_{i: j\in M_i^1} v_{ij}
\leq  \mathop\bE\limits_{v_{j} \sim \cD_{j}}
{\bf I}_{j\in M^2_{i^*}\cap M_{i^*}} \cdot
\max_{i'\neq i^*} v_{i'j}\\
& = &  \mathop\bE\limits_{v_{j} \sim \cD_{j}}
 \sum_i
{\bf I}_{j\in M^2_{i}\cap M_{i}} \cdot
\max_{i'\neq i} v_{i'j},
\end{eqnarray*}
where the last equality is because $i^*$ is the unique player such that
the indicator is 1.

Combining the above two equations with Equation \ref{eq:1la:lowr},
for the $r'_{ij}(v'_{-i,j})<\beta_{ij}$ part in Equation~\ref{eq:1la:lowr_1}
we have
\begin{eqnarray*}
& & \sum_i \mathop\bE\limits_{v_{-i,j} \sim \cD_{-i,j}} \mathop\bE\limits_{v_{ij} \sim \cD_{ij}}
{\bf I}_{j\in M_i^1} \cdot {\bf I}_{v'_{ij} \geq r'_{ij}(v'_{-i,j})} \cdot
{\bf I}_{r'_{ij}(v'_{-i,j})<\beta_{ij}} \cdot r'_{ij}(v'_{-i,j}) \nonumber \\
&\leq &
\sum_i \mathop\bE\limits_{v_{-i,j} \sim \cD_{-i,j}}
{\bf I}_{j\in M_i^1} \cdot {\bf I}_{\beta_{ij}>r'_{ij}(v'_{-i,j})}
\left(
\mathop\bE\limits_{v_{ij} \sim \cD_{ij}}
Rev(\cM_{1LA}(\hat{\cI}_{i, j}, \beta_{ij}))
\right) \nonumber \\
& &
+ \mathop\bE\limits_{v_{j} \sim \cD_{j}}
 \sum_i
{\bf I}_{j\in M^2_{i}\cap M_{i}} \cdot
\max_{i'\neq i} v_{i'j} \nonumber \\
& =& \sum_i
\mathop\bE\limits_{v_{j} \sim \cD_{j}} \left(
{\bf I}_{j\in M_i^1} \cdot {\bf I}_{\beta_{ij}>r'_{ij}(v'_{-i,j})} \cdot
Rev(\cM_{1LA}(\hat{\cI}_{i, j}, \beta_{ij}))
+
{\bf I}_{j\in M_i^2 \cap M_i} \cdot \max_{i'\neq i}v_{i'j} \right).
\end{eqnarray*}
%
%The first inequality holds since $r_{ij}(v'_{-i,j}) < \beta_{ij}$ by assumption and $v_{ij} \geq r_{ij}(v'_{-i,j})$ given the indicator is true.
%The second inequality holds because the 1-Lookahead mechanism for player $i$ on item $j$ given $v_{-i,j}$ fixed chooses the optimal reserve larger than or equal to $\beta_{ij}$.
%Also $r_{ij}(v'_{-i,j}) < \beta_{ij}$ implies that there exists a unique player $i'\neq i$
%where $j\in M_i^2$ and $v_{i'j} = \max_{i^*} v_{i^*j}$.

For the $r'_{ij}(v'_{-i,j})\geq \beta_{ij}$ part in Equation~\ref{eq:1la:lowr_1},
%$r_{ij}(v'_{-i,j}) \geq \beta_{ij}$,
we immediately have $r'_{ij}(v'_{-i,j}) = r_{ij}(v_{-i,j})$ when the indicators are 1,
because $\cD'_{ij} = \cD_{ij}$ and $v'_{ij} = v_{ij}$ whenever $j\in M_i^1$,
and $r_{ij}(v_{-i,j})$ maximizes the expected revenue over $\cD_{ij}$ conditional on
$v_{ij}\geq \beta_{ij}$.
Thus

\begin{eqnarray*}
& & \sum_i \mathop\bE\limits_{v_{-i,j} \sim \cD_{-i,j}} \mathop\bE\limits_{v_{ij} \sim \cD_{ij}}
{\bf I}_{j\in M_i^1} \cdot {\bf I}_{v'_{ij} \geq r'_{ij}(v'_{-i,j})} \cdot
{\bf I}_{r'_{ij}(v'_{-i,j})\geq\beta_{ij}} \cdot r'_{ij}(v'_{-i,j}) \nonumber \\
&=& \sum_i \mathop\bE\limits_{v_{-i,j} \sim \cD_{-i,j}} \mathop\bE\limits_{v_{ij} \sim \cD_{ij}}
{\bf I}_{j\in M_i^1} \cdot
{\bf I}_{v_{ij} \geq r_{ij}(v_{-i,j})} \cdot
{\bf I}_{r'_{ij}(v'_{-i,j})\geq\beta_{ij}}
\cdot r_{ij}(v_{-i,j}) \nonumber  \\
&=& \sum_i \mathop\bE\limits_{v_{-i,j} \sim \cD_{-i,j}}
{\bf I}_{j\in M_i^1} \cdot {\bf I}_{r'_{ij}(v'_{-i,j})\geq\beta_{ij}}
\left(
\mathop\bE\limits_{v_{ij} \sim \cD_{ij}}
{\bf I}_{v_{ij} \geq r_{ij}(v_{-i,j})} \cdot r_{ij}(v_{-i,j})
\right) \nonumber \\
&=& \sum_i \mathop\bE\limits_{v_{-i,j} \sim \cD_{-i,j}}
{\bf I}_{j\in M_i^1} \cdot {\bf I}_{r'_{ij}(v'_{-i,j})\geq\beta_{ij}}
\left(
\mathop\bE\limits_{v_{ij} \sim \cD_{ij}}
Rev(\cM_{1LA}(\hat{\cI}_{i, j}, \beta_{ij}))
\right).
\end{eqnarray*}

Combining the above two equations with Equation~\ref{eq:1la:lowr_1}
and then Equation \ref{equ:23_b},
we have

\begin{eqnarray*}
& & \sum_i \mathop\bE\limits_{v'_j\sim \cD'_j}
{\bf I}_{j\in M_i^1} \cdot Rev(\cM_{1LA}(\cI'_{i, j}, \beta'_{ij})) \nonumber \\
&\leq &
\sum_i
\mathop\bE\limits_{v_{j} \sim \cD_{j}} \left(
{\bf I}_{j\in M_i^1} \cdot {\bf I}_{\beta_{ij}>r'_{ij}(v'_{-i,j})} \cdot
Rev(\cM_{1LA}(\hat{\cI}_{i, j}, \beta_{ij}))
+
{\bf I}_{j\in M_i^2 \cap M_i} \cdot \max_{i'\neq i}v_{i'j} \right) \\
& & +
\sum_i \mathop\bE\limits_{v_{-i,j} \sim \cD_{-i,j}}
{\bf I}_{j\in M_i^1} \cdot {\bf I}_{r'_{ij}(v'_{-i,j})\geq\beta_{ij}}
\left(
\mathop\bE\limits_{v_{ij} \sim \cD_{ij}}
Rev(\cM_{1LA}(\hat{\cI}_{i, j}, \beta_{ij}))
\right) \\
& = & \sum_i \mathop\bE\limits_{v_j\sim \cD_j}
\left( {\bf I}_{j\in M_i^1} \cdot Rev(\cM_{1LA}(\hat{\cI}_{i, j}, \beta_{ij}))
+ {\bf I}_{j\in M_i^2 \cap M_i} \cdot \max_{i'\neq i}v_{i'j} \right),
\end{eqnarray*}
and Equation \ref{eq:1la:bd} holds.

Taking summation over all items $j$ on both sides of Equation \ref{eq:1la:bd}
and combining with Equation~\ref{eq:1la:bd_1},
 we have
\begin{eqnarray*}
& & \bE_{v\sim \cD} Rev(\cM_{IEA}(\cI))
\geq
\sum_{j\in M} \sum_i \mathop\bE\limits_{v'_j\sim \cD'_j}
{\bf I}_{j\in M_i^1} \cdot Rev(\cM_{1LA}(\cI'_{i, j}, \beta'_{ij})) \\
&=& \mathop\bE\limits_{v'\sim \cD'}
\sum_i \sum_{j\in M_i^1}
 Rev(\cM_{1LA}(\cI'_{i, j}, \beta'_{ij}))
= \mathop\bE\limits_{v'\sim \cD'}  Rev(\cM_{1LA}(\cI')) \geq \frac{1}{2} \mathop\bE\limits_{v'\sim \cD'} IM(\cI'),
\end{eqnarray*}
where the inequality is because
$\cM_{1LA}$ is a 2-approximation to Myerson's mechanism for each item~\cite{ronen2001approximating}.
Thus Lemma \ref{lem:IEA:B} holds.
\end{proof}

\section{Proofs for Section \ref{sec:partial}}\label{app:partial}%\ref{sec:k=0}

\subsection{Proof of Theorem \ref{thm:unit-k}}
\label{app:known:unit}

\paragraph*{Theorem \ref{thm:unit-k}.} (restated) {\em
$\forall k\in [n-1]$, any unit-demand auction instances $\hat{\cI} = (N, M, \cD)$ and $\cI = (N, M, \cD, G)$ where $G$ is $k$-informed,
mechanism $\cM'_{IEUD}$ is 2-DST and
$\bE_{v\sim \cD} Rev(\cM'_{IEUD}(\cI)) \geq \frac{\tau_k}{24} \cdot OPT(\hat{\cI}).$
%
%$\forall k\in [n-1]$, unit-demand auction instances $\hat{\cI} = (N, M, \cD)$ and $\cI = (N, M, \cD, G)$ where $G$ is $k$-informed,
%$\cM'_{IEUD}$ is 2-DST and
%$\bE_{v\sim \cD} Rev(\cM'_{IEUD}(\cI)) \geq \frac{\tau_k}{24} \cdot OPT(\hat{\cI}).$
}

\begin{proof}
The proof is almost the same as that of Theorem \ref{thm:unit}, thus most details are omitted.
Below we only show that under the players' truthful strategies,
the probability for each distribution $\cD_{ij}$ to be reported in mechanism $\cM'_{IEUD}$ is at least $\tau_{k}$.
Indeed, for any player $i$ and item $j$,
%We first show that the so-defined $q$ maximizes the probability that a distribution $\cD_{ij}$ is reported:
%We can get the probability for a distribution $\cD_{ij}$ to be reported by following equation.
\begin{eqnarray*}
\Pr(\cD_{ij} \text{ is reported in the mechanism}) &= & \Pr(i\in N_2)\Pr(\exists i'\in N_1, (i, i')\in G_j \ | \  i \in N_2) \nonumber\\
&\geq &(1-q)(1-(1-q)^k),
\end{eqnarray*}
where the inequality is because $\cD_{ij}$ is known by at least $k$ players other than $i$,
and the players are partitioned independently.
Taking derivatives of the last term, we have that it is maximized when $q=1-(k+1)^{-\frac{1}{k}}$ as in the mechanism, in which case
$$\Pr(\cD_{ij} \text{ is reported in the mechanism})\geq
\frac{1}{(k+1)^{\frac{1}{k}}}\cdot \frac{k}{k+1}
%(k+1)^{-\frac{1}{k}}-(k+1)^{-\frac{k+1}{k}}
= \tau_k.$$

Combined with the proof of Theorem \ref{thm:unit}, Theorem \ref{thm:unit-k} holds.
\end{proof}

\subsection{Proof of Theorem \ref{thm:additive}}\label{app:known:additive}

\paragraph{The Information Elicitation Individual Myerson Mechanism.}
%To prove Theorem \ref{thm:additive}, w
We start by
introducing the information elicitation individual Myerson mechanism $\cM_{IEIM}$,
which runs the following mechanism $\cM_{IEIM,j}$ for each item $j$ separately.
Mechanism $\cM_{IEIM,j}$ is similar to $\cM_{IEUD}$ and $\cM'_{IEUD}$, thus we have omitted many details in the analysis.

\begin{algorithm}
\floatname{algorithm}{Mechanism}
  \caption{\hspace{-4pt} $\cM_{IEIM,j}$}
  \label{alg:IEMj}
  \begin{algorithmic}[1]
  \STATE Each player $i$ reports a value $b_{ij}$ and a knowledge $K_{ij} = (\cD^i_{i'j})_{i'\neq i}$.

  \STATE Randomly partition the players into two sets, $N_1$ and $N_2$, where each player is independently put in $N_1$ with probability $q=1-(k+1)^{-\frac{1}{k}}$ and $N_2$ with probability $1-q$.

 \STATE  Let $N_3$ be the set of  players in $N_2$ whose distributions are reported by some players in $N_{1}$, and $\cD'_{N_3, j}$ be the vector of reported distributions.

 % \STATE Reward players in $N_{1}$ using Brier's scoring rule, properly scaled so that the total reward $R$ is at most $\epsilon/m$.

  \STATE Run Myerson's mechanism on the single-good Bayesian instance
  $\hat{\cI}_{N_3, j} = (N_{3}, \{j\}, \cD'_{N_3, j})$ with the values being $(b_{ij})_{i\in N_3}$; and use the resulting allocation and prices to sell to players in $N_3$.

  \end{algorithmic}
\end{algorithm}

For each item $j$,
let $v_j = (v_{ij})_{i\in N}$, $\cD_j = (\cD_{ij})_{i\in N}$,
$\hat{\cI}_j = (N, \{j\}, \cD_j)$ be the corresponding single-good Bayesian instance,
and $\cI_j = (N, \{j\}, \cD_j, G_j)$ be the corresponding single-good information elicitation instance.
Lemma \ref{add:IM:truthful} below is similar to Lemma \ref{ud:truthful} and we provide its statement only.

\begin{lemma}
\label{add:IM:truthful}
For any additive auction instances $\hat{\cI} = (N, M, \cD)$ and $\cI = (N, M, \cD, G)$, mechanism $\cM_{IEIM,j}$ is 2-DST for $\cI_j$ for each $j\in M$, and mechanism $\cM_{IEIM}$ is 2-DST for $\cI$.
\end{lemma}

Next, we consider the expected revenue of $\cM_{IEIM}$.

\begin{lemma}
\label{lem:add:IM}
$\bE_{v\sim \cD} Rev(\cM_{IEIM}({\cI})) \geq \tau_k  \bE_{v\sim \cD} Rev(IM(\hat{{\cI}}))$.
\end{lemma}

\begin{proof}
By definition,
$$\bE_{v\sim \cD} Rev(\cM_{IEIM}({\cI})) = \sum_{j\in M} \bE_{v_j \sim \cD_j} Rev(\cM_{IEIM, j}(\cI_j))$$
% where
%$\cI_j = (N, \{j\}, (\cD_{ij})_{i\in N}, G_j)$ is the single-good crowdsourced Bayesian instance with item $j$.
%Also notice that
and
$$\bE_{v\sim \cD} Rev(IM(\hat{{\cI}})) = \sum_{j\in M}  OPT(\hat{\cI}_j).$$
%where $\hat{\cI}_j = (N, \{j\}, (\cD_{ij})_{i\in N})$ is the single-good Bayesian instance with item $j$.
Accordingly, it suffices to show that for each item $j$,
$$\bE_{v_j \sim \cD_j} Rev(\cM_{IEIM, j}(\cI_j)) \geq \tau_k OPT(\hat{\cI}_j).$$

Using ideas and notations similar to those in the proofs of Theorem \ref{thm:unit} and Lemma \ref{lem:proj}, we have

\begin{eqnarray*}
%\label{eq:myerson}
%& &Rev(\cM_{IEIM,j}, {\cI}) \nonumber \\
& & \bE_{v_j \sim \cD_j} Rev(\cM_{IEIM, j}(\cI_j))
%=  \mathop\mathbb{E}\limits_{N_3} \ \mathop\mathbb{E}\limits_{v_{N_3, j}\sim {\cD_{N_{3}, j}}}
%OPT(\hat{{\cI}}_{N_3, j}) - R
\nonumber \\
&=& \mathop\mathbb{E}\limits_{N_3} \ \mathop\mathbb{E}\limits_{v_{N_3, j}\sim {\cD_{N_{3}, j}}}
OPT(\hat{{\cI}}_{N_3, j})  \geq  \mathop\mathbb{E}\limits_{N_3}\ \mathop\mathbb{E}\limits_{v_j\sim \cD_j}
OPT(\hat{{\cI}}_j)_{N_3}  \nonumber \\
&=& \mathop\mathbb{E}\limits_{N_3}\ \mathop\mathbb{E}\limits_{v_j\sim \cD_j}
\sum\limits_{i\in N_3}P_i(OPT(\hat{{\cI}}_j)) =  \mathop\mathbb{E}\limits_{v_j\sim \cD_j} \  \mathop\mathbb{E}\limits_{N_3} \sum\limits_{i\in N_3}
P_i(OPT(\hat{{\cI}}_j))  \nonumber \\
&=& \mathop\mathbb{E}\limits_{v_j \sim \cD_j} \sum\limits_{i}\Pr(i\in N_3) \cdot
P_i(OPT(\hat{{\cI}}_j))  \geq  \tau_k \mathop\mathbb{E}\limits_{v_j \sim \cD_j} \sum\limits_{i}
P_i(OPT(\hat{{\cI}}_j))  \nonumber \\
& = & \tau_k OPT(\hat{{\cI}}_j),
\end{eqnarray*}
as desired. Thus Lemma \ref{lem:add:IM} holds.
%
%The analysis of (\ref{eq:myerson}) is the same with the proof in theorem \ref{thm:unit}.
%Hence, $\cM_{IEIM,j}$ achieves a $(\tau_k,\epsilon)$-approximation for IM on each item $j$,
%which implies a $(\tau_k,\epsilon)$-approximation for $\cM_{IEIM}$.
\end{proof}

\paragraph{The Information Elicitation Individual 1-Lookahead Mechanism.}
%Since the 1-Lookahead mechanism is a 2-approximation to Myerson's mechanism,
Next,
we introduce the information elicitation 1-lookahead mechanism
$\cM_{IE1LA}$, which
runs the following
mechanism $\cM_{IE1LA, j}$ for each item $j$ separately.
We will show that the revenue of $\cM_{IE1LA}$ matches that of mechanism $\cM_{1LA}$ for any $k\geq 1$.

\begin{algorithm}
\floatname{algorithm}{Mechanism}
  \caption{\hspace{-4pt} $\cM_{IE1LA, j}$}
  \label{alg:ola}
  \begin{algorithmic}[1]

  \STATE Each player $i$ reports a value $b_{ij}$ and a knowledge $K_{ij} = (\cD^i_{i'j})_{i'\neq i}$.

 % \STATE Reward each player using Brier's scoring rule, properly scaled so that  the total reward $R$ is at most~$\epsilon/m$.

  \STATE Set $i^* = \argmax_i b_{ij}$ and $p_{second} = \max_{i\neq i^*} b_{ij}$.

  \STATE If $i^*$'s distribution is not reported,
  % then the item is unsold.
  % and every player pays 0.
  sell item $j$ to him at price $p_{second}$ and halt here.

  \STATE Otherwise, let $\cD'_{i^*j}$ be the reported distribution for $i^*$ (if there are many reported distributions for him, take the one by the lexicographically first reporter).

  \STATE Let $p_{i^*} = \max\limits_p \Pr\limits_{v_{i^*j}\sim \cD'_{i^*j}} (v_{i^*j}\geq p \ | \  v_{i^*j}\geq p_{second})\cdot p$. If $b_{i^*j}\geq p_{i^*}$ then sell item $j$ to $i^*$ at price~$p_{i^*}$; otherwise the item is unsold.\label{step6a}

\end{algorithmic}
\end{algorithm}

Note that $\cM_{IE1LA, j}$ does not partition the players into two groups.
Also, it is not exactly using the 1-lookahead mechanism as a blackbox, because it has to handle boundary cases where the players' distributions are not all reported.
%As we will show, the mechanism is 2-DST.
When the players all tell the truth,
all true distributions will indeed be reported.
However, for the mechanism to be well defined, it has to know what to do in all possible cases.
Moreover,
running the 1-lookahead mechanism on the set of players whose distributions are reported is not 2-DST: for example, if the player with the second highest value is the only one who knows the distribution for the player with the highest value, then he may choose not to report his knowledge about the latter, so that he himself has the highest value in the 1-Lookahead mechanism and gets a high utility. That is
why the mechanism only tries to sell to the player with the highest value.
We have the following two lemmas.

\begin{lemma}
\label{add:ola:truthful}
For any additive auction instances $\hat{\cI} = (N, M, \cD)$ and $\cI = (N, M, \cD, G)$,
mechanism $\cM_{IE1LA,j}$ is 2-DST for each $\cI_j$, and mechanism $\cM_{IE1LA}$ is 2-DST for $\cI$.
\end{lemma}
\begin{proof}
As in mechanism $\cM_{IEA}$, in each mechanism $\cM_{IE1LA,j}$, the fact that it is dominant for each player $i$ to
report his true value no matter what knowledge the players report
follows from the truthfulness of the second-price mechanism
and that of the 1-lookahead mechanism.
Given that all players report their true values, a player $i$'s reported knowledge does not affect whether he is $i^*$ or not. It may affect the other players' utilities, but not his own. Thus reporting his true knowledge never hurts him,
and mechanism $\cM_{IE1LA,j}$ is 2-DST for $\cI_j$.
%
%Arbitrarily fix an item $j$, a player $i$, a strategy subprofile of the other players and a knowledge of $i$. We show that $\cM_{IE1LA,j}$ is {\em monotone} and
%uses the {\em threshold} payment for player $i$. Indeed, assume $i$ gets the item by reporting a value $b_{ij}$ and let him report a higher value $b'_{ij}$.
%Notice that $b_{ij}$ is the highest value among all players, and so is $b'_{ij}$.
%If the other players did not report a distribution for $i$'s value, then $i$'s price is the second highest value, and he still gets the item at the same price by reporting $b'_{ij}$.
%Otherwise, $i$'s price is the reserve price of the 1-Lookahead mechanism as defined in Step \ref{step6a}, which does not depend on his reported value, thus he still gets the item at this price by reporting $b'_{ij}$.
%
%The above analysis also shows
%that player $i$ pays the threshold payment,
%which is the highest value among the other players if his distribution is not reported,
%and is the reserve price of the 1-Lookahead mechanism otherwise.
%Accordingly, it is dominant for $i$ to report his true value.
%
%Given that all players report their true values, it is easy to see that for each player $i$, reporting his true knowledge does not hurt him.
%%Indeed, the only way for a player's reported knowledge to affect his own utility is to set his reward according to Brier's scoring rule, which is maximized by reporting his true knowledge.

Since the players have additive valuations and $\cM_{IE1LA}$ runs each mechanism $\cM_{IE1LA, j}$ separately for item $j$, we have that $\cM_{IE1LA}$ is 2-DST for $\cI$ and Lemma \ref{add:ola:truthful} holds.
\end{proof}

\begin{lemma}
\label{lem:add:ola}
$\bE_{v\sim \cD} Rev(\cM_{IE1LA}({\cI}))
= \bE_{v\sim \cD} Rev(\cM_{1LA}(\hat{{\cI}}))
\geq \frac{1}{2} \bE_{v\sim \cD} Rev(IM(\hat{{\cI}}))$.
\end{lemma}

\begin{proof}
When the players report their true values and true knowledge,
%despite of the reward given to the players,
the outcome of each $\cM_{IE1LA, j}$ on
the information elicitation instance $\cI_j$
is the same as that of mechanism $\cM_{1LA}$ on the Bayesian instance $\hat{\cI}_j$,
because the distribution for $i^*$ is reported.
% = (N, \{j\}, \cD_j)$.
Accordingly,

\begin{eqnarray*}
& & \bE_{v\sim \cD} Rev(\cM_{IE1LA}({\cI}))
= \sum\limits_{j\in M} \bE_{v_j\sim \cD_j} Rev(\cM_{IE1LA, j}({\cI_j}))
= \bE_{v\sim \cD} Rev(\cM_{1LA}(\hat{{\cI}})) \\
& \geq & \sum\limits_{j\in M} \dfrac{1}{2}OPT(\hat{{\cI}}_j)
= \dfrac{1}{2} \bE_{v\sim \cD} Rev(IM(\hat{{\cI}})),
\end{eqnarray*}
where the inequality is because the 1-lookahead mechanism
is a 2-approximation to the optimal Bayesian mechanism for each item $j$ \cite{ronen2001approximating}.
Thus Lemma \ref{lem:add:ola} holds.
\end{proof}

Note that the approximation ratio of $\cM_{IE1LA}$ does not depend on the specific value of $k$, as long as $k\geq 1$.

\paragraph{The Information Elicitation $BVCG$ Mechanism.}
The mechanism $\cM_{IEBVCG}$ is defined in Mechanism~\ref{alg:bvcg}.
It is similar to $\cM_{IEA}$ and approximates mechanism $BVCG$ in information elicitation settings.
If a player $i$'s value distributions are not all reported, $\cM_{IEBVCG}$ throws $i$ away and leaves
his winning set unsold.
This
simplifies the instructions compared to $\cM_{IEA}$ and still ensures truthfulness. Doing so would seriously damage the revenue if the knowledge graphs can be totally arbitrary.
However, when everything is known by somebody and when the players report their true knowledge,
no player is actually thrown away.
%But again, this instruction is needed to ensure truthfulness.
We have the following two lemmas.

%, and this step is needed for completeness.

\begin{algorithm}
\floatname{algorithm}{Mechanism}
  \caption{\hspace{-4pt} $\cM_{IEBVCG}$}
  \label{alg:bvcg}
  \begin{algorithmic}[1]

  %\STATE Each player $i$ reports a valuation $b_i = (b_{ij})_{j\in M}$ and a knowledge $K_i = (\cD^i_{i'j})_{i'\neq i, j\in M}$.

%  \STATE Reward the players using Brier's scoring rule, such that the total reward $R$ is at most $\epsilon$.

 \STATE  Each player $i$ reports a valuation $b_i = (b_{ij})_{j\in M}$ and a knowledge $K_i = (\cD^i_{i'j})_{i'\neq i, j\in M}$.

  \STATE For each item $j$, set $i^*(j)= \argmax_i b_{ij}$ (ties broken lexicographically) and $p_j = \max_{i\neq i^*} b_{ij}$.

  \FOR{each player $i$}
   \STATE Let $M_i = \{j \ | \  i^*(j) = i\}$ be player $i$'s winning set.

   \STATE If not all $m$ distributions of $i$'s values are reported, $i$ gets no item and items in $M_i$ are unsold.

   \STATE Otherwise, let $\cD'_i$ be the vector of reported distributions for $i$'s values (if there are more than one reporters for an item, take the lexicographically first).

   \STATE Compute the entry fee $e_i(\cD'_i, b_{-i})$ using $BVCG$. Note that different from mechanism $Bund$, $BVCG$ does not compute extra reserve prices for $i$.

   \STATE Sell $M_i$ to player $i$ according to $BVCG$.
   That is, if $\sum_{j\in M_i} b_{ij}\geq e_i(\cD'_i, b_{-i})
   + \sum_{j\in M_i} p_j$
   then $i$ gets $M_i$ with price $e_i(\cD'_i, b_{-i}) + \sum_{j\in M_i} p_j$;
   otherwise $i$ gets no item and the items in $M_i$ are unsold.

\ENDFOR

\end{algorithmic}
\end{algorithm}

\begin{lemma}
\label{add:bvcg:truthful}
Mechanism $\cM_{IEBVCG}$ is 2-DST.
\end{lemma}
\begin{proof}
Arbitrarily fix a player $i$, a strategy subprofile
of the other players, and a knowledge of $i$.
If not all $m$ distributions of $i$'s values are reported by the others,
then $i$ gets nothing and pays nothing, so it does not matter what valuation he reports about himself.
Otherwise, $\cM_{IEBVCG}$ sells to player~$i$
in the same way as $BVCG$:
using the other players' highest reported value as
the reserve
price for each item, either player $i$ gets the whole set of items for which his value passes the reserve price (i.e., his winning set), or he gets nothing and those items are unsold to anybody. Following~\cite{yao2015n}, it is dominant for $i$ to report his true values given any entry fee that does not depend on his reported values, so is it
%reporting his true values is still dominant for $i$
when the entry fee is computed based on $\cD'_i$ and $b_{-i}$.
%, the latter of which determines the reserve prices.

Moreover, a player $i$'s reported knowledge
$K_i$ about others affects
neither $M_i$ nor~$e_i$, nor the reserve prices for him.
Thus
%$K_i$ is only used to compute player $i$'s reward based on Brier's scoring rule, and $i$'s expected reward is maximized by
reporting his true knowledge never hurts him and Lemma \ref{add:bvcg:truthful} holds.
\end{proof}

\begin{lemma}
\label{lem:add:bvcg}
$\mathop\mathbb{E}\limits_{v\sim \cD}Rev(\cM_{IEBVCG}({\cI})) =  \mathop\mathbb{E}\limits_{v\sim \cD}Rev(BVCG(\hat{{\cI}}))$.
\end{lemma}

\begin{proof}
Since $\cM_{IEBVCG}$ retrieves the whole distribution $\cD$ from the players, its outcome is exactly the same as that of $BVCG$ under the Bayesian instance $\hat{\cI}$.
%Thus, its expected revenue is that of $BVCG$ minus the total reward, which is at most $\epsilon$.
\end{proof}

\vspace{-10pt}
\paragraph*{Remark.} Similar to $\cM_{IE1LA}$,
the revenue of $\cM_{IEBVCG}$ does not
depend on the specific value of~$k$, as long as $k\geq 1$.
Indeed, notice the special structures of
the two Bayesian mechanisms $\cM_{1LA}$ and $BVCG$:
the winning set of a player $i$
solely depends on the players' values;
the distribution $\cD_i$ is only used to
compute better reserve prices or entry fee
to increase revenue;
and the distribution $\cD_{-i}$ is irrelevant to $i$.
Therefore, in the information elicitation setting we can allow a player to be both a reporter about the others' distributions and a potential buyer of some items.
In some other mechanisms such as Myerson's mechanism,
all players' distributions are used both to choose the potential winner and to set his price,
thus in the information elicitation setting we must separate the knowledge reporters and the potential winners.
%Indeed, if player $i$ would have gotten some items when he reports nothing about player $j$, but those items will be given to $j$ if he reports truthfully,
%then $i$ may be better off reporting nothing. That is why random partitioning is needed in these cases.
%

\medskip

We are now ready to prove Theorem \ref{thm:additive}.

\vspace{-5pt}
\paragraph*{Theorem \ref{thm:additive}.} (restated) {\em
\sloppy
For any $k \in [n-1]$, any additive auction instances $\hat{\cI} = (N, M, \cD)$ and $\cI = (N, M, \cD, G)$ where $G$ is $k$-informed,
the mechanism $\cM'_{IEA}$ is 2-DST and
$$\bE_{v\sim \cD} Rev(\cM'_{IEA}(\cI)) \geq \max\{\frac{1}{11}, \frac{\tau_k}{6+2\tau_k}\} \cdot OPT(\hat{\cI}).$$
}

\begin{proof}
%[Proof of Theorem \ref{thm:additive}]
Recall that mechanism $\cM'_{IEA}$ is defined as follows:
when $k\leq 7$, it runs $\cM_{IEBVCG}$ with probability $\frac{2}{11}$ and
$\cM_{IE1LA}$ with probability $\frac{9}{11}$;
when $k> 7$, it runs $\cM_{IEBVCG}$ with probability $\frac{\tau_k}{3+\tau_k}$ and
$\cM_{IEIM}$ with probability $\frac{3}{3+\tau_k}$.
%The choice of the two cases is to achieve the best approximation ratio for each $k$.
%It flips a fair coin; if heads comes up then it runs $\cM_{IEBVCG}$; and if tails comes up, then it runs $\cM_{IE1LA}$ when $k\leq 3$ and $\cM_{IEIM}$ when $k>3$.

The mechanism $\cM'_{IEA}$ is clearly 2-DST, since all the sub-mechanisms are 2-DST, and which mechanism is chosen does not depend on the players' strategies.

When $k\leq 7$, we have $\frac{\tau_k}{6+2\tau_k} < \frac{1}{11}$. By Lemmas \ref{lem:add:ola} and \ref{lem:add:bvcg},
\begin{eqnarray}
&&\mathop\mathbb{E}\limits_{v\sim \cD} Rev(\cM'_{IEA}( {\cI}))
= \dfrac{2}{11}\mathop\mathbb{E}\limits_{v\sim \cD}Rev(\cM_{IEBVCG}( {\cI})) +
\dfrac{9}{11}\mathop\mathbb{E}\limits_{v\sim \cD} Rev(\cM_{IE1LA}( {\cI})) \nonumber \\
& \geq& \dfrac{2}{11}\mathop\mathbb{E}\limits_{v\sim \cD}Rev(BVCG(\hat{{\cI}})) +
 \frac{3}{11}
\mathop\mathbb{E}\limits_{v\sim \cD}Rev(IM( \hat{{\cI}}))
+ \dfrac{3}{11}\mathop\mathbb{E}\limits_{v\sim \cD} Rev(\cM_{1LA}( {\hat{\cI}})). \label{equ:10_0}
%\\
%& \geq& \dfrac{1}{14} OPT_{B}(\hat{{\cI}}) - \epsilon.
%
%&=& \dfrac{1}{2}\mathop\mathbb{E}\limits_{v\sim \cD}Rev(BVCG(\hat{{\cI}})) +
% \max\left\{\frac{1}{4}, \frac{\tau_k}{2}\right\}
%\mathop\mathbb{E}\limits_{v\sim \cD}Rev(IM( \hat{{\cI}})) - \epsilon.
\end{eqnarray}
When $k>7$, we have $\frac{\tau_k}{6+2\tau_k} > \frac{1}{11}$. By Lemmas \ref{lem:add:IM} and \ref{lem:add:bvcg},
\begin{eqnarray}
\mathop\mathbb{E}\limits_{v\sim \cD} Rev(\cM'_{IEA}( {\cI}))
 & = & \frac{\tau_k}{3+\tau_k}\mathop\mathbb{E}\limits_{v\sim \cD}Rev(\cM_{IEBVCG}( {\cI})) +
\frac{3}{3+\tau_k}\mathop\mathbb{E}\limits_{v\sim \cD} Rev(\cM_{IEIM}( {\cI})) \nonumber \\
& \geq& \frac{\tau_k}{3+\tau_k}\mathop\mathbb{E}\limits_{v\sim \cD}Rev(BVCG(\hat{{\cI}})) +
 \frac{3\tau_k}{3+\tau_k}
\mathop\mathbb{E}\limits_{v\sim \cD}Rev(IM( \hat{{\cI}})). \label{equ:10_1}
\end{eqnarray}
By \cite{cai2016duality},
\begin{eqnarray*}
OPT(\hat{\cI}) &\leq&
2\mathop\mathbb{E}\limits_{v\sim \cD}Rev(BVCG(\hat{{\cI}}))
+ 3\mathop\mathbb{E}\limits_{v\sim \cD}Rev(IM(\hat{{\cI}}))
+ 3\mathop\mathbb{E}\limits_{v\sim \cD}Rev(\cM_{1LA}(\hat{{\cI}})) \\
&\leq& 2\mathop\mathbb{E}\limits_{v\sim \cD}Rev(BVCG(\hat{{\cI}})) + 6\mathop\mathbb{E}\limits_{v\sim \cD}Rev(IM(\hat{{\cI}})).
\end{eqnarray*}
Combined the first inequality with Inequality \ref{equ:10_0}
and the second with Inequality \ref{equ:10_1}, we have
$$
\mathop\mathbb{E}\limits_{v\sim \cD}Rev(\cM'_{IEA}({\cI}))
\geq \max\left\{\frac{1}{11}, \dfrac{\tau_k}{6+2\tau_k}\right\} \cdot
OPT(\hat{{\cI}}),
$$
and Theorem \ref{thm:additive} holds.
\end{proof}

\subsection{An Illustration of Mechanism $\cM_{IEM}$}
\label{app:figure:2con}

The sets of players involved in the first round of mechanism $\cM_{IEM}$
are illustrated in Figure~\ref{fig:IEM}.

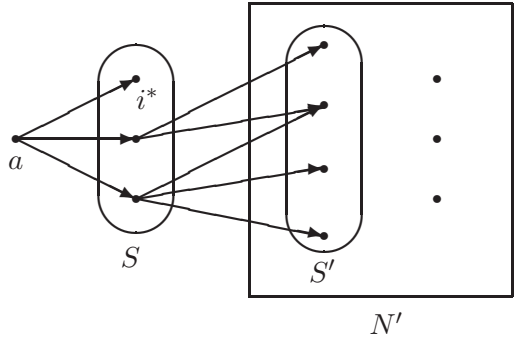
\begin{figure}[htbp]
\begin{center}
\setlength{\unitlength}{1cm}
\thicklines
\begin{picture}(7,4.8)

\put(0.4,2.7){\circle*{0.1}}
\put(2,1.9){\circle*{0.1}}
\put(2,2.7){\circle*{0.1}}
\put(2,3.5){\circle*{0.1}}

\put(4.5,1.4){\circle*{0.1}}
\put(4.5,2.3){\circle*{0.1}}
\put(4.5,3.15){\circle*{0.1}}
\put(4.5,3.95){\circle*{0.1}}

\put(6,1.9){\circle*{0.1}}
\put(6,2.7){\circle*{0.1}}
\put(6,3.5){\circle*{0.1}}

\put(0.4,2.7){\vector(1,0){1.6}}
\put(0.4,2.7){\vector(2,-1){1.6}}
\put(0.4,2.7){\vector(2,1){1.6}}

\put(2,1.9){\line(5,-1){2.5}}
\put(4.5,1.4){\vector(4,-1){0}}
\put(2,1.9){\line(6,1){2.5}}
\put(4.5,2.3){\vector(1,0){0}}
\put(2,1.9){\line(2,1){2.5}}
\put(4.5,3.15){\vector(2,1){0}}

\put(2,2.7){\line(6,1){2.5}}
\put(4.5,3.15){\vector(4,1){0}}
\put(2,2.7){\line(2,1){2.5}}
\put(4.5,3.95){\vector(2,1){0}}

\put(2,2.7){\oval(1,2.5)}
\put(4.5,2.7){\oval(1,3)}

\put(3.5,0.6){\line(0,1){3.9}}
\put(7,0.6){\line(0,1){3.9}}
\put(3.5,0.6){\line(1,0){3.5}}
\put(3.5,4.5){\line(1,0){3.5}}

\put(0.3,2.3){$a$}
\put(2,3.1){$i^*$}
\put(1.8,1){$S$}
\put(4.3,0.8){$S'$}
\put(5.1,0.1){$N'$}

\end{picture}
\caption{The sets of players involved in the first round of Mechanism $\cM_{IEM}$. The edges in the figure correspond to distributions reported by the players.
In each round, the mechanism keeps
in $S$ the player with the highest virtual value so far, drops
all the other players from $S$, and adds the players whose distributions are reported
for the first time by the dropped ones.}
\label{fig:IEM}
\end{center}
\end{figure}

\subsection{Proof of Theorem \ref{thm:myerson}}\label{app:proofwarm}
%Figure \ref{fig:IEM} illustrates the sets of players involved in the first round of
%mechanism $\cM_{IEM}$.

\begin{lemma}
\label{lem:2dst}
$\cM_{IEM}$ is 2-DST.
%can guarantee that each player's distribution can be reported by the end.
\end{lemma}

\begin{proof}
Similar to Lemma \ref{ud:truthful}, the proof takes two steps.
\begin{claim}\label{clm:1d}
For any player $i$, true value $v_i$,
%fixing any reported knowledge $K_i$, it is {\em (weakly) dominant} for $i$ to report his true valuation $v_i$. That is, for
 value $b_i$, knowledge $K_i$, and strategy subprofile $s_{-i} = (b_j, K_j)_{j\neq i}$ of the other players,
% in mechanism $\cM_{IEUD}$ we have
$\bE_{\cM_{IEM}}u_i((v_i, K_i), s_{-i}) \geq \bE_{\cM_{IEM}}u_i((b_i, K_i), s_{-i})$,
where the expectation is taken over the mechanism's random coins.
%For any player $i$, knowledge $K_i$ and true valuation $v_i$, fixing the second component of $i$'s strategy to be $K_i$, it is dominant for $i$ to report $v_i$.
%
%For any player $i$, knowledge $K_i$ and true value $v_i$, fixing the second component of $i$'s strategy to be $K_i$, it is dominant for $i$ to report $v_i$.
\end{claim}
\begin{proof}
%Arbitrarily fix a value $b_i$ and a strategy $s_j = (b_j, K_j)$ for each player $j\neq i$.
%We need to compare
%$\bE_{\cM_{IEM}} u_i(v_i, K_i)$ and $\bE_{\cM_{IEM}} u_i(b_i, K_i)$,
% where  the expectation is taken over the mechanism's random coins.%
%%\footnote{Strictly speaking, each $s_j$ should depend on $v_j$ and the expectation should also be taken over $\cD_{-i}$. However,
%%reporting $v_i$ is dominant for player~$i$ no matter what $v_{-i}$ is, thus there is no need to consider $\cD_{-i}$.}

First, conditional on $a=i$, player $i$ does not get the item and his reported value is not used by the mechanism.
% is only used in Step \ref{step13} to compute the reward of
%the other players who have reported about $i$'s distribution,
Thus $u_i((v_i, K_i), s_{-i}) = u_i((b_i, K_i), s_{-i}) = 0$ in this case.
%$i$'s own utility is the same no matter whether he reports $v_i$ or $b_i$.

Second, we compare the two utilities conditional on $a\neq i$.
Notice that when $a\neq i$,
whether or not player $i$'s distribution is reported---that is, whether or not $\cD'_i$ is defined--- only depends on $s_{-i}$. Thus $\cD'_i$ is defined under $(v_i, K_i)$ if and only if it is defined under~$(b_i, K_i)$.

If $\cD'_i$ is not defined, then $i\in N'$ at the end of the mechanism,
he does not get the item, and his reported value is not used.
%his utility is equal to the total reward he gets
%from Brier's scoring rule in Step \ref{step13}, which solely depends on $K_i$ and $b_{-i}$.
Therefore $u_i((v_i, K_i), s_{-i}) = u_i((b_i, K_i), s_{-i}) = 0$  again.
%player $i$'s utility is again the same, whether he reports $v_i$ or $b_i$.

If $\cD'_i$ is defined, then it is defined in the same round of the mechanism under both $(v_i, K_i)$ and $(b_i, K_i)$, which we refer to as round $r$. Also, $\cD'_i$ is the same in both cases and the mechanism's execution is the same till this round.
Notice that
\begin{itemize}
\item[(1)] $\phi_i(\cdot ; \cD'_i)$ is monotone in its input;

\item[(2)] $i$ gets the item if and only if
$i = i^*$ in all rounds $\ell$ with $\ell\geq r$ and his virtual value is at least~0; and

\item[(3)] when $i$ gets the item, $K_i$ is never used by the mechanism and
thus does not affect the execution of any round $\ell$ with $\ell\geq r$.
\end{itemize}
Accordingly, the mechanism is {\em monotone} in player $i$'s reported value: if $i$ gets the item by reporting some value, then he still gets it by reporting a higher value.
Moreover, when $i$ gets the item, his price in Step \ref{step13a} is the {\em threshold} payment. Following standard characterizations of single-parameter DST mechanisms,
%we have that
%if $\cD'_i$ is defined then
it is the best for player $i$ to report his true value $v_i$.
That is,
$u_i((v_i, K_i), s_{-i}) \geq u_i((b_i, K_i), s_{-i})$ when $\cD'_i$ is defined.
% for all $s_{-i}$ such that $\cD'_i$ is defined.
%, and the inequality is strict for some of them.

Combining the above cases together, we have $\bE_{\cM_{IEM}} u_i((v_i, K_i), s_{-i})\geq \bE_{\cM_{IEM}} u_i((b_i, K_i), s_{-i})$
%for all $s_{-i}$,
% and the inequality is strict for some~$s_{-i}$.
%Thus, fixing $K_i$ in player $i$'s strategy, it is dominant for $i$ to report $v_i$
and Claim \ref{clm:1d} holds.
\end{proof}

%Now we move to the second step.

\begin{claim}\label{clm:2d}
For any player $i$, true value $v_i$, true knowledge $K_i$,
% $i$'s true knowledge by $TK_i$, then
% = \times_{i'\neq i, j\in [m]} \cD^i_{i'j}$, where $\cD^i_{i'j} = \cD_{i'j}$ if edge $(i, i')$ is in the knowledge graph for item $j$ and $\cD^i_{i'j} = \bot$ otherwise.
%for any true valuation $v_i$,
knowledge $K'_i$, and knowledge subprofile $K'_{-i}(v_{-i}) = (K'_j(v_j))_{j\neq i}$ of the other players,
where each $K'_j(v_j)$ is a function of player $j$'s true value~$v_j$,
$\bE_{v_{-i}\sim \cD_{-i}} u_i((v_i, K_i), (v_{-i}, K'_{-i}(v_{-i}))) \geq \bE_{v_{-i}\sim \cD_{-i}} u_i((v_i, K'_i), (v_{-i}, K'_{-i}(v_{-i})))$.
%
%
%Given that all players report their true values,
%for a player $i$, it does not hurt him to report his true knowledge
%%it is dominant for a player $i$ to report truthfully
%$K_i = (\cD^i_j)_{j\neq i}$ as defined by the knowledge graph $G$:
%that is, $\cD^i_j = \cD_j$ for all~$j$ such that $(i, j)\in G$, and $\cD^i_j = \bot$ otherwise.
\end{claim}

\begin{proof}
%\sloppy
%Arbitrarily fix a knowledge $K'_i = (\cD'^i_j)_{j\neq i}$ that is different from $K_i$,
%%Note that, under the no-bluff assumption, $\cD'^i_j = \bot$ for all $j$ such that $(i, j)\notin G$.
%and a strategy $s_j(v_j) = (v_j, K_j(v_j))$ for each $j\neq i$ and true value $v_j$.
%%Note that $j$'s reported knowledge depends on his true value and need not be truthful.
%We now compare
%$\bE_{\cM_{IEM}, \cD_{-i}} u_i(v_i, K_i)$ and $\bE_{\cM_{IEM}, \cD_{-i}} u_i(v_i, K'_i)$,
% where  the expectation is taken over both the mechanism's random coins and the distributions of the other players' true values.

Similar to Claim \ref{clm:1d}, conditional on $a=i$, player $i$ does not get the item no matter what knowledge he reports.
%the utility of $i$ is equal to the reward he gets in Step~\ref{step13}. Since this step is reached by the mechanism with probability 1, given that it is reached, from $i$'s point of view, $v_{-i}$ is still distributed according to $\cD_{-i}$, thus his expected reward is maximized only by reporting $K_i$.
Thus
$$\bE_{v_{-i}\sim \cD_{-i}} [u_i((v_i, K_i), (v_{-i}, K'_{-i}(v_{-i}))) \  | \ a=i] = \bE_{v_{-i}\sim \cD_{-i}} [u_i((v_i, K'_i), (v_{-i}, K'_{-i}(v_{-i}))) \ | \ a=i] = 0.$$
%where the inequality is because $BSR$ is strictly proper.

Next, we compare the two utilities conditional on $a\neq i$.
%Again, since Step \ref{step13} is reached with probability 1,
%the expected reward $i$ gets in that step is strictly larger by reporting $K_i$ than $K'_i$.
%That is, we compare $i$'s expected utilities at the end of Step \ref{step13a}, under $(v_i, K_i)$ and $(v_i, K'_i)$.
Again similar to Claim \ref{clm:1d},
$\cD'_i$ is the same under both strategies of $i$, and the mechanism's execution is also the same till the round $r$ where $\cD'_i$ is defined (or till the end if $\cD'_i$ is not defined). There are three cases:
\begin{itemize}
\item
If $\cD'_i$ is not defined, then neither $K_i$ nor $K'_i$ is used by the mechanism, and $i$ has utility 0 under both strategies.

\item
If $\cD'_i$ is defined and $i = i^*$ from round $r$ to the end of the mechanism, then again $K_i$ and $K'_i$ are not used.
Thus $i$ has the same utility (maybe non-zero) under both strategies.

\item
If $\cD'_i$ is defined and $i\neq i^*$ starting from some round $r'\geq r$, then $i$ does not get the item under either strategy, thus his utility is 0 under both of them.
\end{itemize}
In sum, $\bE_{v_{-i}\sim \cD_{-i}} u_i((v_i, K_i), (v_{-i}, K'_{-i}(v_{-i}))) = \bE_{v_{-i}\sim \cD_{-i}} u_i((v_i, K'_i), (v_{-i}, K'_{-i}(v_{-i})))$
and reporting his true knowledge does not hurt player $i$.
% $i$'s utility.
%Summing it up with $i$'s expected reward in Step \ref{step13}, we have
%$$\bE_{v_{-i}\sim \cD_{-i}} [u_i(v_i, K_i) \  | \ a\neq i]>\bE_{v_{-i}\sim \cD_{-i}} [u_i(v_i, K'_i) \ | \ a\neq i].$$
%
%Combining the two inequalities above, we have
%$$\bE_{\cM_{IEM}, \cD_{-i}} u_i(v_i, K_i) > \bE_{\cM_{IEM}, \cD_{-i}} u_i(v_i, K'_i),$$
%thus Claim \ref{clm:2d} holds.
\end{proof}

%
%If the knowledge $K'_a(v_a)$ causes the mechanism to stop in Step \ref{step3},
%then $i$'s utility is 0 under both $K_i$ and $K'_i$. Otherwise, the mechanism stops in Step \ref{step13}, and below we compare player $i$'s expected utilities by the end of Step \ref{step13a} and his expected rewards in Step \ref{step13} separately. Notice that $i$'s final expected utility under a strategy is the sum of the two under the same strategy.
%
%\begin{claim}
%Conditional on $a\neq i$, for any true value subprofile $v_{-i}$ of the other players and their reported knowledge under $v_{-i}$ such that the mechanism stops in Step \ref{step13}, player $i$'s utility by the end of Step \ref{step13a} under $(v_i, K_i)$ is the same as that under $(v_i, K'_i)$.
%\end{claim}
%
%\begin{proof}
%As before, $\cD'_i$ is defined under $(v_i, K_i)$ if and only if it is defined under $(v_i, K'_i)$.
%If $\cD'_i$ is not defined, then $i\in N'$ in both cases and both utilities are 0.
%
%\end{proof}

Lemma \ref{lem:2dst} follows directly from Claims \ref{clm:1d} and \ref{clm:2d}.
\end{proof}

\vspace{-10pt}
\paragraph*{Theorem \ref{thm:myerson}.} (restated) {\em
For any single-good auction instances $\hat{\cI} = (N, M, \cD)$ and $\cI = (N, M, \cD, G)$ where $G$ is 2-connected,
$\cM_{IEM}$ is 2-DST and
$\bE_{v\sim \cD} Rev(\cM_{IEM}(\cI)) \geq (1-\frac{1}{n})OPT(\hat{{\cI}})$.
}

\begin{proof}%[Proof of Theorem \ref{thm:myerson}]
Following Lemma \ref{lem:2dst}, it remains to show $\bE_{v\sim \cD} Rev(\cM_{IEM}(\cI)) \geq (1-\frac{1}{n})OPT(\hat{{\cI}})$
under the players' truthful strategies.
The key is to explore the structure of the knowledge graph to make sure that the player with the highest virtual value is found by the mechanism with high probability.

More specifically, arbitrarily fix the player $a$ chosen by the mechanism. Notice that throughout the mechanism, $N'$ is the set of players $i\in N\setminus\{a\}$ such that $\cD'_i$ is not defined. We show that $N'= \emptyset$ at the end of the mechanism.
Indeed, since $G$ is 2-connected,
the out-degree of $a$ in $G$ is at least 2: otherwise, either $a$ cannot reach any other node in $G$, or this becomes the case after removing the unique node $j$ with $(a, j)\in G$, contradicting 2-connectedness.
Since $a$ reports his true knowledge $K_a$,
we have $|S|\geq 2$ in Step \ref{step3} and the mechanism does not stop there.
Moreover, at the beginning of each round, we have $|S|\geq 2$ and thus $S\setminus \{i^*\}\neq \emptyset$: otherwise $S'=\emptyset$ in the previous round, and the mechanism would not have reached this round.

Assume, for the sake of contradiction that the mechanism finally reaches a round~$r$ where $N'\neq \emptyset$ at the beginning but  $S' = \emptyset$ in Step \ref{step6}.
Since all players report their true knowledge, by the definition of $S'$ we have that, in graph $G$, all neighbors of $S\setminus\{i^*\}$ are in $N\setminus N'$.
Furthermore,
for any player $i\in (N\setminus N')\setminus S$,
%we have that in graph $G$,
all neighbors of $i$ are also in $N\setminus N'$: indeed, $i$ has been moved from $N'$ to $S$ and then dropped from $S$ (except player $a$, whose neighbors are in $N\setminus N'$ by definition); and when $i$ is dropped from $S$, all his neighbors in $N'$ are moved to $S$.
Accordingly, all the edges going from  $N\setminus N'$
to $N'$ are from player $i^*$,
and $G$ becomes disconnected after removing $i^*$, again contradicting 2-connectedness.
Thus $S'\neq \emptyset$ in all rounds and $N'=\emptyset$ in the end,
 as we wanted to show.

Because all players report their true values and true knowledge, we have $\cD'_{-a} = \cD_{-a}$ and $\phi_i(v_i; \cD'_i) = \phi_i(v_i; \cD_i)$ for all $i\neq a$.
Letting $\hat{\cI}_a = (N\setminus\{a\}, M, \cD_{-a})$, we claim
\begin{equation}\label{equ:1}
\bE_{v\sim \cD} [Rev(\cM_{IEM}(\cI)) | a] = OPT(\hat{\cI}_a).
\end{equation}
To see why this is true, note that by construction, in each round the mechanism keeps the player with the highest virtual value in $S$.
Thus, the final player $i^*$ has the highest virtual value in $N\setminus\{a\}$,
and $\phi_{second}$ is the second highest virtual value in $N\setminus\{a\}$.
Accordingly, the outcome of Step \ref{step13a} is the same as that of Myerson's mechanism on $\hat{\cI}_a$, so is the revenue.
%Since Brier's scoring rule is bounded in $[0, 2]$, the reward each player $i$ gets in the last step is no more than $\frac{\epsilon}{n}$ and the total reward given to the players is no more than $\epsilon$.
%It is easy to see that the total reward given to the players is no more than $\epsilon$ always.
Therefore Equation~\ref{equ:1} holds.

Finally, it remains to show that, by throwing away a random player $a$,
the mechanism does not lose much revenue.
For each player $i$, letting $P_i(OPT(\hat{\cI}))$ be the
expected price paid by $i$ in Myerson's mechanism under $\hat{\cI}$,
we have $OPT(\hat{\cI}) = \sum_{i\in N} P_i(OPT(\hat{\cI}))$.
Similar to the proof of Lemma \ref{lem:proj},
consider the following Bayesian mechanism $\cM'$ on $\hat{\cI}_a$:
it runs Myerson's mechanism on $\hat{\cI}$ and
then projects the outcome to players $N\setminus\{a\}$.
It is easy to see that $\cM'$ is DST,
thus it cannot generate more revenue than $OPT(\hat{\cI}_a)$.
As the expected revenue of $\cM'$ is $\sum_{i\neq a} P_i(OPT(\hat{\cI}))$, we have
\begin{equation}\label{equ:2}
OPT(\hat{\cI}_a) \geq \bE_{\cD_{-a}} Rev(\cM'(\hat{\cI}_a)) = \sum_{i\neq a} P_i(OPT(\hat{\cI})).
\end{equation}
Combining Equations \ref{equ:1} and \ref{equ:2}, we have
\begin{eqnarray*}
& & \bE_{v\sim \cD} Rev(\cM_{IEM}(\cI)) = \sum_{a\in N} \frac{1}{n} \bE_{v\sim \cD} [Rev(\cM_{IEM}(\cI)) | a] = \sum_{a\in N} \frac{1}{n} \left(OPT(\hat{\cI}_a)\right) \\
& \geq &  \sum_{a\in N} \frac{1}{n} \left(\sum_{i\neq a} P_i(OPT(\hat{\cI}))\right)  =  \left(\sum_{i\in N} \frac{n-1}{n} P_i(OPT(\hat{\cI}))\right)
= (1-\frac{1}{n}) OPT(\hat{\cI}),
\end{eqnarray*}
and
%Lemma \ref{lem:IEMrev} holds.
Theorem \ref{thm:myerson} holds.
\end{proof}

%\paragraph{Remark.}
%For additive auctions, when the knowledge graphs are 2-connected, instead of using mechanism
% $\cM_{IE1LA}$ or $\cM_{IEIM}$, one can use $\cM_{IEM}$
% for each item $j$.
% %we can use mechanism $\cM_{IEM}$ from Section \ref{sec:warm:myerson}
%%for single-good auctions to replace $\cM_{IE1LA}$ and $\cM_{IEIM}$, we can make
%%to improve the approximation ratio:.
%We thus have the following corollary, obtained by running  $\cM_{IEM}$ with probability $\frac{3}{4}$ and $\cM_{IEBVCG}$ with probability~$\frac{1}{4}$.
%
%% in additive auction if the knowledge graph of each item is 2-connected.
%
%\begin{corollary}\label{col:additive}
%When the knowledge graphs are 2-connected,
%the revised mechanism $\cM_{IEA}$ is 2-DST and
%$\mathop\mathbb{E}\limits_{v\sim \cD}Rev(\cM_{IEA}({\cI}))
%\geq \frac{1}{8}(1-\frac{1}{n})
%OPT(\hat{{\cI}}) - \epsilon$.
%%
%%$(\frac{1}{12}(1-\frac{1}{n}),\epsilon)$-approximation.
%\end{corollary}

\section{Using Scoring Rules to Buy Knowledge from Players}
%\label{sec:bluffing}
\label{sec:buyknowledge}

In this section we use proper scoring rules to reward the players for their knowledge, so that it is strictly better
for them to report truthfully.
More precisely,
a  {\em scoring rule} is a function~$f$ that takes as inputs a distribution $\cD'$ over a state space $\Omega$
 and a random sample $\omega$ from an underlying true distribution $\cD$ over $\Omega$, and outputs a real number.
Scoring rule $f$ is {\em proper} if
$$\bE_{\omega\sim \cD} f(\cD, \omega)\geq \bE_{\omega\sim \cD} f(\cD', \omega)$$
for any $\cD$ and $\cD'$,
and {\em strictly proper}
if the inequality is strict for any $\cD'\neq \cD$.
Moreover, $f$ is {\em bounded} if there exist constants $c_1, c_2$ such that $c_1\leq f(\cD', \omega) \leq c_2$ for any $\cD'$ and $\omega$.
Our mechanisms can use any strictly proper scoring rules that are bounded. For concreteness, we use Brier's scoring rule~\cite{brier1950verification}:
$$
BSR(\cD',\omega) =  2-(\sum_{s\in \Omega}(\delta_{\omega,s}-\cD'(s))^2) = 2\cD'(\omega) - ||\cD'||_{2}^{2}  +1,
$$
where $\cD'(s)$ is the probability of $s$ according to $\cD'$, and $\delta_{\omega, s}$ is the indicator for $\omega = s$.
Note that $BSR(\cD', \omega) \in [0,2]$ for any $\cD'$ and $\omega$.%
\footnote{The original version of $BSR$ is bounded by $-1$ and $1$, and we have shifted it up by $2$.}

In all our mechanisms, when player $i$ reports $\cD^i_{i'j}\neq \bot$ for player $i'$ and item $j$,
% distribution $\cD'_{ij}$ for player $i$ and item $j$ is defined,
the seller rewards~$i$
based on $BSR(\cD^{i}_{i'j}, b_{i'j})$.
If there are more than one reporters for the same distribution,
the seller can
either reward all of them or
randomly choose one.
%To upper-bound the total reward given to the players,
We can scale the reward for each distribution so that the total reward given to the players is at most some constant $\epsilon$,
which is an $\epsilon$ additive loss to the revenue.
For example, in Mechanism $\cM_{IEUD}$,
%to reward player $i$ for reporting player $i'$'s distribution on item $j$,
the seller could reward each player $i$ with
$$r^{i}_{i'j} = \frac{\epsilon}{2mn^{2}}BSR(\cD^i_{i'j}, b_{i'j})$$
for reporting the value distribution of player $i'$ on item $j$.
%where $\cD'_{i'j}$ is the distribution reported by $i$ and $b_{i'j}$ is $i'$'s bid on item $j$.

%Thus all the weakly 2-DST crowdsourced mechanisms can be transformed into a 2-DST mechanisms using proper scoring rules to reward players.

Although scoring rules help breaking utility-ties, they cause another problem:
a player who does not know a distribution may report something he made up, just to receive a reward.
Therefore we start by considering our mechanisms under the {\em no-bluff} assumption:
that is, a player will not report anything about a distribution that he does not know.
%transforming any 2-DST crowdsourced mechanism into a strictly 2-DST mechanism using properly scaled Brier's scoring rule
%will reduce revenue by at most $\epsilon$.
%
%Formally, the  {\em no-bluff} assumption is defined as:
%% in the main body of the paper,
%% in Sections \ref{sec:warm:myerson}-\ref{sec:player-wise};
More precisely, a player $i$ in an information elicitation auction is {\em no-bluff} if, for any knowledge graph $G_j$ and player $i'$ with $(i, i')\notin G_j$, and for any strategy $(b_i, K_i)$ of $i$,
$i$ reports $\bot$ for the corresponding distribution of $i'$.
%Notice that this assumption applies to both the partial information setting and the player-wise information setting.
Note that for a player~$i'$ with $(i, i')\in G_j$, $i$ may report any distribution about $i'$, including~$\bot$.
In some sense, the no-bluff assumption is the analogy of the no-overbidding assumption adopted in budget-constrained auctions: a player
will not bid higher than his true value or budget, even if doing so may not lead to a price higher than the latter.

For all our mechanisms, it is easy to see that the reward will only affect the players' incentives for reporting their knowledge, not their incentives for reporting their values.
Accordingly, it is still dominant for the players to report their true values, no matter what knowledge they report.
Given that the players all report their true values, the reported values are distributed according to the prior.
Thus reporting his true knowledge is now strictly better than lying for a player, because it maximizes his reward.
%This is because for each player $i$, he will be rewarded at most $mn$ times for reporting other players' distributions
%and there are $n$ players in total.
Rather than restating all our previous theorems, we summarize them in the theorem below.
%the following theorem is a summarization

\begin{theorem}\label{thm:nobluff}
Under the no-bluff assumption, for any information elicitation
mechanism in previous sections,
the revised mechanism with proper scoring rules
is 2-DST, and reporting his true knowledge is strictly better than lying for each player $i$.
Moreover, the mechanism's revenue is the same as before with an $\epsilon$ additive loss.
%there exists a BIC crowdsourced mechanism $\cM'$ with the same expected revenue in the same auction settings, without making the no-bluff assumption.
%Moreover, all players reporting their true values and true knowledge is the unique Bayesian Nash equilibrium in $\cM'$, besides the unachieveable ones where each player reports the unknown true distributions.
\end{theorem}

Next, we show how to remove the no-bluff assumption when everything is known.
%
%
%\paragraph{Removing the no-bluff assumption.}
%\label{sec:bluffing}
%%\vspace{-5pt}
%
%To remove the no-bluff assumption,
% from crowdsourced Bayesian auctions.
%Recall that, given a knowledge graph $G$, a player $i$'s true knowledge consists of the true distribution for each player $i'$ such that $(i, i')\in G$ and ``$\bot$'' for all other players.
%Without this assumption,
Without this assumption, player $i$
may report a distribution for another player $i'$'s value for an item~$j$,
even if $(i, i')\notin G_j$.
However, if there exists a third player $\hat{i}$ who knows $i'$'s distribution $\cD_{i'j}$,
and if player~$i$ is also rewarded for player
$\hat{i}$'s report, then intuitively player $i$ would have no incentive to bluff about~$i'$.
That is, {\em not only a player is paid for reporting the distributions he knows,
but he is also paid for keeping quiet about the distributions he does not know and letting the experts speak.}
Surely reporting $\cD_{i'j}$ maximizes player $i$'s expected reward, but he does not have the information to decide what $\cD_{i'j}$ is.%
\footnote{Using the standard language from epistemic game theory, player $i$'s information set contains at least two different distributions for $i'$'s value for $j$.}
Therefore, as long as reporting ``$\bot$''
 gives player $i$ the same utility as the unknown strategy of reporting $\cD_{i'j}$, and as long as reporting any distribution other than $\cD_{i'j}$ gives him a strictly smaller utility,
 %(which can be guaranteed by a proper scoring rule),
player $i$ will report ``$\bot$'' about $\cD_{i'j}$.

%More precisely, t
Taking mechanism $\cM'_{IEUD}$ in Section~\ref{sec:partial} as an example,
the players are rewarded as follows:
%we can change the reward steps to the following:
\begin{itemize}
%\vspace{-5pt}
\item
For each player $i'$ and item $j$, let $R_{i'j}$ be the set of players who did not
report ``$\bot$'' about the value of $i'$ for $j$.
Randomly select a player $\hat{i}$ from $R_{i'j}$ and
let $r_{i'j} = BSR(\cD^{\hat{i}}_{i'j}, b_{i'j})$.
%, where $\cD^{i'}_{ij}$ is the distribution reported by $i'$ and $b_{ij}$ is $i$'s reported value for $j$.
Reward {\em every} player $i\neq i'$ using $r_{i'j}$, properly scaled.
%, scaled properly so that the total reward given by the mechanism is at most $\epsilon$.
%\vspace{-5pt}
\end{itemize}
%Other parts of the mechanism remains the same.
Note that the reward $r_{i'j}$ is given to player $i$ even if he has reported $\cD^i_{i'j} = \bot$.

It is easy to see that, for any $k$-informed information elicitation instance with $k\geq 1$,
if all players except $i$ report their true values and true knowledge, then player $i$'s best strategy is to tell the truth about his own.
Indeed, reporting his true values is still dominant no matter what knowledge the players' report.
Moreover, for each player $i'$ and item $j$, there are two ways for player~$i$ to maximize
the reward $r_{i'j}$ he receives: (1) reporting $\cD^i_{i'j} = \bot$, so that $\hat{i}$ is chosen with probability 1
from the set of players who actually know $\cD_{i'j}$; or (2) successfully guessing $\cD_{i'j}$ and reporting it, so that
$\hat{i}$'s report is still $\cD_{i'j}$ with probability 1.
Note that the latter is not a well-defined Bayesian strategy, because $i$ does not have enough information to carry it out.
Accordingly, the resulting information elicitation mechanism is Bayesian incentive compatibility (BIC):
%% is defined for crowdsourced Bayesian mechanisms
%%in a similar way to that for Bayesian mechanisms:
that is, all players reporting their true
values and true knowledge is a Bayesian Nash equilibrium.
In fact, this is the only Bayesian Nash equilibrium in the mechanism,
besides the unachieveable ones where each player reports the unknown true distributions.
%%
%All the other mechanisms with knowledge graphs being at least 1-informed
%%%The mechanisms in Sections \ref{sec:warm:myerson} and \ref{sec:player-wise}
%%can be changed similarly.
%%except that the distributions are for different players rather than player-item pairs.
%We again summarize
%
%have the following theorem, whose proof has been omitted.
We again summarize our results in the theorem below.

\begin{theorem}\label{thm:bluff}
For any information elicitation
mechanism in previous sections where the knowledge graph is at least 1-informed,
the revised mechanism does not rely on the no-bluff assumption, and
%there exists a BIC crowdsourced mechanism $\cM'$ with the same expected revenue in the same auction settings,
%without making the no-bluff assumption.
%Moreover,
all players reporting their true values and true knowledge is the unique Bayesian Nash equilibrium.
Moreover, the mechanism's revenue is the same as before with an $\epsilon$ additive loss.
%, with the same expected revenue as before except
%, besides the unachieveable ones where each player reports the unknown true distributions.
\end{theorem}

%\paragraph{Remark.}
%By changing the mechanisms to extensive-form, let the players report their knowledge one by one and reward them
%according to
%the last report that is not ``$\bot$'',
%the same results hold under the unique subgame-perfect equilibrium.

When not everything is known and the players may bluff,
a player may not report ``$\bot$'' about a distribution he does not know,
because he may still get some reward in case nobody knows that distribution. It is an interesting open problem to design information elicitation mechanisms when not everything is known and without the no-bluff assumption.

\section{Information Elicitation Mechanisms with Efficient Communication}
\label{app:efficient}

%If communication complexity is a concern, the mechanisms in Sections \ref{sec:k=0} and \ref{sec:partial} are not efficient,
%especially when the distributions are continuous or if their supports are exponentially large.
%%It is easy to see that, if the distributions are continuous or if their supports are exponentially large in $m$ and $n$,
%%then it may not be feasible for the players to report the distributions in their entirety.
%In this section,
%When the players' value distributions have exponentially large supports, or when the distributions are
%continuous but the density functions do not have a succinct representation,
%having the players report the distributions in their entirety could be
%communicationally costly or infeasible.
To improve the communication complexity of our mechanisms, rather than
%
%we define two efficient ways for the seller to query the distributions
%and both of them still approximate the optimal revenue.
%
%Instead of letting
asking each player to report his known distributions in their entirety,
the seller can make specific queries to the players about the distributions.
Indeed, the {\em query complexity} of Bayesian auctions has been studied by \cite{chen2018bayesian} very recently,
where the seller does not know the prior distributions but is given oracle accesses to them.
More precisely, for any distribution $D$ over reals,
%could query the distribution via two types of single-value communication:
%{\em value queries} and {\em quantile queries}.
%Given a distribution $\cD$ over reals,
in a {\em value query}
the seller sends a value~$v$ and the oracle returns the corresponding quantile $q(v)=\Pr_{x\sim D} [x\geq v]$.
In a {\em quantile query}, the seller sends a quantile $q \in [0,1]$ and the oracle returns the corresponding value $v(q)$ such that $\Pr_{x\sim D} [x\geq v(q)] = q$.

In information elicitation auctions, as the players have knowledge about the distributions,
it is very natural for the seller to use them as oracles.
However, it is important to ensure that
the queries to the players
do not destroy their incentives
to be truthful: both to report their true values and to report their true knowledge.%
\footnote{Here a player reporting his true knowledge no longer means that
he reports the true distributions,
but that he answers the seller's queries truthfully.}
%
%a player is never hurt by
%truthfully answering queries about the other players' distributions,
%so that the mechanisms are still 2-DST.
%For example, if after player $i'$ reported $b_{i'j}$ as his value for $j$,
%the seller asks ...
Fortunately, truthfulness in our mechanisms can be easily guaranteed by {\em non-adaptive queries},
where all the queries are made together, before the players report their values.
As shown by \cite{chen2018bayesian},
when the players' value distributions are bounded within $[1, H]$ for a given value $H$,
the number of non-adaptive queries enough
to approximate $OPT$ in Bayesian auctions
is polynomial in $m$ and $n$, but only logarithmic in $H$, which is very efficient. Moreover,
only value queries are needed in this case.
When the distributions have unbounded supports but satisfy small-tail assumptions,
non-adaptive quantile queries are enough,
and the query complexity
is polynomial in $m, n$ and logarithmic in the cut-off value of the tail.
We make the same queries in our mechanisms as in \cite{chen2018bayesian}.

%In particular, when the seller wants to learn a distribution $D$ from a player,
%the seller can provide a value $v\in \bR$ to the player and let the player response the corresponding quantile $q=\Pr_{x\sim D}[x\geq v]$;
%or the seller can provide a quantile $q\in [0,1]$ to the player and let the player response the corresponding value such that $q=\Pr_{x\sim D}[x\geq v]$.

%If the queries are answered by an oracle,
%%if the player aways tells the truth (like an oracle),
%for any Bayesian instance $\hat{\cI}=(N,M,\cD)$ which is either single-good auction, unit-demand auction or additive auction,
%\cite{chen2017query} shows that with polynomial number of oracle queries to the distribution $\cD$,
%the seller can ensure an $(1+\epsilon)$-approximation
%to the revenue that the seller can generate with the exact distribution.
%Since all of the mechanisms in Sections \ref{sec:k=0} and \ref{sec:partial} are 2-DST,
%reporting the true knowledge does not hurt them.
Below we show how to revise our information elicitation mechanisms
to query the players,
%following the design of the queries in \cite{chen2017query},
using mechanism $\cM_{IEUD}$ as an example and for bounded distributions.
The mechanism now has a parameter $\epsilon>0$, which affects its approximation ratio.

\begin{itemize}
\item In Step \ref{step1}, given $\epsilon>0$, let $k = \lceil \log_{1+\epsilon}H\rceil$ and
$\nu=(\nu_{0}, \nu_{1},\dots, \nu_{k-1}, \nu_{k})=(1,(1+\epsilon), (1+\epsilon)^{2}, \dots, (1+\epsilon)^{k-1}, H).$

%For any distribution $\cD_{i'j}$ where $i'\in N_{2}$, let  and

Each player $i$ reports, for each player $i'\neq i$ and item $j$, either $\bot$ or
%Query each player $i\in N_{1}$ in the reporting group for $\cD_{i'j}$ with $v$, and receive
a non-increasing quantile vector
$q^i_{i'j} = (q^i_{i'j;0},\dots, q^i_{i'j;k})$, where $q^i_{i'j;0}=1$.
Allegedly, $q^i_{i'j;l}  = q_{i'j}(\nu_l)$ for each $l\in \{0, \dots, k\}$,
where $q_{i'j}(\cdot)$ is defined by $\cD_{i'j}$. That is, if $(i, i')\in G_j$ then
player $i$ answers the value queries for distribution $\cD_{i'j}$ and value vector $\nu$.

Simultaneously, each player $i$ also reports a valuation $b_i = (b_{ij})_{j\in M}$.

\item
In Step \ref{step5}, if player $i\in N_1$ is the reporter for player $i'\in N_2$ and item $j$, then
construct a discrete distribution $\cD'_{i'j}$ as follows:
$\cD'_{i'j}(\nu_{l}) = q^i_{i'j;l} - q^i_{i'j;l+1}$ for every $l\in \{0, \dots, k\}$, where $q^i_{i'j;k+1}\triangleq 0$.
\end{itemize}
The other parts of the mechanism remain unchanged.

In the revised mechanism, it is still dominant for the players to report their true values, no matter how the queries are answered. Indeed, the fact that distribution $\cD'_{ij}$ is now different from $\cD_{ij}$ does not affect the players' truthfulness in the Bayesian mechanism $\cM_{UD}$.
Moreover,
having player $i$ answer the value queries for $\cD_{i'j}$ is equivalent to
first having him report $\cD^i_{i'j}$ and then having the seller answer the value queries accordingly.
In the latter, reporting $\cD^i_{i'j}$ truthfully never hurts player $i$, because $i\in N_1$ when his knowledge is used.
Thus answering the value queries truthfully never hurts $i$ either, and the mechanism is still 2-DST.
Because the value queries for a distribution $\cD_{i'j}$ may be answered by all the other $n-1$ players (when they all know $\cD_{i'j}$), the query complexity and thus the communication complexity of our mechanisms have
an extra factor $n$ compared with the query complexity in \cite{chen2018bayesian}.

More precisely, we state the following theorem for arbitrary knowledge graphs and bounded distributions.
The proof is relatively easy following those for Section \ref{sec:k=0} and those in  \cite{chen2018bayesian}, thus
has been omitted.
%The results for Section \ref{sec:partial} are similar, thus omitted here.

\begin{theorem}
$\forall \epsilon>0, H>1$ and for any auction instances $\hat{\cI} = (N, M, \cD)$ and
$\cI = (N, M, \cD, G)$,
where each $\cD_{ij}$'s support is bounded in $[1,H]$,
our revised information elicitation mechanisms are 2-DST and make non-adaptive value queries.
Moreover,
\begin{itemize}
\item for single-good auctions,
         with $O(n^2\log_{1+\epsilon} H)$ queries,
	 the mechanism achieves revenue at least $\frac{OPT_K(\cI)}{4(1+\epsilon)}$;
\item for unit-demand auctions,
         with $O(mn^2\log_{1+\epsilon} H)$ queries,
         the mechanism achieves revenue at least $\frac{OPT_K(\cI)}{96(1+\epsilon)}$; and
\item for additive auctions,
        with $O(mn^2\log_{1+\epsilon} H)$ queries,
	the mechanism achieves revenue at least $\frac{OPT_K(\cI)}{70(1+\epsilon)}$.
\end{itemize}
\end{theorem}

The case of unbounded distributions with small-tail assumptions, as well as
the cases of $k$-informed knowledge graphs with $k\geq 1$ and bounded/unbounded distributions,
are similar. Indeed, the approximation ratios of our main results
and the query complexity of Bayesian auctions in \cite{chen2018bayesian} combine nicely with each other, resulting in
crowdsoruced Bayesian mechanisms with very efficient communication.

To further improve the communication complexity of our mechanisms,
the seller can change them into extensive-form mechanisms
and ask each player $i$ to first report a bit about each pair $(i',j)$, indicating whether $i$ knows $\cD_{i'j}$ or not.
The seller then selects one reporter and only asks him to answer the oracle queries. By doing so,
the extra factor $n$ in the query complexity of our mechanisms can be dropped, with the players communicating
 $O(n^2 m)$ bits besides the queries.

If scoring rules are used to buy the players' knowledge, then we can use the following {\em value scoring rule} $g_V$
to reward value queries, which follows directly from Brier's scoring rule~\cite{brier1950verification}.
More precisely, for any value query $v\in \bR$, letting $x$ be a sample from the underlying value distribution and~$q$ be the answer of a reporter to the query, then
$$g_V(x,q; v) \triangleq 1 + 2q {\bf I}_{x \leq v} - q^2.$$
To reward the players' answers to quantile queries, we define the following {\em quantile scoring rule} $g_Q$, which is a variant of the one in  \cite{cervera1996proper}.
More precisely, for any quantile query $q\in [0, 1]$, letting $x$ be a sample from the underlying value distribution and $z$ be the answer of a reporter to the query, then
$$g_{Q}(x,z; q) \triangleq q\arctan z - (\arctan z - \arctan x){\bf I}_{z\geq x}.$$
Both scoring rules are strictly proper scoring rules with bounded ranges.

\section{Aggregating the Players' Refined Insider Knowledge}\label{app:refine}
As mentioned in the introduction of the paper,
since the common prior assumption implies that every player has
correct and exact (that is, no more, no less) knowledge
about all the distributions,
in the main body of this paper we do not consider scenarios where the players have ``insider'' knowledge.
% and can
%refine the prior known to them.
Incorrect insider knowledge has been studied in
\cite{bergemann2012robust, artemov2013robust, chen2013mechanism, chen2015tight, bergemann2015informational}
and is not the focus of this paper.
However, sometimes
each player may have {\em correct} insider knowledge
about the other players' value
distributions%
%, and such knowledge is actually
%correct%
 \footnote{A similar scenario in the literature of contracts was considered in \cite{carroll2016robust}
with different concerns.}:
 that is, his knowledge is a refinement of the prior.
%Each player's extra knowledge naturally depends  on the other players' private values.

Different players' knowledge, although all correct, may refine the prior in different
ways.
For example, when the prior distribution of a player $i$'s value for an item $j$ is
 uniform over $[0, 100]$, after $v_{ij}$ is drawn,
another player $i'$ may observe whether
$v_{ij} \geq 50$ or not,
and a third player $i''$ may observe whether $v_{ij}\in [20, 80]$.
Thus, player $i'$ knows whether $v_{ij}$ is uniform over $[0, 50]$ or $(50, 100]$,
depending on
his signal; and player $i''$ knows whether $v_{ij}$ is uniform over $[20, 80]$ or $[0, 20)\cup(80, 100]$, depending on the signal $i''$ observes.

The players' correct knowledge must be consistent with each other and
one can obtain an even better refinement of the prior
by combining their knowledge together.
In the example above, if player $i'$ observes $v_{ij}< 50$ and player $i''$ observes $v_{ij}\notin [20, 80]$,
then it must be that $v_{ij}\in [0, 20)$ and the posterior distribution is uniform in this range.
However, neither $i'$ nor $i''$ knows this fact.

\paragraph{Enhanced knowledge graphs.}
To model the players' insider knowledge,
we equip the knowledge graphs in the information elicitation setting
with {\em information sets}.
To begin with, for any two players $i, i'$ and item $j$ such that $(i, i')\in G_j$,
there is a partition $\cP^{i}_{i'j}$ of the support of $\cD_{i'j}$, representing the possible signals player $i$ will observe about $\cD_{i'j}$.
After the true value $v_{i'j}$ is drawn, letting $S(v_{i'j})$ be the unique set
in the partition that contains $v_{i'j}$,
player $i$ learns the fact that player~$i'$'s true value for $j$ falls
into $S(v_{i'j})$ and the posterior distribution is $\cD_{i'j}|S(v_{i'j})$.
More generally, the partitions may  even depend on player $i$'s own true valuation $v_i$,
because $v_i$ is part of the information he has.
Because different values are independently drawn,
given
$v_i$ and
the information sets observed by $i$ for different distributions of
the other players,
$i$ considers their posterior distributions to be independent.

All our mechanisms remain 2-DST with respect to the players' refined knowledge,
%
%It is not hard to see that for unit-demand auctions with arbitrary knowledge graphs,
%when the players have refined knowledge
%our mechanisms remain 2-DST,
where a player's true knowledge is now the posterior distributions known by him.
Since the optimal Bayesian revenue increases
when the distributions are refined \cite{azar2012crowdsourced},
the expected revenue of our mechanisms also increases,
where the expectation is further taken over the private signals observed by the players.
%
%For additive auctions the same holds when everything is known,
%and a more careful analysis will be needed to understand how the revenue of our mechanism changes
%with respect to the players' refined knowledge when not everything is known.
%For single-good auctions when the knowledge-graph is 2-connected,
%one needs to compute a player's virtual value
%only when his distribution is reported for the first time, otherwise
%the players may not report their true knowledge.
%%Still, it remains to be understood how the revenue of the corresponding mechanism changes.
However, the revenue benchmark is still defined as before: that is, with respect to the prior and the knowledge graphs, without considering the refinements.
%This is also the case in \cite{azar2012crowdsourced}, where each player may privately refine the prior.
%Finally, although our mechanisms and the ones in
%allow each player to privately refine the prior,
%the revenue benchmarks are still defined with respect to the prior (and the knowledge graphs),
%without considering the refinements.
An interesting open problem is to design information elicitation mechanisms
whose revenue approximates a more demanding benchmark ---
the optimal revenue based on
the ``aggregated refinement'' obtained by combining all players' refinements together.

\section{Information Elicitation Mechanisms for Combinatorial Auctions}
\label{sec:player-wise}

%\vspace{-5pt}

In our main results, the knowledge graphs for different items can be totally different from each other:
player 1's value distribution for item 2 may be known by player 3,
while his value distribution for item 4 may be known by player 5, etc.
% may know player 2's value distribution for item 3, while .
When all the knowledge graphs are the same, we say that the auction has {\em player-wise information}:
all value distributions of a player have the same ``knower''.

%
%the problem becomes a special case which is also interesting:
%that is, a player knows another player's value distributions for all items.
%We denote this setting as , in contrast with previous settings which can be considered as {\em (player, item)-wise information}.

With player-wise information, only one knowledge graph $G$ is needed: an edge $(i, i')$ is in $G$ if and only if player $i$
knows the distribution $\cD_{i'} = \cD_{i'1}\times\cdots\times\cD_{i'm}$.
Such an information setting can also model arbitrary {\em combinatorial auctions}, where
%a player's valuation function may not be defined by his item values.
%combinatorial auctions.
%Recall that i
%In such an auction,
each player $i$'s valuation function $v_i$
maps each subset of items to a non-negative real, with $v_i(\emptyset) = 0$.
Thus the distribution $\cD_i$ is over such functions, and
$i$'s values for two subsets of items can be arbitrarily correlated.
%% in crowdsourced Bayesian settings with player-wise information.
%Here a Bayesian instance is denoted by $\hat{\cI} = (N, M, \cD = \times_{i\in N} \cD_i)$, where $\cD_i$ is the distribution of $i$'s valuation function, and $i$'s values for different items or subsets of items can be arbitrarily correlated.
%In particular, our results here apply to all unit-demand auctions and additive auctions.%
%\footnote{Note that in unit-demand auctions and additive auctions studied in the literature and in previous sections, a player $i$'s values for different subsets of items are correlated, but his values for different items are independent.}
%, where player has arbitrarily correlated values towards different bundles of items.
Given a combinatorial Bayesian auction instance $\hat{\cI} = (N, M, \cD = \times_{i\in N} \cD_i)$,
a corresponding information elicitation instance is denoted by $\cI= (N, M, \cD, G)$,
where $G$ is a single knowledge graph rather than a vector of graphs.
%Optimal $\epsilon$-BIC and BIC mechanisms have been designed in \cite{cai2012algorithmic, cai2012optimal} for combinatorial auctions, and
%little is known in the literature about how to construct (approximately) optimal DST Bayesian mechanisms here.

Player-wise information is a much stronger assumption than {\em (player, item)-wise information}, the information settings considered in the main body of this paper.
We mention this model here mainly for completeness and to facilitate the comparison of our results with the literature.
%we show that
%: an edge $(i, i')\in G$ means that player $i$ knows $\cD_{i'}$.
Indeed, with player-wise information and
%when the knowledge graph is at least 1-informed,
by randomly partitioning the players
into reporters and potential buyers as we have done in mechanisms $\cM_{IEUD}$ and $\cM'_{IEUD}$,
we have a simple
blackbox reduction from Bayesian mechanisms to information elicitation mechanisms.
Our mechanism $\cM_{IEB}$ is shown
%as shown in mechanism
%where the latter can be obtained almost for free from the former.
%$\cM_{IEB}$ below.
in Mechanism \ref{alg:general},
where $\cM_B$ can be any Bayesian mechanism.
% $\cM_B$ as a black-box.
We have the following theorem, whose proof is similar to that of Theorem \ref{thm:unit-k} and thus omitted.

\vspace{-5pt}

\begin{algorithm}
\floatname{algorithm}{Mechanism}
  \caption{\hspace{-4pt}  $\cM_{IEB}$}
  \label{alg:general}
  \begin{algorithmic}[1]

  \STATE Each player $i$ reports a valuation function $b_i$ and a knowledge $K_i = (\cD^i_{i'})_{i'\neq i}$.

  \STATE Randomly partition the players into two sets, $N_1$ and $N_2$, where each player is independently put in $N_1$ with probability $q=1-(k+1)^{-\frac{1}{k}}$ and $N_2$ with probability $1-q$.

  \STATE Let $N_3$ be the set of players in $N_2$ whose distributions are reported by some players in $N_1$, and let $\cD'_{N_3}$ be the vector of reported distributions.

%  \STATE Reward players in $N_1$ using Brier's scoring rule, with total reward $R\leq \epsilon$.

  \STATE Run $\cM_{B}$ on the Bayesian instance $\hat{\cI}_{N_3} = (N_3, M, \cD'_{N_3})$ and the valuation functions $b_{N_3}$;
  and use the resulting allocation and prices to sell to players in $N_3$.

\end{algorithmic}
\end{algorithm}

\begin{theorem}\label{thm:fullinfo}
For any $k\in [n-1]$, for any combinatorial auction instances $\hat{\cI} = (N, M, \cD)$ and $\cI = (N, M, \cD, G)$ where $\cI$ has player-wise information
and $G$ is $k$-informed,
if $\cM_B$ is DST then $\cM_{IEB}$ is 2-DST; and if $\cM_B$ is BIC then $\cM_{IEB}$ is BIC. Moreover, if $\cM_B$ is a $\sigma$-approximation to $OPT$, then $\cM_{IEB}$ is a $\tau_k \sigma$-approximation to $OPT$.
%$\bE_{v\sim \cD} Rev(\cM_{IEB}(\cI)) \geq \tau_k \bE_{v\sim \cD} Rev(\cM_B(\hat{\cI})) - \epsilon$.
\end{theorem}
With player-wise information,
a player $i$'s valuation distribution $\cD_i$
is either ``completely known'' or
``completely unknown''.
Thus for unit-demand auctions there is no need to
adopt the COPIES setting to handle the scenario where only part of $\cD_i$
is reported.
As before, the approximation ratio of $\cM_{IEB}$ increases
as $k$ gets larger
and
converges to that of the Bayesian mechanism.
When $k=1$, it is a 4-approximation to $OPT$
using
the optimal Bayesian mechanism as a black-box.
Note that the model studied in \cite{azar2012crowdsourced} is a very special case even compared with the player-wise information setting:
that is, $G$ is the complete graph and $k=n-1$.
Since $\tau_{n-1} = \frac{n-1}{n^{n/(n-1)}} \rightarrow 1-\frac{1}{n}$ when $n$ gets larger,
the revenue of our mechanism essentially matches that of \cite{azar2012crowdsourced} for combinatorial auctions.
Moreover, if the players can observe private signals and refine their knowledge about the prior,
then our mechanism can also aggregate such refinements as in \cite{azar2012crowdsourced}.
Finally, we have the following corollary for arbitrary knowledge graphs (i.e., $k$ may be 0).
%by setting $q=\frac{1}{2}$ in $\cM_{IEB}$,

%we can approximate the knowledge-based revenue benchmark $OPT_K$.
%consider the Bayesian instance projected to the knowledge graph,
%and benchmark against the optimal revenue there. The following corollary holds
%simply by setting $q=\frac{1}{2}$ in $\cM_{IEB}$.
%
%benchmark which is the optimal
%revenue on the Bayesian instance projected to the knowledge graph.
%

\begin{corollary}
For any combinatorial auction instances $\hat{\cI} = (N, M, \cD)$ and $\cI = (N, M, \cD, G)$ with player-wise information,
%and $\cI' = (N, M, \cD')$
%where $\cD'$ is $\cD$ projected on $G$, for any Bayesian mechanism $\cM_B$,
mechanism $\cM_{IEB}$ with $q=\frac{1}{2}$ is a $\frac{\sigma}{4}$-approximation to $OPT_K(\cI)$
%such that:
% if $\cM_B$ is DST then $\cM_{IEB}$ is 2-DST; if $\cM_B$ is BIC then $\cM_{IEB}$ is BIC.
% Moreover,
if $\cM_B$ is a $\sigma$-approximation to $OPT$.
% then $\cM_{IEB}(\cI)$ is.
\end{corollary}

\newpage

\bibliographystyle{abbrv}% the recommended bibstyle
\bibliography{ref}

\end{document}